%% file: main.tex
\documentclass[journal]{IEEEtran}
\newlength{\flexwidth}
\usepackage[top=0.75in, left=0.625in, right=0.625in, bottom=1in]{geometry}
\usepackage{silence}
\usepackage[table]{xcolor}
\usepackage{cite} 
\usepackage{amsmath,amssymb,amsthm,fixmath}
\usepackage{mathtools}
\usepackage{accents}
\usepackage{xcolor}
\usepackage{siunitx}
\usepackage{cuted}
\usepackage{multirow,booktabs}
\usepackage{subfigure}
\usepackage{wrapfig}
\usepackage[inline]{enumitem}
\usepackage{optidef}
\usepackage{graphicx}
\usepackage{epstopdf}
\usepackage{lipsum}
\usepackage{amssymb}
\usepackage{diagbox}
\usepackage{adjustbox}
\usepackage{tcolorbox}
\usepackage{float}
\newtheorem{hypothesis}{Hypothesis}
\usepackage[normalem]{ulem}
\newcommand{\reviseprev}[2]{#2} 
\newcommand{\revise}[2]{{\textcolor{red}{#1}}{\textcolor{blue}{#2}}}
\renewcommand{\revise}[2]{#2} 
\newcommand{\revisebox}[1]{\fcolorbox{blue}{white}{#1}}
\newcommand{\revisefinal}[2]{{\textcolor{red}{#1}}{\textcolor{blue}{#2}}}
\renewcommand{\revisefinal}[2]{#2} 
\renewcommand{\revisebox}[1]{#1} 
\newcommand{\reviseboxfinal}[1]{\fcolorbox{blue}{white}{#1}}
\renewcommand{\reviseboxfinal}[1]{#1}
\usepackage{dblfloatfix}
\pdfoutput=1
\usepackage[linesnumbered,ruled,vlined]{algorithm2e}

\usepackage{graphicx}
\usepackage[hidelinks]{hyperref}

\usepackage{epstopdf}
\usepackage[textsize=tiny,colorinlistoftodos]{todonotes}
\makeatletter
\define@key{todonotes}{bh}[]{
	\setkeys{todonotes}{author=\textbf{Bin - notes}, color=red!30}}%
\define@key{todonotes}{bh2}[]{
		\setkeys{todonotes}{author=\textbf{Bin - urges}, color=green!30}}%
\define@key{todonotes}{zf}[]{
	\setkeys{todonotes}{author=\textbf{Zexin}, color=blue!30}}%
\makeatother

\usepackage{algpseudocode}

\usepackage[shortcuts,acronym,automake]{glossaries}
\input{glossary}

\DeclareSIUnit{\belmilliwatt}{Bm}
\DeclareSIUnit{\dBm}{\deci\belmilliwatt}

\usepackage{tcolorbox}
\tcbuselibrary{many}

\newtheorem{theorem}{Theorem}

\usepackage[free-standing-units=true]{siunitx}

\begin{document}

	\title{Trustworthy UAV Cooperative Localization: \reviseprev{On}{} Information Analysis of Performance and Security \reviseprev{Concerns}{}}
	
	
	 \author{
	 	Zexin~Fang,~\IEEEmembership{Student Member,~IEEE,}
		Bin~Han,~\IEEEmembership{Senior Member,~IEEE,}
		and
		Hans~D.~Schotten,~\IEEEmembership{Member,~IEEE}
		\thanks{
    Copyright~\copyright~2025 IEEE. Personal use of this material is permitted. However, permission to use this material for any other purposes must be obtained from the IEEE by sending a request to pubs-permissions@ieee.org.
    
    Z. Fang, B. Han, and H. D. Schotten are with University of Kaiserslautern-Landau (RPTU), Germany. H. D. Schotten is with the German Research Center for Artificial Intelligence (DFKI), Germany. B. Han (bin.han@rptu.de) is the corresponding author. This work is supported by the German Federal Ministry of Education and Research within the project Open6GHub (16KISK003K/16KISK004). }
	}
	
	\bstctlcite{IEEEexample:BSTcontrol}
	
	\maketitle

\begin{abstract}
  This paper presents a trustworthy framework for achieving accurate cooperative localization in multiple \gls{uav} systems. The \gls{crlb} for the \gls{3d} cooperative localization network is derived, with particular attention given to practical scenarios involving non-uniform spatial distribution of anchor nodes. Challenges of mobility are then addressed with \gls{magd}.

  In the context of system security, we derive the \gls{crlb} of localization under the influence of falsified information. The methods and strategies of injecting such information and their impact on system performance are studied. To assure robust performance \reviseprev{while}{} under falsified data, we propose \reviseprev{mitigate solutions:}{a mitigation solution named} \gls{tad}. Furthermore, we model \reviseprev{the relationship between system performance and the density and magnitude of falsified information}{the system performance regarding the density and magnitude of falsified information}, focusing on realistic scenarios where the adversary is resource-constrained. With the vulnerability of cooperative localization understood, we \reviseprev{applied}{apply} \gls{tad} and \reviseprev{folumated}{formulate} an optimization problem from the adversary's perspective. Next, we discuss the design principles of an anomaly-detector, with emphasis of the trade-off of reducing such optimum and system performance. Additionally, we also deploy a \gls{rp} mechanism to fully utilize the anomaly detection and further optimize the \gls{tad}. Our \reviseprev{methodology is}{proposed approaches are} demonstrated through numerical simulations.
\end{abstract}
	
	
	\IEEEpeerreviewmaketitle
	
	\glsresetall
	
	\section{Introduction}\label{sec:introduction}
    Multi-\gls{uav} systems hold significant promise for revolutionizing various domains, particularly the future \reviseprev{Sixth Generation (6G)}{\gls{b5g} and \gls{6g} mobile networks}. It is notably simpler to deploy \gls{uav}-mounted \acp{bs} and relay nodes than traditional cellular \acp{bs}, providing a cost-efficient and reliable solution to fulfill the requirements of \gls{b5g}/\gls{6g} by tackling coverage gaps and intelligent resource allocation. For instance, multi-\gls{uav} systems can provide reliable communication links in challenging environments, support connectivity in remote or disaster areas, and overcome limits of ground infrastructure \cite{FengSCE}. 
   
    \reviseprev{\gls{gps} modules may face challenges in providing \acp{uav} precise position information}{As the most popular solution to provide \acp{uav} precise position information, \ac{gps} can be challenged} to support such use cases, especially in in urban areas, tunnels, or indoor environment\reviseprev{ for \gls{uav} to support such functionalities}{}. Alternative options like radio trilateration\reviseprev{ can be used but have limited coverage and require calibration and installation cost}{, on the other hand, are commonly inferior in coverage while taking extra cost for installation and calibration} \cite{5Gposition,Radiotri}. \revise{}{Some cutting-edge research are focusing on \ac{3d} space localization \cite{RIS3d,siglenode3d,5sensors}, offering complementary insights. Yet \reviseprev{there are}{} a few drawbacks \reviseprev{that}{} shall be noted: \begin{enumerate*}[label=\emph{\roman*)}]
    	\item Single-node localization methods can be highly unreliable under \reviseprev{challenging}{harsh} channel conditions. 
    	\item To measure angles on an \gls{uav}, it typically requires either an antenna array or rotating the \gls{uav}, introducing extra weight and energy consumption\reviseprev{to \gls{uav}}{}. 
    	\item Moreover, sophisticated algorithms are required for accurate angle estimation. 
    \end{enumerate*}
    Overall, the low reliability and high cost associated with angle measurement are constraining these methods from practical applications.}
    
    \reviseprev{As a consequence for that}{As a pillar \cite{3gpp_uas_support}}, the \gls{3gpp} \reviseprev{defined}{specifies in its Release 17}
    \reviseprev{}{ as a requirement of \gls{uav} applications,} that \reviseprev{}{the} network \reviseprev{should}{shall} have the capability to facilitate network-based \gls{3d} space positioning\reviseprev{ and}{, as well as} \gls{c2} communication\reviseprev{, as specified it as requirements of \gls{uav} applications}{}.  \revise{ }{Meanwhile, \gls{3gpp} Release 17 require a \gls{uav} to broadcast its position information within a short-range area (at least $50\metre$) for collision avoidance. A cooperative localization network utilizing broadcast position information and distance measurements can provide accurate estimates for targets in a \gls{gps}-denied environment.}
    
    \revise{}{Though cooperative localization based on mobile anchors has attracted recent research attention \cite{cooploc2024Hu, Cooploc2024jin, CP2022minetto, huang2015distributed}, its security aspects remain critical and understudied, as location information can be easily manipulated.
    While security issues in static sensor networks have been extensively studied, previous works \cite{Ragmalicous, Jhagame, DOCwon2019} focus primarily on comparing anomaly detection algorithms based on performance or computational cost, yet fail to examine how detection errors further reduce the pool of reliable information, compounding overall performance. In \cite{Tomic2022DetectingDA, scp2021beko, LighadXie2021}, the authors provide a comprehensive analysis of detector performance, offering valuable insights into the interplay between detection error and localization capabilities. However, they omit information-theoretic analysis of falsified data's impact and only examine distance-spoofing attacks, whereas in our multi-UAV scenarios, more sophisticated attack strategies should be considered. On the one hand, \cite{Ragmalicous, Jhagame, DOCwon2019, Tomic2022DetectingDA, scp2021beko} neglects position errors while focusing solely on distance estimation errors in static networks. In \gls{uav} scenarios, position errors are unavoidable, requiring sophisticated modeling. Additionally, \glspl{uav} mobility creates varying attack densities where static defenses may fail during concentrated attacks. This necessitates a resilient adaptive detection strategy leveraging historical data.}

    \revise{}{Inspired by prior research, our work delivers a comprehensive \gls{crlb}-based performance analysis while establishing a secure cooperative localization framework.} 
    We outline the novelties of our work: 
     \begin{itemize}
     \item In contrast to the majority of existing research that solely rely on a Gaussian model to approximate all error sources, we separately consider measurement errors in self-positioning and distance estimation. The power of both errors are randomized to closely emulate real localization modules and channel environment. 
     \item  We derive the \gls{3d} \gls{crlb} for localization networks, with geometric interpretations and considerations for non-uniform anchor node distributions. After evaluating traditional localization methods in such scenarios, we introduce the \gls{magd} algorithm to enhance robust and accurate cooperative localization in dynamic environments with mobile target and anchor nodes.
     \item We design various attack modes to inject falsified information into the system. This information can be biased or contain large errors. We derive the effectiveness of these attack modes and subsequently validated them through simulation. The interaction between the attacker and the detector is formulated as an optimization problem.
     \end{itemize}
     
	Especially, compared to its preliminary conference version~\cite{FHS2023reliable}, this article \begin{enumerate*}[label=\emph{\roman*)}]
	\item further deepens our understanding to localization networks with the bound analysis,
	\item enriches our knowledge of conventional solutions through \reviseprev{the}{} comprehensive benchmarks, and
	\item \reviseprev{}{delivers} performance and security measures from an information analysis perspective. 
\end{enumerate*}
     
     The remainder of this paper is organized as follows. In Sec.~\ref{Errmodel} and Sec.~\ref{sec:crlb} we investigate the error model and \gls{crlb} for the cooperative localization network.
     In Sec.~\ref{sec:MAGD}, we demonstrate the performance of localization algorithm in the presence of non-uniformly distributed anchor \gls{uav}s. We also assess the effectiveness of our proposed \gls{magd} through numerical simulations. Then in Sec.~\ref{attack} we outline the potential attack scheme and propose defend scheme, and validate our proposed scheme with numerical simulations presented in Sec.~\ref{sec:SIMUTDAD}. The paper concludes by summarizing \reviseprev{the}{our} main findings in \reviseprev{Section}{Sec.}~\ref{conclu}.
    \section{Related works}
    \revise{}{Existing research on the localization network predominantly \reviseprev{focuses}{focus} on terrestrial scenarios within static sensor networks, frequently overlooking altitude considerations \cite{sensor2d,staticsensor1,staticsensor2}. A more detailed investigation is \reviseprev{required}{widely suggested} for 3D space localization involving airborne applications\reviseprev{, as many researchers have suggested}~\cite{drones6020028,3dspacelocal}. Diverging from terrestrial scenarios, \gls{3d} environments pose challenges for localization frameworks that \reviseprev{have}{are} proven effective in \gls{2d} settings \cite{Kumari2019}. \cite{3duavloc2013vill} demonstrates the application of \gls{3d} sensor networks for distance-based \gls{uav} localization. However, this work does not fully address the implications of non-uniform anchor node distribution in \gls{3d} space, with particular emphasis on the latitude distribution. To the best of our knowledge, this aspect has only been investigated in our previous work \cite{FHS2023reliable}. Similarly, \emph{Zheng} et al. propose in \cite{5sensors} a \gls{3d} sensor network localization method for UAVs that integrates computationally expensive \gls{aoa} and \gls{tdoa} techniques. Regarding network-free \gls{3d} space localization,
    \reviseprev{He}{\emph{He}} et al. \reviseprev{proposed}{propose} a \gls{3d} localization framework~\cite{RIS3d}, which incorporates \gls{ris} and \gls{aoa}, demonstrating centimeter-level accuracy for targets positioned meters away from a mobile station. However, the \reviseprev{author notes}{authors note} that including \gls{nlos} path \reviseprev{could}{can} sharply increase the computational complexity of the framework. Meanwhile, an alternative \reviseprev{methodology}{method} presented in \cite{siglenode3d} focuses on achieving \gls{3d} \gls{uav}\reviseprev{s}{} localization with one mobile node and \gls{doa} technique. } 
    \revise{}{From the perspective of anomaly detection in localization networks, \cite{MLXu2021} developed a machine learning-based anomaly detection method for cooperative underwater localization, though its pre-training and online learning requirements introduce significant computational overhead. Meanwhile, \cite{MSK2021rss} proposed \gls{wls} and \gls{ln-1} techniques that effectively secure localization networks against both uncoordinated and coordinated malicious anchor attacks. \emph{Jha}, \emph{Tomic}, and \emph{Beko} et al. developed several \gls{doc}-based approaches \cite{DOCwon2019, Tomic2022DetectingDA, scp2021beko}, each employing distinct methods for anomaly exclusion. Specifically, \cite{DOCwon2019} implements an iterative exclusion process for anomalies, whereas \cite{Tomic2022DetectingDA} and \cite{scp2021beko} rely on the \gls{glrt} methodology. Similar to our approach, \cite{LighadXie2021} proposed a lightweight secure localization method that evaluates \gls{tdoa} results by comparing the receiver noise variance with the estimated noise variance, offering an efficient alternative to more computationally intensive techniques. }
   

	\section{Error model setup}\label{Errmodel}
    \subsection{Distance estimation error}\label{disERR}
    Measuring \gls{tof}, typically demands precise clock resolution, which is crucial for accurate and reliable performance. Its accuracy can be centimeter level within a valid ranging range of several meters. However, ensuring and sustaining synchronization between devices can be challenging, particularly in complex settings with multiple \gls{tof} devices and fluctuating ambient conditions. Compared to \gls{tof} ranging, \gls{rssi} ranging doesn't require synchronization or specific sensors, as long \gls{uav} is equipped with a radio receiver. 
    Due to \gls{rf} signal propagation effects, \gls{rssi} faces challenges, making precise distance estimates hard to obtain \cite{rssidisad}. Despite this limitation, its implementation in \gls{uav} scenarios is still straightforward and cost-efficient, given its reliance solely on signal strength information. \gls{rssi} ranging relies on a path loss model:
    \begin{equation}\label{eq:pathloss}
		P_\mathrm{r}\textit{(d)}= P_\mathrm{r}(\textit{d}_0) - 10n_\mathbf{p} \mathrm{log}(\frac{\textit{d}}{\textit{d}_0}) + P_\mathrm{r}^{\Delta}(\textit{d}),
	\end{equation}
    where $P_\mathrm{r}\textit{(d)}$ indicates the \gls{rssi} at the distance $\textit{d}$ from the anchor \gls{uav}; $\textit{d}_0$ is the predefined reference distance, meanwhile $P_\mathrm{r}(\textit{d})$ is the \gls{rssi} measurement value at $\textit{d}_0$; $n_p$ denotes the path loss factor of the radio link.
    
    Due to the channel fading, \gls{rssi} measurement error $P_\mathrm{r}^{\Delta}(\textit{d})$ is inevitable and results in a distance estimation error $\varepsilon(d)$. The measurement error $P_\mathrm{r}^{\Delta}(\textit{d})$ follows a zero-mean Gaussian distribution with standard deviation $\sigma_\mathrm{r}$. We examined \gls{rssi} measurement results between two Zigbee nodes \cite{shortrangenovel, shortrangeoutdoor} and Sigfox nodes \cite{longseparation}. The results reveal that $\sigma_\mathrm{r}$ shows a seemingly random pattern with minimal correlation with distance, suggesting fading dominates over path loss in urban or indoor environments. We consider both outdoor urban environments and large indoor factories, where \gls{nlos} radio links are prevalent. Consequently, $P_\mathrm{r}^{\Delta}(\textit{d})$ is expected to behave similarly. To simplify our analysis, we consider $P_\mathrm{r}^{\Delta }{\mathrm{(\textit{d},t)}}\sim\mathcal{N}\left(0,\sigma_\mathrm{r}^2(t)\right)$ and $\sigma_\mathrm{r}(t)\sim\mathcal{U}\left(\sigma^{\mathrm{min}}_r,\sigma^{\mathrm{max}}_r\right)$.

  
    Specifying the path loss model after Tab.~\ref{tab:pathpara}, the simulation results are illustrated in Fig.~\ref{fig:RSSImean2}. The inaccurate \gls{rssi} estimation leads to a distance estimation error $\varepsilon(\textit{d})$, as shown in Fig.~\ref{fig:RSSI100}, which is zero-mean Gaussian distributed with $\varepsilon(\textit{d})\sim\mathcal{N}(0,\sigma_\textit{d}^2(\textit{d}))$, while $\sigma_\textit{d}(\textit{d})$ fluctuates but increases over the distance. It becomes evident that distance estimation for large distances can be highly unreliable. \revise{}{The communication cost increases significantly with additional exposure of anchor \glspl{uav} to the target \gls{uav} due to the maintenance of multiple communication links, which often demands higher energy consumption and increased data transmission rates. Furthermore, anchor \glspl{uav} must manage interference across multiple links and ensure stable connectivity, adding to the overall overhead.}  
    \begin{table}[ht]
       \centering\caption{Pass loss model parameters}\label{tab:pathpara}
       \begin{tabular}{|c|c|c|c|}
         \hline
         $n_\mathbf{p}$ & $\textit{d}_0$ & $P_\mathrm{r}(\textit{d}_0)$ & $[\sigma_{\mathrm{min}}^r,\sigma_{\mathrm{max}}^r$]\\
         \hline
         3.0 & $1.0\metre$ & $-30.0 \si{\dBm}$& [0.5,2] \\
         \hline
    \end{tabular}
    \vspace{-3mm}
    \end{table}
    \begin{figure}[!htbp]
		  \centering
      \includegraphics[width=0.8\linewidth]{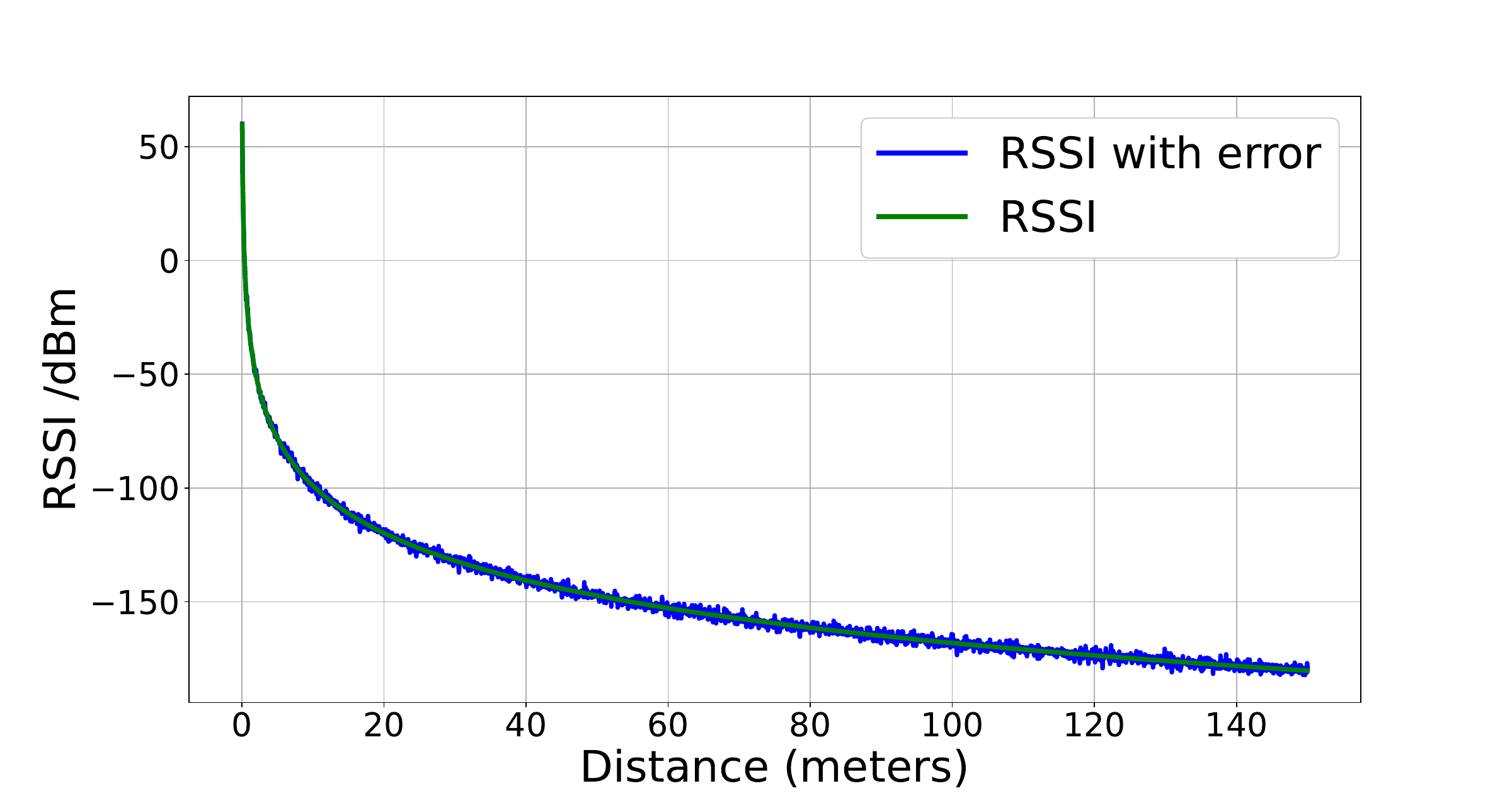}
      \vspace{-3mm}
		  \caption{Simulated \gls{rssi} over distance}
		  \label{fig:RSSImean2}
	  \end{figure}
    \begin{figure}[!htbp]
		\centering
		\includegraphics[width=0.85\linewidth]{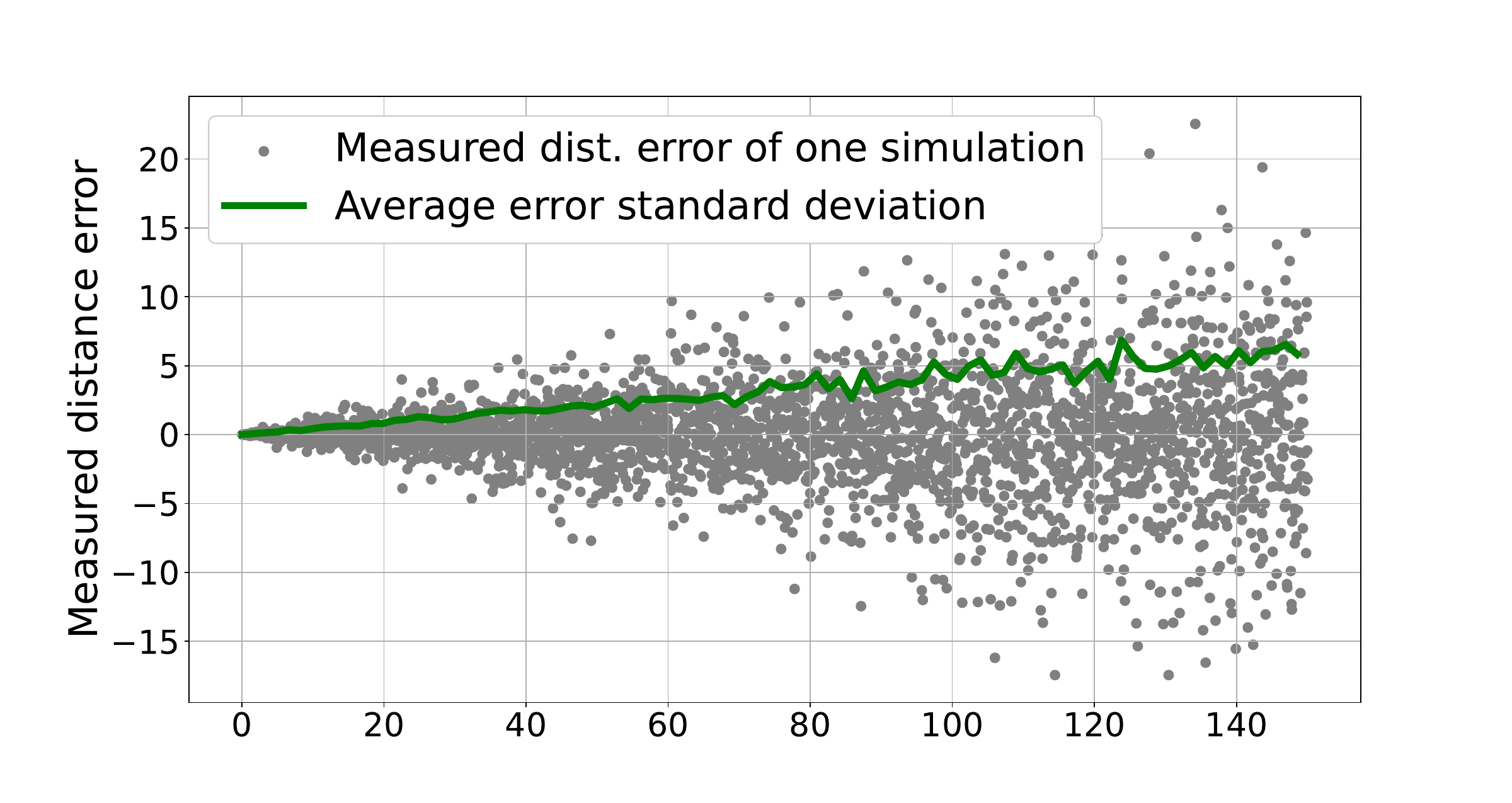}
		\vspace{-3mm}
    \caption{Distance estimation error (meters) over spacing between two \acp{uav}}
		\label{fig:RSSI100}
	\end{figure}
    
   \subsection{Position measurement error}
    With the distance error model addressed, we can examine the position error. Altitude readings from a GPS module are often highly unreliable. The altitude readings from a \gls{gps} module can be highly unreliable. As a result, alternative methods are frequently employed in aircraft for altitude estimation, with technologies such as \gls{arhs} and altimeters being common choices. In practice, a comprehensive solution emerges through the integration of data from the \gls{gps} module and altitude estimation methods. 
   
    Assuming such an integrated solution is deployed in \gls{uav}s, we consider a Gaussian localization error with uniformly distributed power across dimensions. For an anchor \gls{uav} set $\mathcal{U}=\{u_0,u_1,u_2...u_n\}$ at time $t$, the actual and pseudo positions of the $n_\mathrm{th}$ \gls{uav} are modeled as $\mathbf{p}_n(t)$ and $\mathbf{p^{\circ}}_n(t)$, respectively:
   \begin{equation*}
        \mathbf{p}_n(t) = [x_n(t),y_n(t),z_n(t)],
    \end{equation*}
    \begin{equation*}\label{eq:PosErr1}
        \mathbf{p^{\circ}}_n(t) = \mathbf{p}_n(t) + [x_n^{\Delta}(t), y_n^{\Delta}(t),z_n^{\Delta}(t)],
    \end{equation*}    
   \begin{equation*}\label{eq:PosErr2}
		[x_n^{\Delta}(t), y_n^{\Delta}(t),z_n^{\Delta}(t)]\sim\mathcal{N}^2\left(0,\sigma_{\mathbf{p},n}^2/3\right),
   \end{equation*}
   \begin{equation*}\label{eq:PosErr3}     \sigma_{\mathbf{p},n}\sim\mathcal{U}\left(\sigma^{\mathrm{min}}_\mathbf{p},\sigma^{\mathrm{max}}_\mathbf{p}\right),   
   \end{equation*}
   where $[x_n^{\Delta}(t), y_n^{\Delta}(t),z_n^{\Delta}(t)]$ is the position error with power $\sigma^2_{\mathbf{p},n}$. Within a valid coverage of cooperative localization, a target \gls{uav} $u_k$ measures its distance to anchor \gls{uav}s. We model the real distance $d_{k,n}=\left\Vert\mathbf{p}_k(t) - \mathbf{p}_n(t)\right\Vert$ and the measured distance $d^{\circ}_{k,n}={d}_{k,n} + {\varepsilon({d}_{k,n})}$.
   \section{Bound Analysis of Localization Networks}{\label{sec:crlb}}
  
   \subsection{Error modelling}\label{subsec:ERRMODEL}
   To achieve a comprehensive and accurate simulation of \gls{uav} localization, we factor in both position errors and distance estimation errors. 
   A \gls{mle} to estimate the target position $\mathbf{p}_k$ can be described:
   \begin{equation}\label{eq:optimize}
   \begin{split}
         {\mathbf{p}}_k &= \arg\min\limits_{[x,y,z]}\sum_{n\neq k}\left\vert\Vert\mathbf{p^{\circ}}_n-{\mathbf{p}}_k\Vert - {d}^{\circ}_{k,n}\right\vert.\\
   \end{split}     
   \end{equation}

  A zero-mean Gaussian variable $X$ with uniformly distributed standard deviation $\sigma\in\mathcal{U}(a,b)$ has the \gls{pdf}, while take $t = \frac{1}{\sigma}$ and $d\sigma = -\frac{1}{t^2}dt$:
  \begin{equation}
   f_X(x) = \frac{1}{b - a} \cdot\int_{t_b}^{t_a} \frac{-1}{t\sqrt{2\pi}} e^{-\frac{(x-\mu)^2}{2}*t^2} dt.
   \end{equation}
   \begin{figure}[!htbp]
		\centering
        \subfigure[Numerical assessment of distance error (meters) model\label{fig:FittedERRC}]{
            \includegraphics[width=0.82\linewidth]{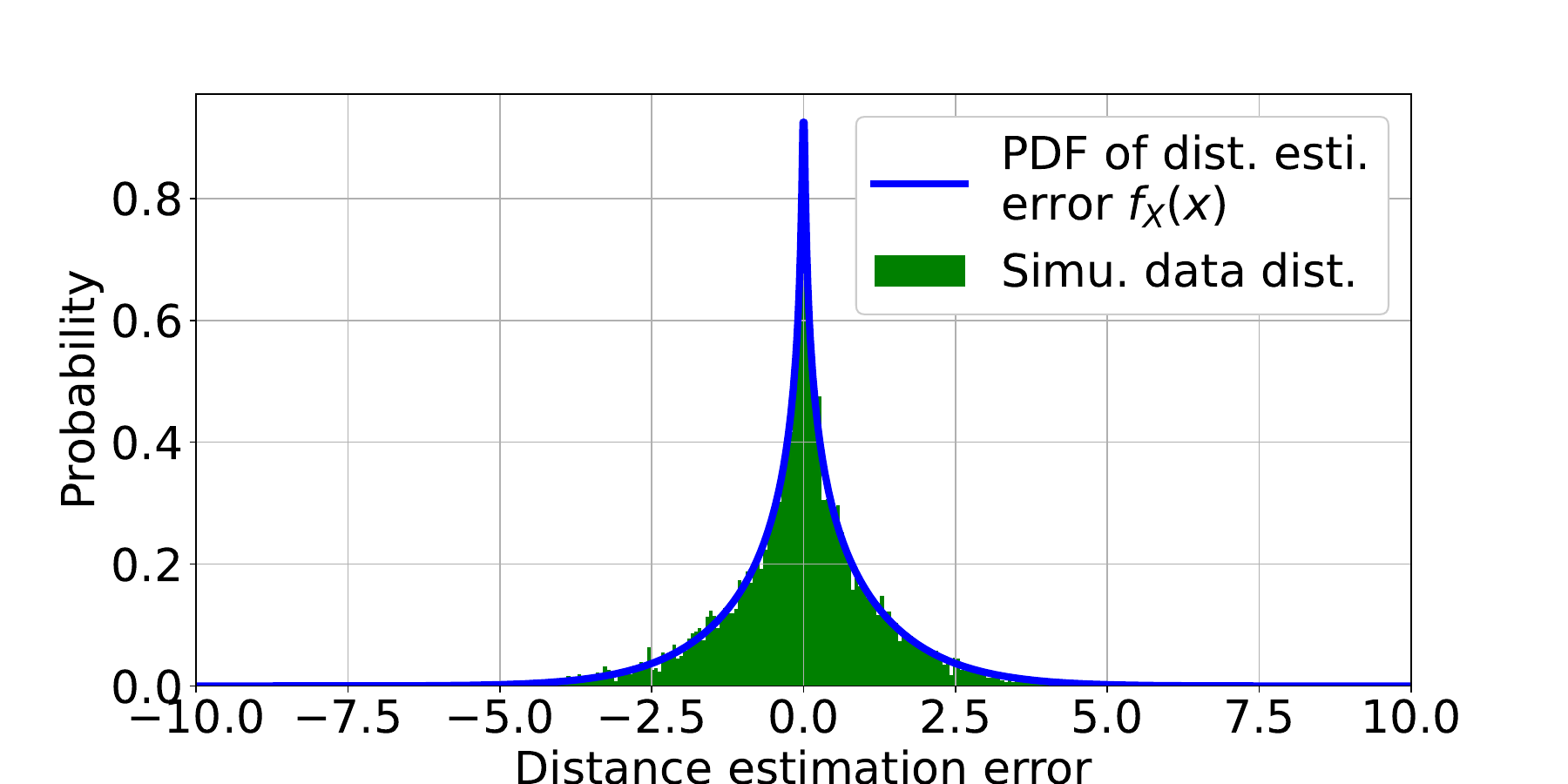}}
        \subfigure[Numerical assessment of converted error (meters) model\label{fig:FittedG}]{
		      \includegraphics[width=0.82\linewidth]{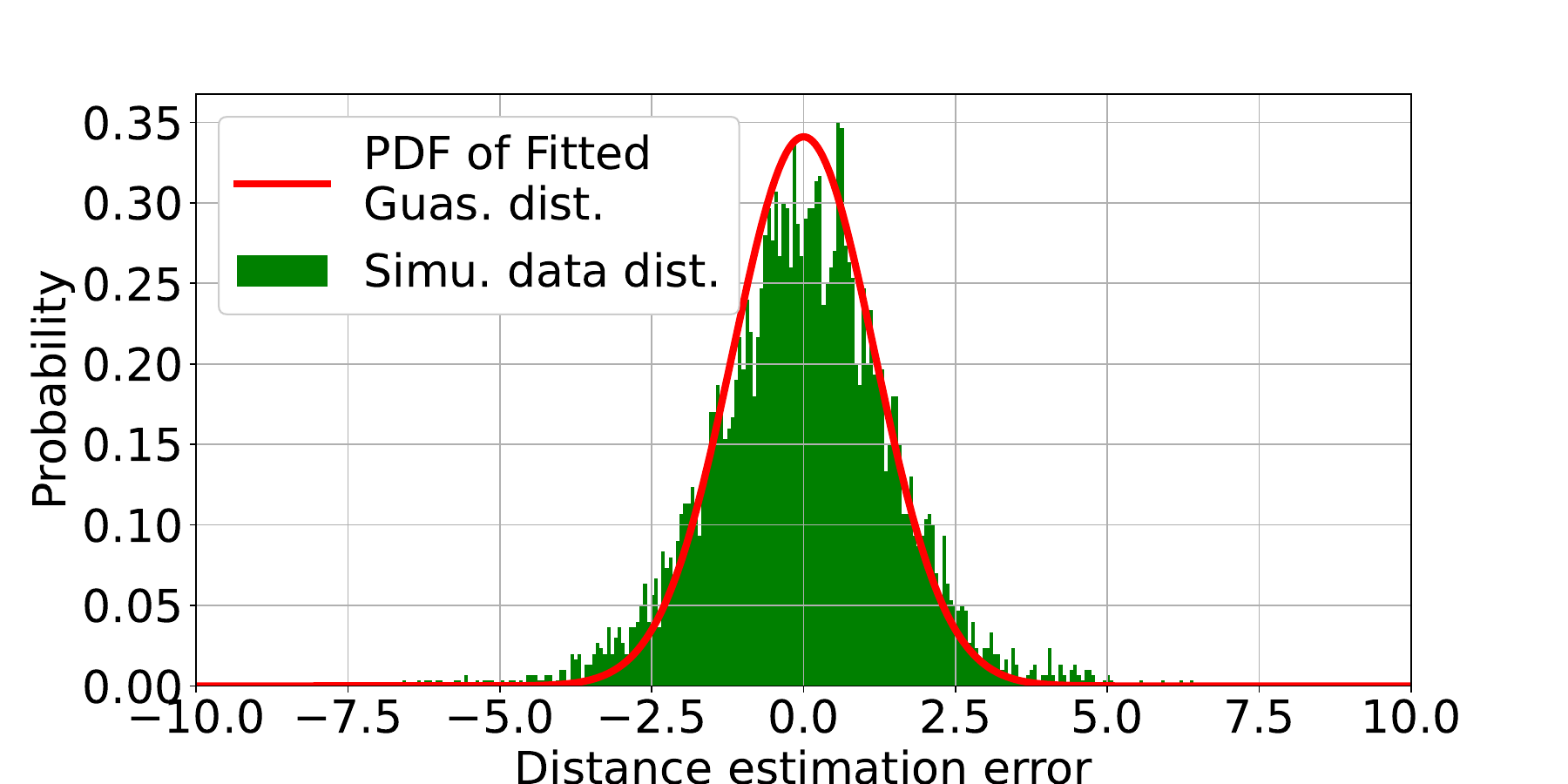}}
        \subfigure[Numerical assessment of converted error (meters) standard deviation\label{fig:coverteddist}]{
		      \includegraphics[width=0.82\linewidth]{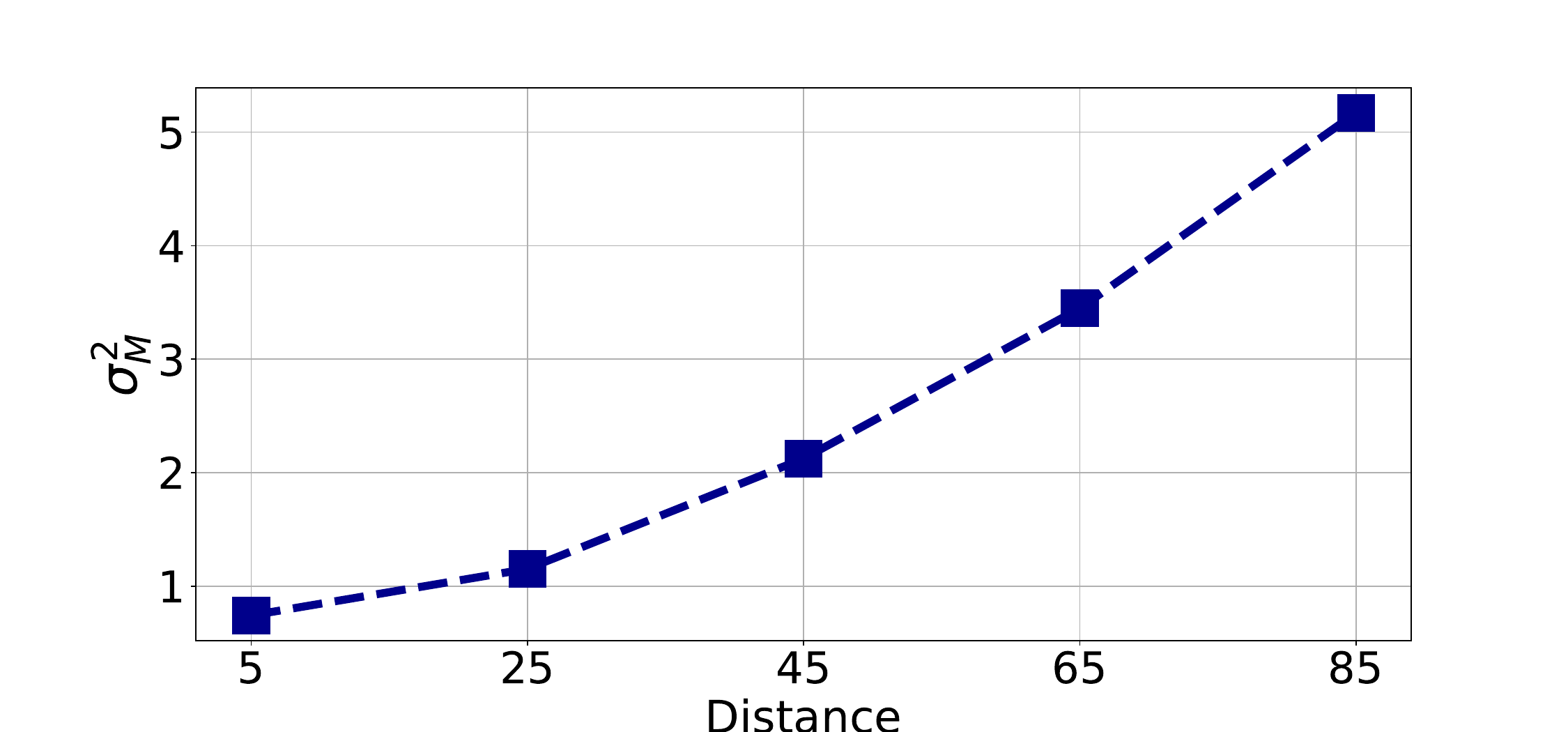}}
              \vspace{-1mm}
	\caption{Error modelling}
    \label{fig:errconvert}
 \end{figure}
 
  As illustrated, the $f_X(x)$ resembles an expression with $\frac{1}{t}$ in the integrand, reminiscent of the form of an error function. However, it shall be noted that $f_X(x)$ is analytically intractable and can only be evaluated numerically.
 With given parameters listed in Subsec.~\ref{disERR} and maximum distance $\textit{d}_\text{max} = 50\metre$, the \gls{pdf} of distance estimation error $\varepsilon_{k,n}^{d}$ can be accessed, presented in Fig.~\ref{fig:FittedERRC}. With the \gls{pdf} of $\varepsilon_{k,n}^{d}$ being accessed, we can further access the modelled distance estimation error $\varepsilon_{k}^{M}$.
  We assume position error is Gaussian distributed, the position measurement error can be easily converted to distance estimation error with $\varepsilon_{k,n}^{p} = \sqrt{(x_n^{\Delta})^2 + (y_n^{\Delta})^2 + (z_n^{\Delta})^2 }$. By adding $\varepsilon_{k,n}^{p}$ and $\varepsilon_{k,n}^{d}$ with random angle, $\varepsilon_{k,n}^{p}$ can be converted to $\varepsilon_{k,n}^{d}$, we can access the \gls{pdf} of $\varepsilon_{k}^{M}$, demonstrated in Fig.~\ref{fig:FittedG}, with maximum $d_{k,n} = 50\metre$ and path loss model parameters follow Tab.~\ref{tab:pathpara}, the position error power $\sigma_{\mathbf{p},n}$ ranging from $0.1$ to $3.0$. The \gls{pdf} of $\varepsilon_{k}^{M}$ can be approximated with a zero mean Gaussian distribution with $\varepsilon_{k}^{M}\sim \mathcal{N}(0 , \sigma_{M}^2)$. For computational convenience, the measured distances take the approximation ${d}^{\circ}_{k,n} \sim \mathcal{N}({d}_{k,n} , \sigma_{M}^2)$.
 Maintaining constant path loss model parameters and position error, we explore variations in the maximum $d_{k,n}$, representing coverage from $5\metre$ to $100\metre$. The relationship between $\sigma_M$ and coverage is shown in Fig.~\ref{fig:coverteddist}.

   \subsection{CRLB in uniform spatial distribution}
   To guide the design of multi-\glspl{uav} system functionalities, it is essential to discuss the \gls{crlb} of \gls{3d} cooperative localization.  For an unbiased estimator $\hat{\theta}=[\hat{\theta}_1, \hat{\theta}_2, \ldots, \hat{\theta}_q]^T$, its variance is constrained by $\sigma_{\hat{\theta}_i} \geqslant [I^{-1}(\theta)]_{ii}$, where $I(\theta)$ represents the \gls{fim} with dimensions $q \times q$ and $\theta$ represents the true value. \gls{fim} measures the amount of information that a random variable $z$ carries about $\theta$. \gls{fim} is defined in the following manner \cite{app13032008}:
   \begin{equation}\label{eq:fimeq}
   [\mathbf{I}(\theta)]_{ij} = -\mathbb{E}\left[\frac{\partial^2}{\partial\theta_i\partial\theta_j} \ln f(z;\theta)\right].
   \end{equation}
   Considering all UAVs uniformly distributed within the area, the conditional \gls{pdf} of measured distance distribution is
   \begin{equation}
   f({d}^{\circ}_{k,n}\vert\mathbf{p}_k) = \frac{1}{\sqrt{2\pi\sigma^2_M}} \exp\left\{\sum_{ n\neq k}^{N}\frac{-({d}^{\circ}_{k,n} - {d}_{k,n})^2}{2\sigma^2_M}\right\}.
   \end{equation}
   The Euclidean distance between $u_k$ and $u_n$ is defined as: $d_{k,n} = \sqrt{(x_k - x_n)^2 + (y_k - y_n)^2 + (z_k - z_n)^2}$. The \gls{fim} is expressed by the following equation Eq.~\eqref{eq:fim3D}.

   \gls{crlb} can be described as: 
   \begin{equation}\label{eq:CRLB}
    \sigma_{p}^2 \geqslant tr\{\mathbf{I}^{-1}(h)\},\\
   \end{equation}
   \begin{equation}\label{eq:ivFIM}
    \mathbf{I}^{-1}(h) = \frac{adj\{\mathbf{I}(h)\}}{|\mathbf{I}(h)|}.   
   \end{equation}  

   \begin{figure*}[!t]
     \begin{align}\label{eq:fim3D}
    \mathbf{I}(h) = \frac{1}{\sigma_M^2} \begin{bmatrix}
    \sum_{n=1}^{N} \frac{(x_k - x_n)^2}{d_{k,n}^2} & \sum_{n=1}^{N} \frac{(x_k - x_n)(y_k - y_n)}{d_{k,n}^2} & \sum_{n=1}^{N} \frac{(x_k - x_n)(z_k - z_n)}{d_{k,n}^2} \\
    \sum_{n=1}^{N} \frac{(x_k - x_n)(y_k - y_n)}{d_{k,n}^2} & \sum_{n=1}^{N} \frac{(y_k - y_n)^2}{d_{k,n}^2} & \sum_{n=1}^{N} \frac{(y_k - y_n)(z_k - z_n)}{d_{k,n}^2} \\
    \sum_{n=1}^{N} \frac{(x_k - x_n)(z_k - z_n)}{d_{k,n}^2} & \sum_{n=1}^{N} \frac{(y_k - y_n)(z_k - z_n)}{d_{k,n}^2} & \sum_{n=1}^{N} \frac{(z_k - z_n)^2}{d_{k,n}^2}
    \end{bmatrix}
    \end{align}
    \rule{\textwidth}{0.4pt} 
    \vspace*{\fill} 
     \end{figure*}

   \begin{theorem}\label{theorem:crlb}
   Considering all anchor \glspl{uav} uniformly distributed in space and positioned symmetrically around the target\revise{}{,} \gls{crlb} of the localization can be described as a function of $\sigma_M$ and number of anchor \glspl{uav}, denoted as $\gls{crlb}(\sigma_M,N)$. 
   \end{theorem}

   \begin{proof}
   \revise{}{See Appendix.}
   \end{proof}

   According to Theorem~\ref{theorem:crlb}, the \gls{crlb} of cooperative localization 
   problem solely relies on $\sigma_M^2$ and $N$. $\sigma_M^2$ is partially determined by $\Bar{d}$ in our error model. To verify the relationship between \gls{crlb} and $N$ as well as $\Bar{d}$, 
   we directly calculated \gls{crlb} with Eq\reviseprev{}{s}.~\eqref{eq:CRLB}--\eqref{eq:fim3D} under different combinations of target \gls{uav} coverage and anchor \gls{uav} number, each combination was simulated 100 times to exclude randomness. Thus, a two dimensional descriptor of \gls{crlb} can be presented in Fig.~\ref{fig:CRLB3d}.
   \begin{figure}[!htbp]
		\centering
		\includegraphics[width=0.72\linewidth]{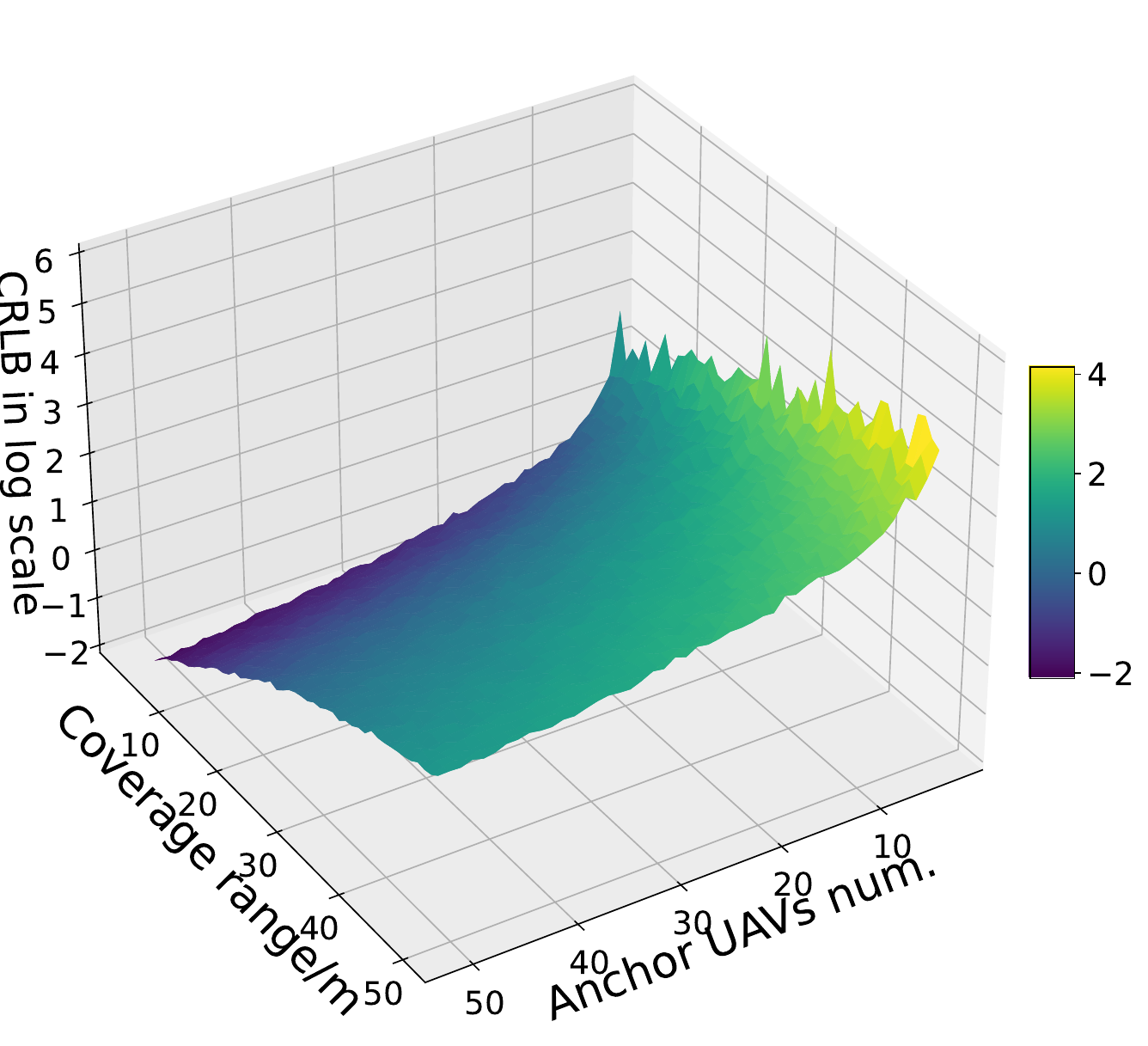}
        \vspace{-3mm}
		\caption{CRLB regarding the coverage and anchor UAVs number}
		\label{fig:CRLB3d}
	\end{figure}
 
   Our numerical assessment verified the relation between \gls{crlb} and the number of anchor \gls{uav}s, \revise{}{specifically as described in Eq.~\eqref{eq:CRLBclosed41}, namely $\sigma_{p}^2 \geqslant \frac{6\sigma_M^2}{N}$}: \revise{}{While verifying the exact qualitative relationship is challenging due to the impossibility of simulating an ideal uniform distribution of \glspl{uav} since $N$ can be small, a clear $\log(\frac{1}{N})$ curve is observable along the number and CRLB axes.}  Moreover, the results also validated the relation between \gls{crlb} and coverage. It is important to note that Eq.~\eqref{eq:CRLBclosed41} was derived under the assumption that \gls{uav}s adhere to a uniform distribution in three dimensions. 
   In practice, the expansion of coverage necessitates an increase in the number of anchor UAVs, thereby establishing a dynamic interdependence between these two factors.
   \subsection{CRLB in non-uniform spatial distribution}
   Computing the \gls{crlb} faces challenges when dealing with features to be estimated, which exhibit varying dynamic ranges. In a multi-UAV scenario, it is common for UAVs to be distributed non-uniformly in three dimensions. More specifically, UAVs are often denser in their distribution along the latitude dimension compared to the other two dimensions. \revise{}{Given the three-dimensional range constraints: $R_z \leqslant R_x = R_y$} The \gls{crlb}, when computed based on \gls{fim} without accounting for this scaling, may not be able to accurately represent the true lower bound of an estimator.  
   We can demonstrate it under an extreme case that all the \glspl{uav} are located in an area with its altitude range $R_z \to 0$. \revise{}{Referring Eqs.~(\ref{eq:CRLBclosed2})-(\ref{eq:fdn2}) in Appendix: $f_1$ represents the sum of squared areas for triangles formed by the target and two random anchor \glspl{uav} projected onto a Cartesian plane, normalized by their squared distances. $f_2$ represents the sum of volumes for triangles formed by the target and three random anchor \glspl{uav}, normalized by their squared distances. We now have:}\revise{}{\begin{align}\label{eq:limfdn}\begin{split}
       \lim_{R_z \to0}f_1 &= \lim_{R_z \to0}\sum_{n=1}^{N}\sum_{m=1}^{N}\frac{V_{n,m}^2(x,y)+V_{n,m}^2(y,z)+V_{n,m}^2(x,z)}{(d_{k,n}d_{k,m})^2}\\
       \lim_{z \to0}f_1 &= \sum_{n=1}^{N}\sum_{m=1}^{N}\frac{V_{n,m}^2(x,y)}{(d_{k,n}d_{k,m})^2,}
   \end{split}\nonumber
   \end{align}}
   \revise{}{\begin{align}\begin{split}
    \lim_{R_z \to0}f_2 & = \lim_{R_z \to0} \sum_{n=1}^N\sum_{m=1}^N\sum_{l=1}^{N}\frac{V_{n,m,l}^2(x,y,z)}{(d_{k,n}d_{k,m}d_{k,l})^2}\\
     \lim_{z \to0}f_2 &= 0.
    \end{split}\nonumber\end{align}}
   \revise{}{It is observed that $f_1$ converges to a constant, while $f_2$ converges to $0$. Without considering scaling, it is apparent that 
   \gls{crlb} may vary as $R_z$ varies, as per Eq.~\eqref{eq:CRLBclosed2}. }

  To demonstrate this, we calculated the \gls{crlb} based on Eqs.~(\ref{eq:CRLB})--(\ref{eq:ivFIM}) under that anchor UAVs were distributed within a spherical \revise{}{cap-like area (depicted in Fig.~\ref{fig:altitudedemo}) determined by varying $R_z$ and maximum distance between target and anchor \gls{uav} $\textit{d}_\text{max} = 50\metre$. } We then compared these \gls{crlb} results with a benchmark scenario in which all anchor \glspl{uav} were distributed within a spherical area with \revise{}{$R_x = R_y = R_z$ and $\textit{d}_\text{max} = 50\metre$}. \revise{}{The \gls{crlb} increases as $R_z$ decreases, which contradicts the findings shown in Fig.~\ref{fig:errorsp}. This suggests that conventional \gls{crlb} estimation approaches may not be suitable for multi-\gls{uav} cooperative localization problems.} 
   
  \revise{}{To handle one feature with different dynamic ranges and given typically known \gls{uav} distributions, a scale factor $s$ can be applied to the \gls{fim} in Eq.~\eqref{eq:fim3D} where this feature is involved. In the scenario previously described where $ R_z \leqslant R_x = R_y$ holds. We define the following constraints: \begin{enumerate*}[label=\emph{\roman*)}]
	  \item $s$ decreases monotonically with $R_z$,
    \item $s \in [1,+\infty)$,
    \item $s = 1$ when $R_x = R_y = R_z$.
\end{enumerate*} 
For \gls{3d} space, given that:
  $\textit{d}_\text{max}^2 = R_x^2 + R_y^2 + R_z^2$ and
  $R_x = R_y$. 
  We can derive:
  $R_x^2 = \frac{\textit{d}_\text{max}^2-R_z^2}{2}$
  The squared scale parameter $s$ is then defined as: 
  \begin{equation*}
      s = \frac{R_x}{R_z} = \frac{\sqrt{2(\textit{d}_\text{max}^2-R_z^2)}}{2R_z}.
  \end{equation*}
   Taking into account the first and second-order terms of this feature in Eq.~(\ref{eq:fim3D}), the corresponding scale matrix is defined in Eq.~(\ref{eq:scalemat}). The comparison between scaled and non-scaled \gls{crlb} results is shown in Fig.~\ref{fig:crlbscale}.}\revise{}{\begin{align}
       &\mathbf{I}_s(h) = \mathbf{I}(h) \circ \mathcal{S},\\
       \label{eq:scalemat} &\mathcal{S} = \begin{bmatrix} 1 & 1 & \frac{\sqrt{2(\textit{d}_\text{max}^2-R_z^2)}}{2R_z} \\ 1 & 1 & \frac{\sqrt{2(\textit{d}_\text{max}^2-R_z^2)}}{2R_z}\\ \frac{\sqrt{2(\textit{d}_\text{max}^2-R_z^2)}}{2R_z} & \frac{\sqrt{2(\textit{d}_\text{max}^2-R_z^2)}}{2R_z} & \frac{(\textit{d}_\text{max}^2-R_z^2)}{2R_z^2}\end{bmatrix}.
   \end{align}}
  
   \begin{figure}[!htbp]
     \centering
     \begin{subfigure}[CRLB without scaling (gray surface) compared to benchmark (black surface)\label{fig:crlbnoscale}]{
        \includegraphics[width=0.47\linewidth]{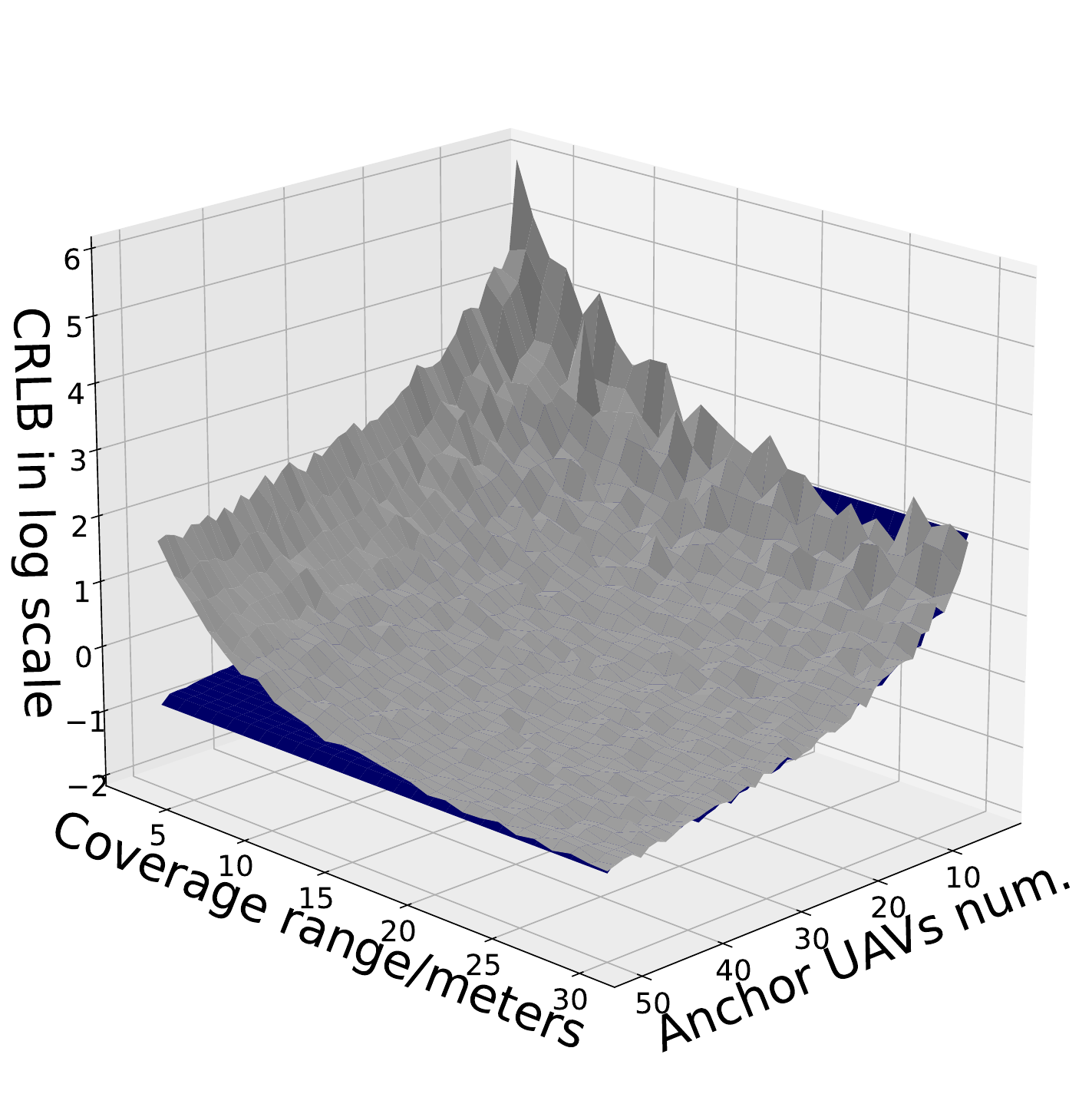}}
     \end{subfigure}
     \begin{subfigure}[CRLB with scaling (gray surface) compared to benchmark (black surface)\label{fig:crlbscale}]{
        \includegraphics[width=0.47\linewidth]{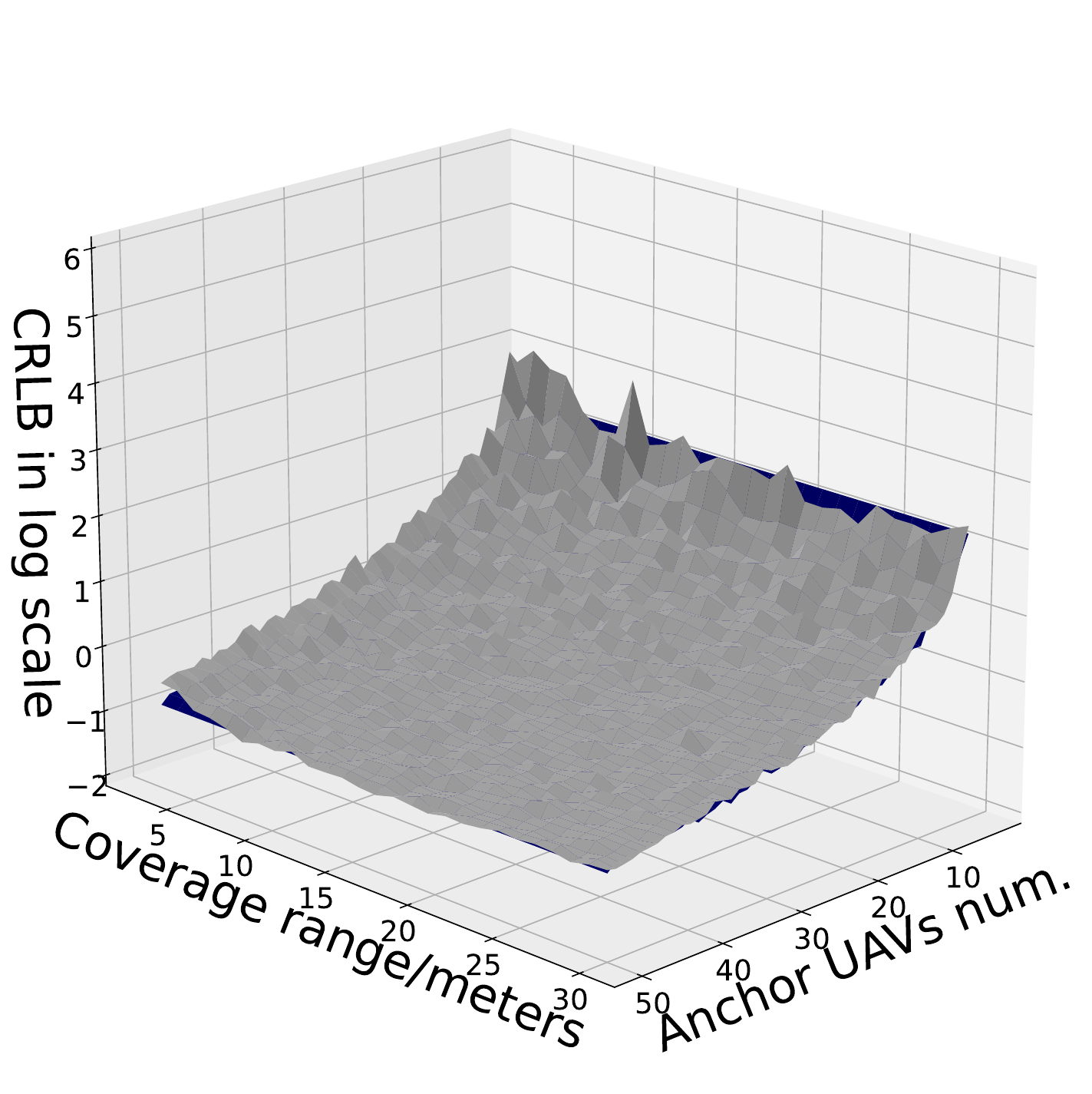}}
     \end{subfigure}
     \vspace{-2mm}
     \caption{CRLB in non-uniform spatial distribution}
     \label{fig:CRLB2subfigures}
   \end{figure}

   \section{Adaptive and robust cooperative localization}{\label{sec:MAGD}}

   \subsection{Introduction to different localization techniques}{\label{introlocalt}}
   Thorough studies have been conducted focusing on distance-based localization techniques. We focus on three different localization techniques: \gls{ls}  based localization, \gls{ln-1} based localization, and \gls{gd} based localization, to investigate their localization performance in the presence of non-uniform spatial distribution of anchor \gls{uav}s. These techniques have been widely recognized for their robustness and efficiency in sensor network scenarios \cite{Ragmalicous,GVW2012efficient}.
   
   \textbf{LS based estimation:}
   A multi-\gls{uav}s localization problem can be directly solved by \gls{ls} technique with
   \begin{equation*}
   [x_k,y_k,z_k,\Vert\mathbf{p}_k\Vert^2]^\mathrm{T} = \mathrm{(\mathbf{A}^T\mathbf{A})^{-1}\mathbf{A}^T\mathbf{b}},   
   \end{equation*} 
   where $\mathbf{A}$ and $\mathbf{b}$ are matrices containing anchor position and measured distances information:
   \[
   \mathbf{A} = \begin{Bmatrix}
   -2x_0 \kern-6pt & -2y_0 \kern-6pt & -2z_0 \kern-6pt & 1 \\
   \vdots \kern-6pt & \vdots \kern-6pt & \vdots \kern-6pt & 1 \\
   -2x_n \kern-6pt & -2y_n \kern-6pt & -2z_n \kern-6pt & 1 \\
   \end{Bmatrix};
   \mathbf{b} = \begin{Bmatrix}
   {{d}^{\circ}}^2_{k,0} - \Vert\mathbf{p}^{\circ}_0\Vert^2 \\
   \vdots   \\
   {{d}^{\circ}}^2_{k,n} - \Vert\mathbf{p}^{\circ}_n\Vert^2 \\
   \end{Bmatrix}.
   \]
   \revise{}{In many cases, \gls{ls} performance can be improved by incorporating a weight matrix $\mathbf{W}$ to account for known anchor error power. The \gls{wls} solution can be expressed as:
   \begin{equation*}
   [x_k,y_k,z_k,\Vert\mathbf{p}_k\Vert^2]^\mathrm{T} = \mathrm{(\mathbf{A}\mathbf{W}^T\mathbf{A})^{-1}\mathbf{A}^T\mathbf{W}\mathbf{b}}.  
   \end{equation*} }
   \revisefinal{}{As introduced earlier, the parameters \( \sigma_{\mathbf{p},n} \) and \( \sigma_d \) can be jointly represented by a single signal metric, \( \sigma_{M,n} \). This bound can be built using a pre-determined reference table. Given that \( \sigma_{M,n} \) is known, the weight matrix \( \mathbf{W}\) is defined as:  $\mathbf{W} = [1/\sigma_{M,1}\dots1/\sigma_{M,n}]^\mathrm{T}$}.

   \textbf{LN-1 based estimation:} Moreover, such a localization problem can be formulated as a plane fitting problem:
   \begin{align}
   \begin{split}
      &\min\limits_{u,w}\Vert w\Vert_1,\\
      &\text{subject to } \mathbf{A}u-W = \mathbf{b},\nonumber
   \end{split}
   \end{align}
   where the objective is to find a 4D plane $W = f(x, y,z)$ that fits the measurements $\mathbf{A}$ and $\mathbf{b}$. The coefficients of the plane are $u = [x_k,y_k,z_k,\Vert\mathbf{p}_k\Vert^2]^{\mathrm{T}}$. The optimization can be performed by minimizing the $l$1 norm-based distance metric, this localization technique is efficient and robust as demonstrated in \cite{Ragmalicous}. 
   \revise{}{The position coordinates $u$ can be solved iteratively using $l_1$-norm solvers such as \gls{fista}. The original \gls{fista} proposed in \cite{FISTABeck2009}, widely used in localization applications, alongside \gls{lp} solvers, can handle localization in wireless networks even with a small number of anchors \cite{LPlocXukun2016}. Similarly, \gls{admm} efficiently solves the target localization in wireless networks where measurements come from both \gls{nlos} and \gls{los} paths \cite{admmHe2021} .}
   
   \textbf{GD based estimation:} Eq.~\eqref{eq:optimize} can be also reformulated as $\mathbf{p}_k =\arg\min\limits_{[x,y,z]}f(x,y,z)$. By applying gradient descent to cost function $f(x,y,z)$, we are able to estimate the position $\hat{\mathbf{p}}_k$ of $u_k$ iteratively. At the $i_\mathrm{th}$ iteration, the negative gradient $g^{i}$ and position $\hat{\mathbf{p}}_k$ can be calculated
   \begin{equation}
		g^{i} = - \nabla_{(x,y,z)} (f(x,y,z))\vert_{(x=\hat{x}_k^{i-1},y=\hat{y}_k^{i-1},z=\hat{z}_k^{i-1})},\nonumber
   \end{equation}
   \begin{equation}
		\hat{\mathbf{p}}^i_k = \hat{\mathbf{p}}^{i-1}_k  + \alpha^{i} \times \frac{g^{i}}{||g^{i}||}.\nonumber
   \end{equation}
   where $\hat{\mathbf{p}}^{i-1}_k$ is the estimated position at the $(i-1)_{th}$ iteration. $\alpha^{i}$ is the step size at the $i_{th}$ iteration. $\alpha^{i}$ can be adjusted by discount factor $\beta$ to prevent over-descending. 
   Similarly, a convergence threshold $\theta_{\mathrm{GD}}$ can be applied.

  
   \subsection{Performance under non-uniform spatial distribution}\label{subsec:perunvenD}

    To compare the localization performance of the \gls{ls}, \gls{ln-1}, and \gls{gd} techniques in a cooperative localization scenario, we considered a setup where the target \gls{uav} $u_k$ is positioned amidst anchor \gls{uav}s distributed within a \revise{}{spherical cap-like area, as depicted in Fig.~\ref{fig:altitudedemo}}.
    \begin{figure}[!htbp]
		\centering
		\revisebox{\includegraphics[width=0.75\linewidth]{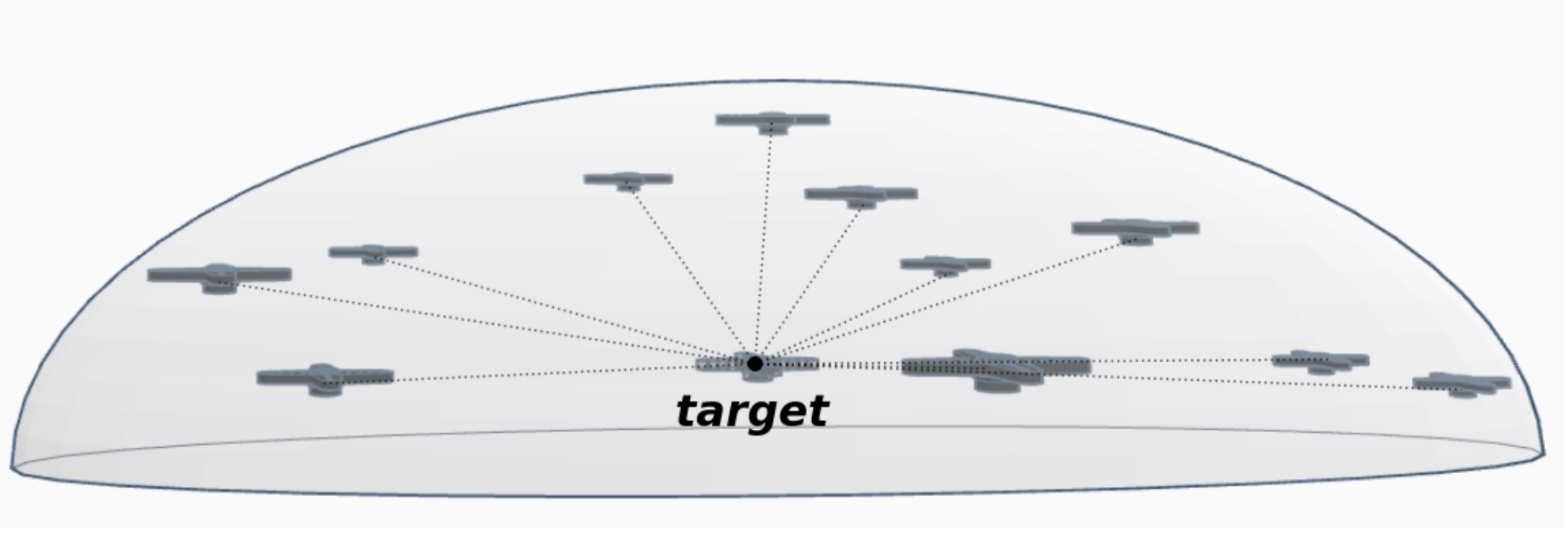}}
        \vspace{-3mm}
		\revise{}{\caption{Visualization of \glspl{uav} distribution}}
		\label{fig:altitudedemo}
   \end{figure}
    \revise{}{The $R_z$ varied from $3$ to $29$, and the maximum distance $d_{R}$ between the target and anchors is set to $50\meter$, therefore $R_x$ and $R_y$ will corresponding change}. We adopted an \gls{rssi} error model, as depicted in Fig.~\ref{fig:RSSImean2}, and configure the simulation with the following settings: $N = 30$, $[\sigma_{\mathrm{min}}^p, \sigma_{\mathrm{max}}^p] = [0.1,3]$. The parameters for each localization technique have been determined through our empirical observations to achieve optimal performance. Specifically, for LN-1, we set $K_{\mathrm{max}} = 300$, $\rho = 0.3$, and $\theta_{\mathrm{LN}} = 1\times 10^{-3}$. For the GD technique, the parameters were $K_{\mathrm{max}} = 50$, $[\alpha_0, \beta] = [1.5, 0.8]$, and $\theta_{\mathrm{GD}} = 1\times 10^{-5}$.
    Each estimation was repeated 100 times to account for randomness in the results, and the average error simulations are presented in Fig.~\ref{fig:errorsp}. 
   \begin{figure}[!htbp]
		\centering
		\revisebox{\includegraphics[width=0.69\linewidth]{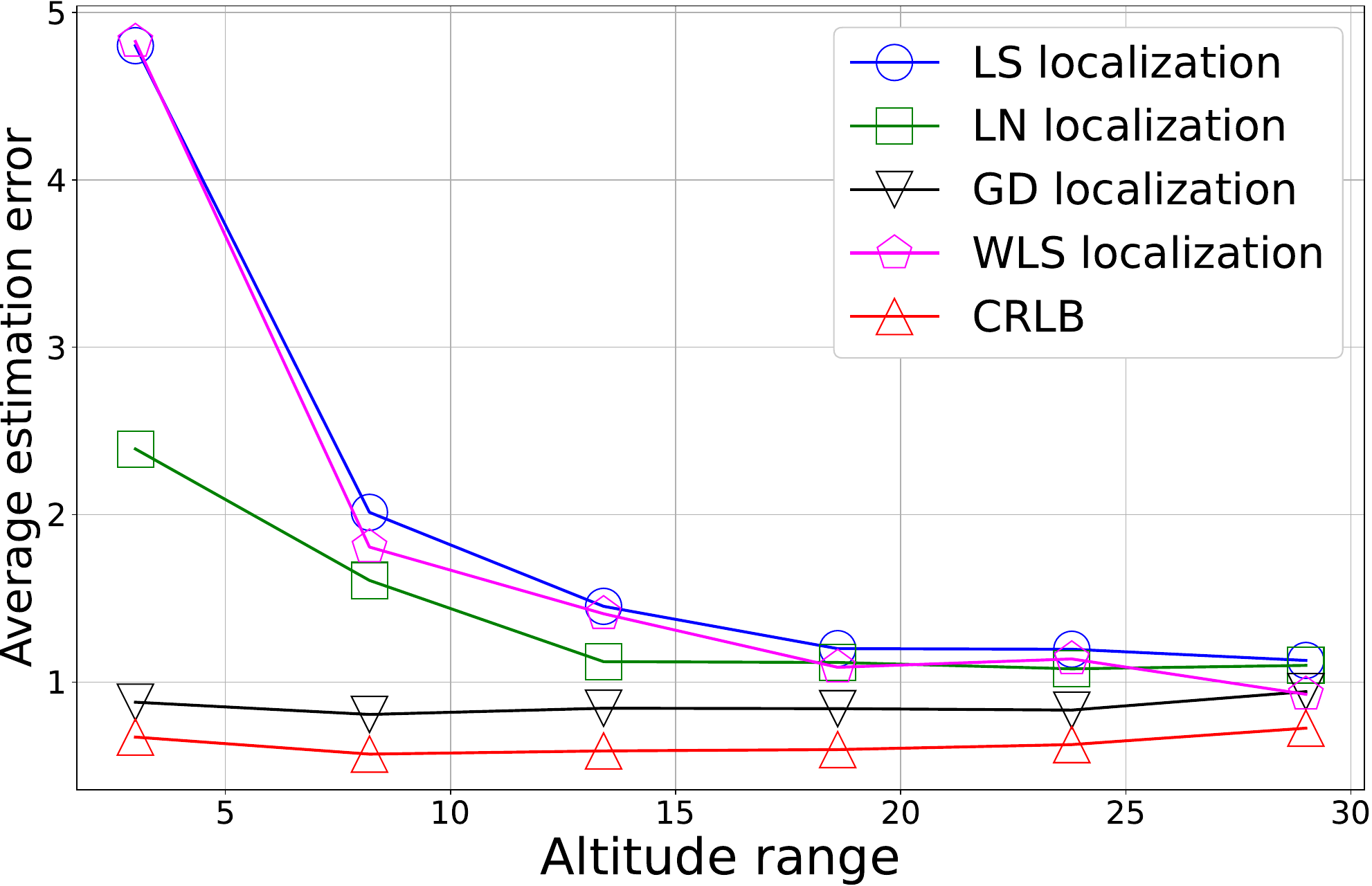}}
		\revise{}{\vspace{-1mm}\caption{Localization error over spatial distribution (scaled CRLB is used to access the estimation error)}}
		\label{fig:errorsp}
	\end{figure}
    
   The simulation results revealed that both the \revise{}{\gls{wls}}, \gls{ls} and \gls{ln-1} techniques are susceptible to the non-uniform distribution of \gls{uav}s. 
   In contrast, \gls{ln-1}-based localization exhibited lower sensitivity to this non-uniformity. In an effort to enhance robustness in localization and address security concerns, \gls{ln-1} \revise{}{technique based on \gls{admm}} modifies measurements while violating the linear constraint. This modification mitigates the impact of non-uniform features in the measurements. \revise{}{}\revise{}{Considering computational costs, \gls{ln-1} with \gls{admm} introduces a higher burden. While it addresses the non-uniform distribution issue, this comes at the expense of significantly greater complexity compared to \gls{ls}-based localization.} The performance of \gls{gd}-based localization appeared to be independent of the spatial distribution, as the gradient is adapted at each iteration based on the collected data. \gls{gd}-based localization is well-suited for scenarios involving multiple \gls{uav}s, where the distribution of \gls{uav}s may be non-uniform across all three dimensions.

    \subsection{The MAGD algorithm}
    The primary object of \gls{gd} is to minimize the loss function: 
    \begin{equation}
         f(x,y,z) = \sum_{n= 0,n\neq k}^{N-1}\left\vert\Vert\mathrm{p^{\circ}}_n-{\mathbf{p}}_k\Vert - {d}^{\circ}_{k,n}\right\vert.
    \end{equation}
    While acknowledging the presence of errors in all measurements and the actual location $p_k$ is unknown, obtaining the precise mathematical representation of $f(x, y, z)$ proves challenging. It is understood that $f(x, y, z)$ is a conventional convex function, typically threshold-ed by the \gls{crlb} and its maximum estimation loss $\mathbb{L}^m$ 
    \begin{equation*}
         \text{CRLB} < f(x,y,z) \leqslant \mathbb{L}^m, 
    \end{equation*}
    and also exemplified in Fig.~\ref{fig:convexCRLB}.
    \begin{figure}[!htbp]
		\centering
		\includegraphics[width=0.81\linewidth]{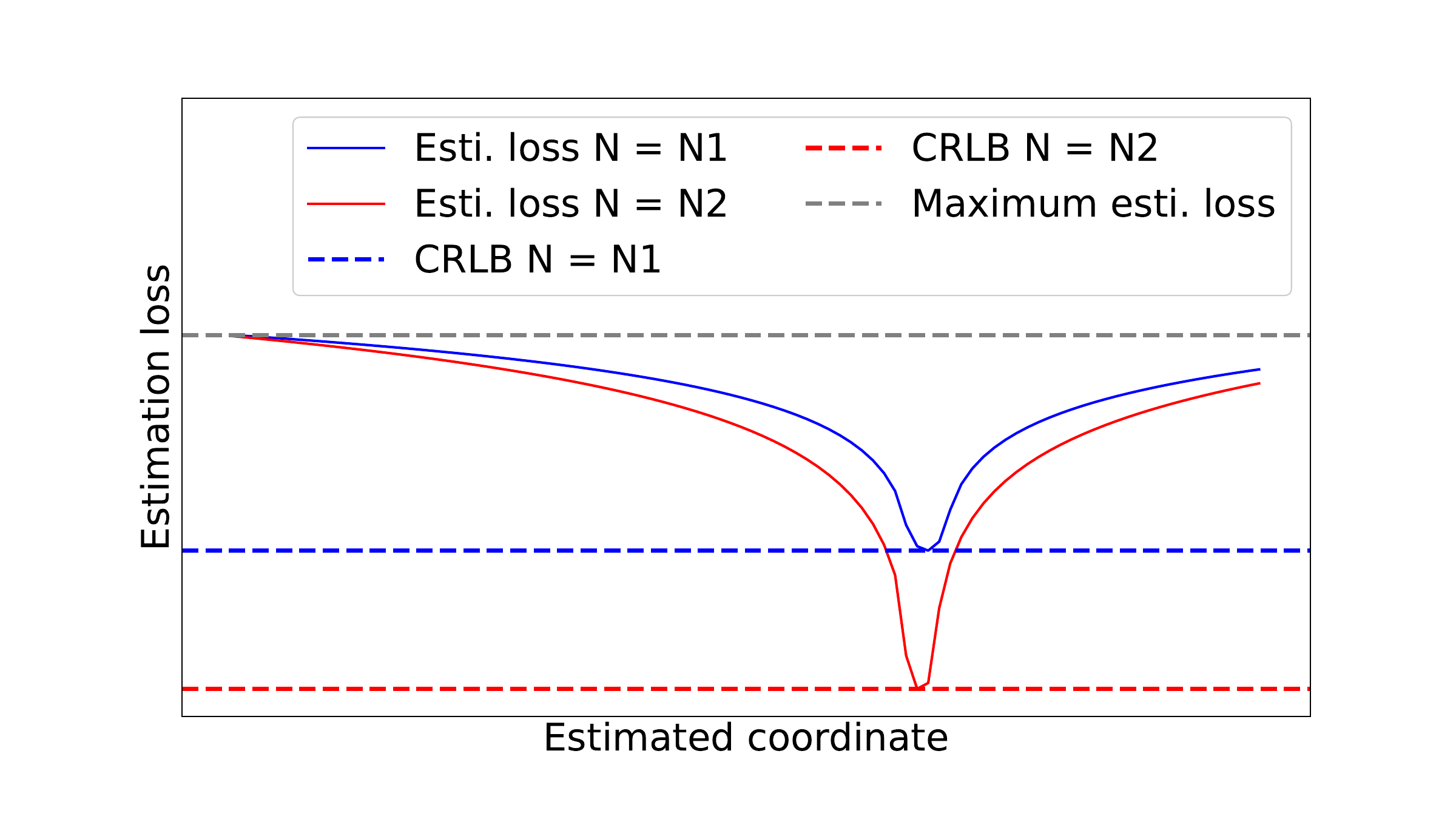}
        \vspace{-2mm}
		\caption{Estimation loss while estimated coordinate varies. It is assumed that $N_2 > N_1$, thus the loss function $N_2$ has a smaller minimum.}
		\label{fig:convexCRLB}
	\end{figure} In the context of continuous estimation scenarios, the previous estimates can be leveraged for better convergence and mitigate computational costs. Therefore, the maximum estimation loss at time step $t$ can be summarized: 
    \begin{equation}\label{eq:losovertime}
         \mathbb{L}_{t}^m = \mathrm{V}\mathrm{\Delta T} + \mathbb{L}_{t-1}, 
    \end{equation}
    where $\mathrm{V}$ is the speed of target \gls{uav} and $\mathrm{\Delta T}$ is the interval between two consequent estimation; $\mathbb{L}_{i-1}$ is the inherited loss from previous estimation. In the stochastic setting, where the convexity and smoothness of the loss function are unknown, choosing the optimal learning rate is not straightforward. One often tested learning rate schedule takes the form \cite{wu2018wngrad}: 
    \begin{equation}\label{eq:initiLR}
         \alpha_{t,k} = \frac{\mathbb{L}_{t,k}^m}{kG}, 
    \end{equation}
    where $k$ is the iteration and $G$ is the average of gradients. The loss overtime $\mathbb{L}_t$ typically follows a decreasing trend and converges above the \gls{crlb} when the speed $\mathrm{V}$ undergoes without large fluctuations. With small $N$, average gradients become less reliable toward optima. Despite common strategies like gradient normalization and \gls{sgd}, they aren't feasible here. Thus, a larger $\alpha$ is warranted when $N$ is small. Summarizing Eqs~.(\ref{eq:losovertime})-(\ref{eq:initiLR}), for a deterministic learning rate, $\alpha$ should be monotonically decreasing with respect to $N$, the time step $i$, and the iteration $k$. Meanwhile, $\alpha$ shall be monotonically increasing to the speed.   
    
    Another common and effective option in determining $\alpha$ is to start with a constant learning rate $\alpha_0$ which gives good empirical convergence and gradually decreasing it with a discount factor $\beta$ in subsequent cycles until convergence. Introducing a learning rate reduction Theorem and a convergence threshold can simplify determining the learning rate, removing the need to explicitly consider the parameter $k$. Combining the methodologies introduced above, a \gls{magd} algorithm \revise{}{(depicted in Fig.~\ref{fig:MAGDdiagram})} can be designed step-by-step to address these deterministic factors.
    \begin{itemize}
     \item \textbf{Step 1:} Initialize the learning rate $\hat{\alpha}$ at time step $t=0$. The loss $\mathbb{L}^m_{t=0}$ is typically very large, considering $\alpha$ to be monotonically decreasing with $N$.
     Set the learning rate: 
     {
     \begin{align*}
     \textbf{if} \quad &t = 0 \nonumber,\\
     &\hat{\alpha} = \max(\frac{\epsilon^{\mathrm{max}}_{\mathrm{t_0}}}{N}, \epsilon^{\mathrm{min}}_{\mathrm{t_0}}), \notag
     \end{align*}}

     where $\epsilon^{\mathrm{max}}_{\mathrm{t_0}}$ and $\epsilon^{\mathrm{min}}_{\mathrm{t_0}}$ are the maximum and minimum learning rate thresholds, determined empirically based on the range of $\mathbb{L}^m_{t=0}$.
     \item \textbf{Step 2:} Within each estimation, $\hat{\alpha}_{t,k}$ shall be reduced if the estimation process is over-descended, i.e. the loss at $k_{th}$
     iteration is larger than at $(k-1)_{th}$ iteration:
     {
     \begin{align*}
     \textbf{for} \quad &k = 1:K, \nonumber\\
     &\bar{d}_n = \Vert \hat{\mathbf{p}}-\mathbf{p}^{\circ}_n\Vert,\\
     &{D}_k = \frac{1}{n} *\sum\limits_{u_n\in\mathcal{U}}(\hat{d}_n - d^{\circ}_n+\mu^f_n)*w_n^f,\\
     &\textbf{if}\quad \bar{D}_k > \bar{D}_{k-1},\\
     &\quad\quad \hat{\alpha}_{t,k} =  \hat{\alpha}_{t,k-1}\beta_1.
     \end{align*}}
    $\bar{d}_n$ represents distance estimates based on the current position $\hat{\mathbf{p}}$. $\bar{D}_i$ is the average distance difference between $\bar{d}_n$ and the measured distance $d^{\circ}_n$, with weight factors $w_n^f$. $\beta_1$ denotes the discount factor for each estimate.
     \item \textbf{Step 3:} After each estimation beyond time step $t=0$, $\hat{\alpha}_t$ shall be reduced if the ongoing estimation proves stable, as inferred from the convergence of $\mathbb{L}^m_{t}$. While direct estimation of $\mathbb{L}^m_{t}$ is impossible, a viable alternative is to leverage $\bar{D}_t$, representing the average distance difference outputted after \textbf{Step 2}, as an indicator of $\mathbb{L}^m_{t}$.
     {\small
     \begin{align*}
     \textbf{for} \quad &t = 1:T \nonumber,\\
     &\Bar{\Bar{D}} = \frac{1}{t}\sum\limits_{t=1}^{t}\bar{D}_t,\\
     &\textbf{if}\quad\frac{\vert\bar{D}_t - \Bar{\Bar{D}}\vert}{\bar{\bar{D}}} <= 0.3,\\
     &\quad\quad {\alpha}_{t+1} =  \max(\hat{\alpha}_{t}-\beta_2,\frac{\epsilon^{\mathrm{min}}_{\mathrm{t_0}}}{n}).
     \end{align*}}
     \item \textbf{Step 4:} If $\mathrm{V}\to 0$, the estimation remains robust, with the learning rate reduced to a minimum threshold in \textbf{Step 3}. Consider the case harness the robustness the most: $\mathrm{V}$ increases and $\hat{\alpha}_t$ is too small for the current setup.
     {\small
     \begin{align*}
     \textbf{for} \quad &t = 1:T \nonumber,\\
     &\mathrm{V}_t=\Vert\hat{\mathbf{p}}_t - \hat{\mathbf{p}}_{t-1}\Vert,\\
     &\bar{\mathrm{V}} = \frac{1}{t}\sum\limits_{t=1}^{t}\mathrm{V}_t,\\
     &\rho= \frac{1}{t-\phi}\sqrt{\sum\limits_{t=t-\phi}^{t}\frac{\bar{D}_t\bar{\mathrm{V}}}{\bar{\bar{D}}\mathrm{V}_t}},\\
     &\textbf{if} \quad \rho > 1.3,\\
     &\quad\quad{\alpha}_{t+1} = {\alpha}_{t}\rho.
     \end{align*}}
     In the given scenario, where $\mathrm{V}_t$ represents the current estimated speed based on position estimation, and $\bar{\mathrm{V}}$ denotes the average speed, the ratios $\frac{\bar{D}_t}{\bar{\bar{D}}}$ and $\frac{\mathrm{V}_t}{\bar{\mathrm{V}}}$ offer insights into the variations relative to their respective averages. Recall Eq.~\eqref{eq:losovertime} , $\mathbb{L}_{t}^m$ is linear to $\mathrm{V}$ and $\bar{D}_t$ is an indicator to $\mathbb{L}_{t}^m$, if $\frac{\bar{D}_t}{\bar{\bar{D}}}$ and $\frac{\mathrm{V}_t}{\bar{\mathrm{V}}}$ increase in parallel, it indicates a configuration change, with the estimation still tracking the actual values effectively. If the increase in $\frac{\bar{D}_t}{\bar{\bar{D}}}$ surpasses that of $\frac{\mathrm{V}_t}{\bar{\mathrm{V}}}$, it suggests that the current $\hat{\alpha}$ is too small to effectively reduce the loss. In such cases, enlarging $\hat{\alpha}$ becomes necessary to better suit the current configuration. To capture trends and mitigate fluctuations, we use the square root and introduce a smoothing window $\phi$ to enhance the robustness of the estimation process. \revise{}{It should be noted that advanced \glspl{uav} equipped with velocity meters, accelerometers, and direction sensors could enhance \gls{magd} performance. These measurements can further improve localization, particularly when anchor information is less reliable, potentially exceeding the \gls{crlb} in poor configurations. However, this work focuses on a minimal system where only position data is available.}
    \end{itemize}   
    \begin{algorithm}[!htbp]
    \caption{\gls{gd} algorithm}
    \label{alg:MAGD}
    \scriptsize
    \DontPrintSemicolon
        {    \For {$n = 1:N$}{get $\mathbf{p}_n^{\circ}, \sigma_{\mathbf{p},n}, d^{\circ}_{n}$ from $u_n$ 
        
                $w_n^f = \frac{\text{max}(\sigma_n^f)}{\sigma_n^f}\quad $; $\hat{d}_n = \Vert \hat{\mathbf{p}} - \mathbf{p}_n^{\circ}\Vert$ }     
             $G^i \gets \sum\limits_{u_n\in\mathcal{U}}\frac{(\hat{\mathbf{p}}-\mathbf{p}^{\circ}_n)*w_n^c }{\hat{d}_n}*(\hat{d}_n - d^{\circ}_n+\mu_{c,n})$ \tcp*{Gradient}
             update: $\hat{\mathbf{p}} \gets \hat{\mathbf{p}} + m*\hat{\mathbf{p}} + \frac{\hat{\alpha}}{n}* \frac{G^i}{\Vert G^i \Vert} $; }
   \end{algorithm}
   \begin{figure}[!htbp]
		\centering
		\revisebox{\includegraphics[width=0.90\linewidth]{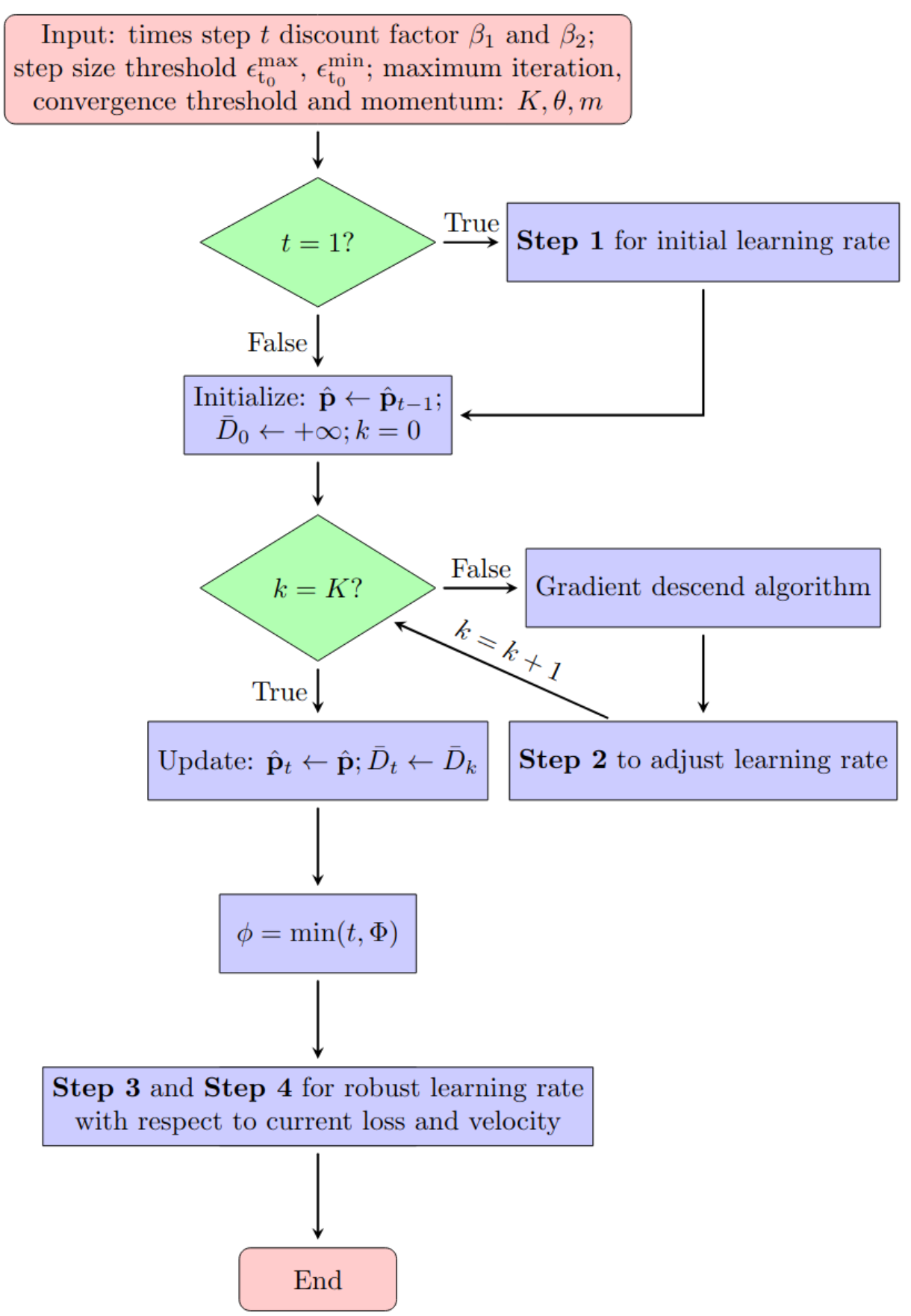}}
		\revise{}{\vspace{-1mm}\caption{Flowchart of \gls{magd} algorithm}}
		\label{fig:MAGDdiagram}
   \end{figure}
   \subsection{Accuracy estimation of cooperative localization system}\label{subsec:evaluofmu}
   To assess the robustness and accuracy of \gls{magd}, we conducted simulations varying anchor \gls{uav} numbers and initial step sizes ${\alpha_0}$ (${\alpha_0}$ will be reduced by $\beta$, it is the same approach we used in Subsec.\ref{subsec:perunvenD}). Subsequently, we compared the simulation results based on \gls{magd}. To emulate real-world mobility, we introduced periodic changes in the movement speed of $u_k$. The anchor \gls{uav}s, denoted as $u_n$, were randomly initialized within a spherical area surrounding the target \gls{uav} $u_k$. Details of the system configuration are provided in Tab.\ref{tab:setup1}.

   \begin{table}[!htbp]
		\centering
        \scriptsize
		\caption{Simulation setup 1}
		\label{tab:setup1}
		\begin{tabular}{>{}m{0.2cm} | m{1.6cm} l m{3.7cm}}
			\toprule[2px]
			&\textbf{Parameter}&\textbf{Value}&\textbf{Remark}\\
			\midrule[1px]
			
            
			  &$n_a$&$5\sim40$& Anchor \gls{uav}s\\

			
			&$\sigma^2_{\mathbf{p},n}$&$\sim\mathcal{U}(0.1,3.0)$& Position error power / \si{\meter^2}\\
			&$V$&$\sim\mathcal{U}(0.6,3.4)$& Travel speed \si{\meter/s}\\
            & $T$ & 50 s &  Simulation time\\ 
            \midrule[1px]
            \multirow{-7}{*}{\rotatebox{90}{\textbf{System}}}
            &$[\epsilon^{\mathrm{max}}_{\mathrm{t_0}},\epsilon^{\mathrm{min}}_{\mathrm{t_0}}]$&[50,5]& Step size thresholds\\ 
            &$[\beta_1,\beta_2]$&$[0.5,0.05]$& Discount factors\\
			 
            &$m$&$1\times 10^{-5}$& Momentum\\
            &$\theta$&$1\times 10^{-8}$& Convergence threshold\\
			\multirow{-5}{*}{\rotatebox{90}{\textbf{MGAD}}}
            & $K$ & 30 &  Maximum iteration\\
            \bottomrule[2px]
		\end{tabular}
	\end{table}
    \begin{figure}[!htbp]
		\centering
		\includegraphics[width=0.72\linewidth]{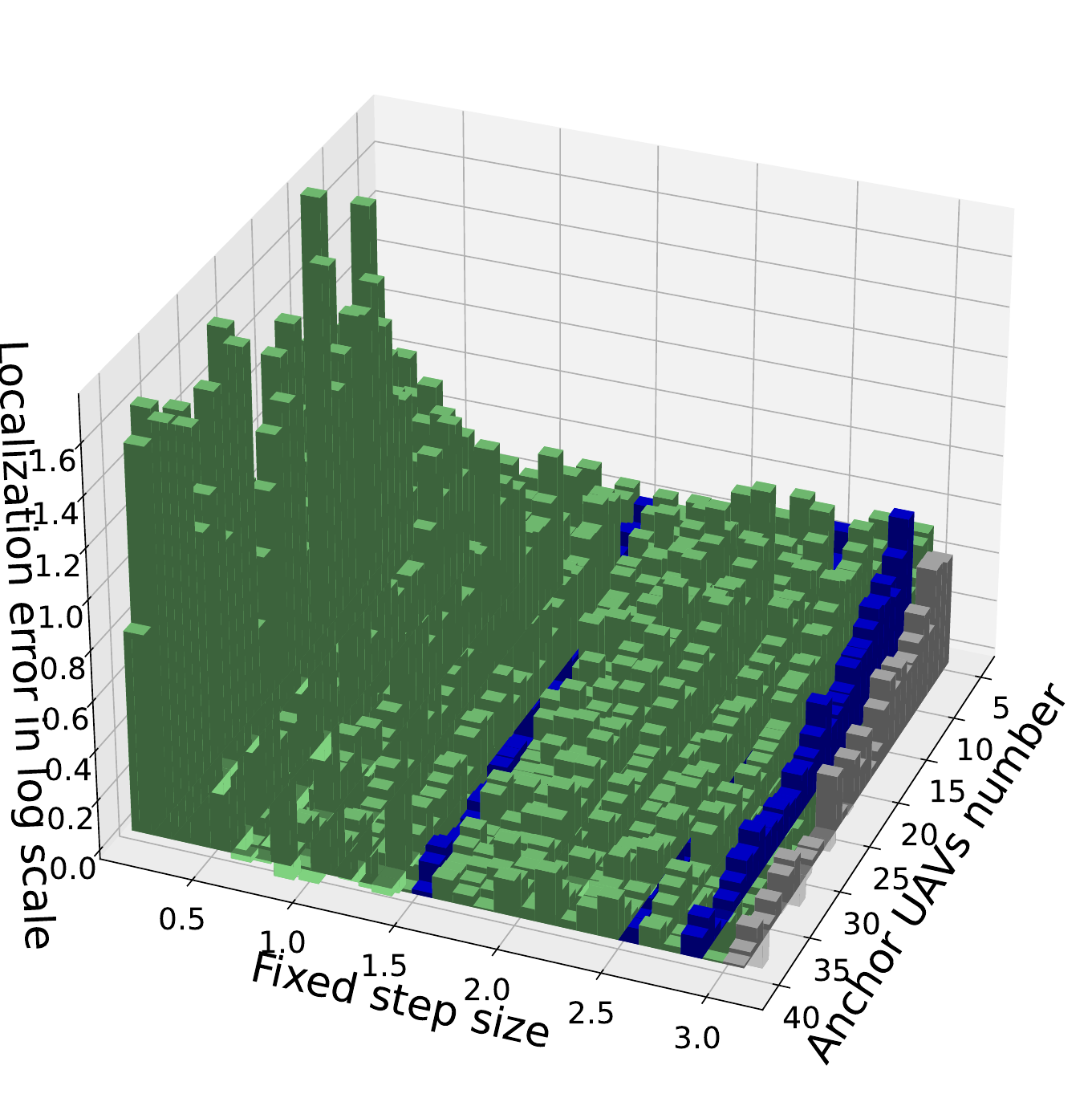}\vspace{-2mm}
		\caption{Localization error of fixed step sizes and \gls{magd} (3 best step sizes are marked in blue and \gls{magd} in gray)}
		\label{fig:MAGDaccuracy}
	\end{figure}
     The simulation results in Fig.~\ref{fig:MAGDaccuracy} \revise{}{and Tab.~\ref{tab:simuMAGD}} reveal average errors from 50 estimates. For the best fixed step sizes, \revise{}{$\alpha = (1.5, 2.5, 2.8)$} yielded average localization errors of $(1.63,\text{m}, 1.67,\text{m}, 1.66,\text{m})$ across all instances of $n_a$. In comparison, \gls{magd} achieved an average error of 1.47,\text{m}, outperforming fixed $\alpha^0$ approaches. Step sizes within the range of 1.4 to 3.0 demonstrated reliable performance. However, identifying  this range can be challenging in practical applications. By adjusting $\alpha^0$ to different scenarios, our approach efficiently ensured accurate and robust position estimation. 
    \begin{table}[ht]
    \centering
    \revise{}{\caption{Average localization error (meters) of $\alpha = [0.1, 2.9)$; last element marked in gray is from \gls{magd}}}
    \revisebox{
   \begin{tabular}{|c@{\hspace{0.15cm}}|c@{\hspace{0.15cm}}|c@{\hspace{0.15cm}}|c@{\hspace{0.15cm}}|c@{\hspace{0.15cm}}|c@{\hspace{0.15cm}}|c@{\hspace{0.15cm}}|c@{\hspace{0.15cm}}|c@{\hspace{0.15cm}}|c|}
   \hline
   16.44 & 6.68 & 5.29 & 3.49 & 2.99 & 2.40 & 2.21 & 2.00 & 1.72 & 1.78 \\ \hline
   1.68  & 1.69 & 1.72 & 1.68 & 1.63 & 1.72 & 1.68 & 1.73 & 1.70 & 1.70 \\ \hline
   1.71  & 1.68 & 1.72 & 1.67 & 1.67 & 1.73 & 1.72 & 1.66 & 1.72 & \cellcolor{gray!30}1.47 \\
   \hline
   \end{tabular}\label{tab:simuMAGD}}
   \end{table}

    \section{Attack Paradigm and attack detection system}\label{attack}
    \subsection{Attack paradigm}{\label{attackpara}}
    \gls{uav}s cooperative localization faces diverse potential attacks that can be systematically categorized. These attacks span different aspects, each presenting unique challenges to the accuracy and reliability of localization systems. 
    
    \textbf{Jamming mode}:
    In the jamming mode, attackers jam the beacon signal of $u_n$ to introduce large distance estimation errors and induce a wrongly received $\mathbf{p}(n)$ and $\sigma_{\mathbf{p},n}$. The received signal can be modeled as
    \begin{equation*}
      \mathbf{r}_{\mathrm{J}} \leftarrow[\mathbf{p}^{\circ}_n + \mathrm{\tilde{p}}_n, ({d}^{\circ}_{k,n}+\tilde{d}_{k,n})^{+}, \sigma_{\mathbf{p},n} + \tilde{\sigma}_{\mathrm{J}}].
    \end{equation*}
     Take the simplification $\mathrm{\tilde{p}}_n\sim\mathcal{N}^3\left(0,\tilde{\sigma}_{\mathrm{J}}^2/3\right)$, $\tilde{d}_{k,n}\sim\mathcal{U}\left(0,\tilde{\sigma}_{\mathrm{J}}^2\right)$, where $\tilde{\sigma}^2_{\mathrm{J}}$ is the jamming power.

    \textbf{Bias mode}:
    In the bias mode, attackers hijack some of the \gls{uav}s to erroneously report its position with a position bias to mislead others. In this scenario, the received signal can be modeled as 
    \begin{equation*}
    \mathbf{r}_{\mathrm{B}} \leftarrow [\mathbf{p}^{\circ}_n + \tilde{\mathrm{B}},{d}^{\circ}_{k,n},\sigma_{\mathbf{p},n}], 
    \end{equation*}
    where $\tilde{\mathrm{B}}$ denotes position bias.

   \textbf{Manipulation mode}:
    In the manipulation mode, attackers hijack some of the \gls{uav}s to report its position with an extra error, simultaneously modifying its $\sigma_{\mathbf{p},n}$ to be extremely small for the intention of manipulating $w_n^f$. The received signal can be concluded as 
    \begin{equation*}
    \mathbf{r}_{\mathrm{M}} \leftarrow [\mathbf{p}^{\circ}_n+\mathrm{\tilde{p}}_n, {d}^{\circ}_{k,n}, 1/\tilde{\sigma}_{\mathrm{M}}], 
    \end{equation*}
    where $\mathrm{\tilde{p}}_n\sim\mathcal{U}\left(0,\tilde{\sigma}^2_{\mathrm{M}}/3\right)$, $\tilde{\sigma}^2_{\mathrm{M}}$ is the manipulation index.
    
    In dynamic scenarios, attack orchestration is categorized as:

    \textbf{Global random attack:}
    Under the global random attack strategy, all malicious UAVs are uniformly distributed and randomly target nearby UAVs. The objective here is to degrade position estimation on a global scale and potentially evade existing attack detection systems, as proposed in \cite{HKZ+2023trustawareness}.

    \textbf{Global coordinated attack:}
    Similar to the global random attack, the global coordinated attack involves malicious \gls{uav}s coordinating their attacks within a specific time frame. This strategy also aims to disrupt position estimation globally. 

    \textbf{Stalking strategy:}
    The stalking strategy entails all malicious \gls{uav}s following a specific victim \gls{uav} and consistently launching attacks against it. While this strategy may not have a widespread impact on global estimation accuracy, its focus is on targeting the victim \gls{uav}.

    \subsection{Anomaly detection and reputation propagation mechanism}
    Considering the previously mentioned attack schemes, a resilient attack detection algorithm should reduce the trustworthiness of suspicious sources. To achieve this, a reputation weight $r_n$ can be used in Step 2 and line 4 of Alg.~\ref{alg:MAGD}:
    \begin{equation}\label{eq:MAGDref}
    \begin{split}
     G^i &\gets \sum\limits_{u_n\in\mathcal{U}}\frac{(\hat{\mathbf{p}}-\mathbf{p}_n^{\circ})*w_n^f*r_n }{d_n}*(d_n - d^{\circ}_n), \\
     \bar{D}_i &\gets \frac{1}{n} *\sum\limits_{u_n\in\mathcal{U}}(\hat{d}_n -   d^{\circ}_n+\mu_{n}^f)*w_n^f*r_n. \nonumber
    \end{split}
    \end{equation}
    Our method in Algorithm~\ref{alg:TDAD} detects anomalies via the \gls{cdf}. In lines 9-10, $u_k$ computes the estimated distance error, $\hat{\mathrm{\mathcal{E}}}_n$, using data from $u_n$.

    To address potential manipulation attacks, lines 12 confine $\sigma_n^f$ to a predetermined minimum position error power, denoted as $\sigma_{\mathrm{min}}^p$. The algorithm subsequently scrutinizes the cumulative density $\xi_n$ against a pre-established probability threshold $\epsilon^t$ to discern the attack behavior of $u_n$. Based on these findings, the reputation weight undergoes an update. This updated $r_t^n$ takes into account its previous value $r_{t-1}^n$, incorporates a forget factor $\gamma$, and the corresponding reward or penalty. Finally, $r_t^n$ is thresholded to remain within the range of [0,1]. The underlying objective of this approach is to counter coordinated attacks effectively. It enables $r_t^n$ to gradually recover when attacks from $u_n$ become less frequent, all while preserving the accuracy of cooperative localization.
     
    \begin{algorithm}[!htbp]
    \caption{\Acrlong{tad}}
    \label{alg:TDAD}
    \scriptsize
    \DontPrintSemicolon
    Input: Reward $\lambda_r$ and penalty $\lambda_p$, 
			 forget factor $\gamma$ and confidence threshold $\epsilon^t$\\
    Output: $r^n_t$
    
    Initialize: $r^n_{t=1} \gets 1$
    \SetKwProg{Fn}{Function}{ is}{end}
    \Fn{\textsc{}$(t = 1:T)$} {
        get current position estimate $\hat{\mathbf{p}}_t$ 
        
        {    \For {$n = 1:N$}{get $\sigma_{\mathbf{p},n}, d^{\circ}_n, \mathbf{p}^{\circ}_n$ from $u_n$
        
                   get $\mu^f_n,\sigma^f_n$
                   
                   
                   $\hat{d}_n \gets \Vert \hat{\mathbf{p}} - \mathbf{p}_n\Vert$ 
                   
                   $\hat{\mathrm{\mathcal{E}}}_n\gets \vert \hat{d}_n - d^{\circ}n + \mu_n^f\vert$ \tcp*{Distance error of $u_n$}
                   
                   calculate CDF $\xi_n$, 
                   
                    $\sigma_n^f \gets \max(\sigma_n^f,\sigma_{\mathbf{p},min})$\\
                   $\xi_n \gets P_{\hat{\mathrm{\mathcal{E}}}_n} \left[\hat{\mathrm{\mathcal{E}}}_n~\vert~\mu_n^f,(\sigma_n^f)^2\right] $ 
                   
                   \uIf{$\xi_n > \epsilon^t$}{
                
                $\hat{r}_n \gets \lambda_r$ \tcp*{Assign reward}
               }\Else{
               
               $ \hat{r}_n \gets \lambda_p $ \tcp*{Assign penalty} }

             update $ r_n(t) $

             $ r_t^n \gets \gamma*(r^n_{t-1} + 1) - 1 + \hat{r}_n $

             $ r_t^n \gets \min(1,\max(0,r_t^n) $
             } 
    }}
    \end{algorithm}
    
    \begin{equation*}
      \tilde{r}_{k,n}  = \frac{\sum\limits_{m\notin{k,n}}{r}_{k,m}*{r}_{m,n}}{\sum\limits_{m\notin{k,n}}{r}_{k,m}} , \tilde{\tilde{r}}_{k,n} = \frac{\mathrm{F_p}(\tilde{r}_{k,n})+{r}_{k,n}}{2}.
    \end{equation*}
    While $u_k$ leverages the cooperative localization system, $u_1...u_n$ represent the accessible anchor \gls{uav}s, and ${r}_{k,n}$ signifies the local reputation. $u_m$ has uploaded its local reputation to the cloud, allowing access to the reputation ${r}_{m,n}$ from $u_m$ to $u_n$. Simultaneously, $u_k$ possesses a local reputation ${r}_{k,m}$ towards $u_m$. The uploaded reputation will be evaluated based on the local reputation ${r}_{k,m}$. A propagated reputation $\tilde{r}_{k,n}$ tends to favor \gls{uav}s with a good reputation with respect to $u_k$. The propagation function $\mathrm{F_p}$ (a convex function), designed to discriminate against already notorious $u_n$, is then applied. Subsequently, the mean of the local reputation and the propagated reputation is utilized in \gls{magd}.
    \subsection{Effectiveness analysis without anomaly detection}
    When subjected to falsified information, estimators encounter more complex conditions. The \gls{crlb} is consequently altered, as the usable information available to the estimator is manipulated. For \gls{crlb}, we consider two distinct scenarios: a lower bound case where no falsified information is involved, designated as \gls{crlb}-1, and an upper bound case where the estimator possesses no anomaly detection capability, designated as \gls{crlb}-2. The \gls{crlb} for jamming and manipulation modes can be understood as follows: As previously discussed, the jamming mode is capable of misleading distance estimation and introducing minor unbiased errors in position information. In contrast, the manipulation mode can inject significant unbiased position errors. Despite the difference, the estimator under those two modes can still be considered unbiased,\revise{}{~\cite{estimation1975Brad}, as the mean of possible estimates still satisfies $\mathbf{E}(\hat\theta) = 0$.} \revise{}{As demonstrated in Fig.~\ref{fig:FittedG}, position errors are transformed into distance estimation errors, and the modeled error in both distance and position estimates can be approximated by a Gaussian distribution.} Consequently, the modeled distance estimation error under attacks can be assumed to follow a Gaussian distribution with standard deviation $\sigma'_M$. Considering the set of \glspl{uav} providing accurate information as $\mathcal{U}_n$, and the set providing malicious information i.e. deliberately falsified information as $\mathcal{U}_{nm}$, respectively. The \gls{pdf} of measured distance can thus be expressed as:
    \begin{equation}\label{eq:fasifiedPDF}
    \begin{split}
    f(\mathbf{p}_k) &= \frac{(1-{p}_a)}{\sqrt{2\pi\sigma^2_M}} \exp\left\{\sum_{ n\in\mathcal{U}_n}\frac{-({d}^{\circ}_{k,n} - {d}_{k,n})^2}{2\sigma^2_M}\right\}  \\
    &+ \frac{{p}_a}{\sqrt{2\pi{\sigma^{'}_M}^2}} \exp\left\{\sum_{ n\in\mathcal{U}_{nm}}\frac{-({d}{'}_{k,n}^{\circ}- {d}_{k,n})^2}{2{\sigma^{'}_M}^2}\right\}.
   \end{split}
   \end{equation}
   ${d}{'}_{k,n}^{\circ}$ represents the modeled distance error when malicious UAVs provide falsified information. $\mathrm{p}_a$ denotes the probability of malicious information being injected. The \gls{fim} under attacks can be reformulated as follows: 
    \begin{equation}\label{eq:fimofattk}
        \tilde{\mathbf{I}}(h) = (1-{p}_a)\mathbf{I}(h) + {p}_a\mathbf{I}'(h).
    \end{equation}
   where $\mathbf{I}(h)$ represents the \gls{fim} of anchor \glspl{uav} providing accurate information, as described in Eq.~\eqref{eq:fim3D}. $\mathbf{I}'(h)$ denotes the falsified \gls{fim} from malicious \glspl{uav}. For bias attack, malicious \glspl{uav} will be providing position information with the same bias, thus the second part of Eq.~\eqref{eq:fasifiedPDF} can be replaced with \revise{}{\begin{equation}
    \cdots + \frac{{p}_a}{\sqrt{2\pi\sigma^2_M}} \exp\left\{\sum\limits_{ n\in\mathcal{U}_{nm}}\frac{-({d}_{k,n}^{\circ}+d^B_{k,n}- {d}_{k,n})^2}{2\sigma^2_M}\right\}, 
   \end{equation}
   where $d^B_{k,n}$ indicates the injected bias.}
    The \gls{crlb} was designed for unbiased estimators, thus it can't be directly applied while the information were biased. Combining Eq.~\eqref{eq:fasifiedPDF}, the \gls{pdf} will be linearly shifted towards the bias with probability $\mathrm{p}_a$. We consider \gls{crlb}-2 for the bias mode to be the sum of \gls{crlb}-1 and a linearly varied bias.
   We assign the following attack parameters:
   \begin{enumerate*}
   [label=\emph{\roman*)}]
   \item $\tilde{\sigma}^2_{\mathrm{J}} = 8$\reviseprev{}{,}
   \item $\tilde{\sigma}^2_{\mathrm{M}} = 400$\reviseprev{}{,}
   \item $\tilde{B} = [3,3,3]$\reviseprev{}{, and}
   \item ${p}_a$ ranging from $0.15$ to $0.75$, while the average number of anchor \glspl{uav} is $20$. 
   \end{enumerate*}

   The \gls{crlb}-1 and \gls{crlb}-2 of the three attack modes are depicted in Fig.~\ref{fig:crlb_eva}. Although direct comparison is challenging due to differing parameters for each attack mode, we can still draw the conclusions: The bias mode attack performance shows a linear relationship with the malicious information rate, in other words, proportional to adversary invested resources. Contrastingly, manipulation attacks exhibit an exponential relation, with a steep increase when the rate exceeds $55\%$. Fig\reviseprev{}{s}.~\ref{fig:bias_eva}--\subref{fig:jam_eva} demonstrate\reviseprev{d}{} the actual attack performances of three attack modes as well as coordinated and random strategies. In general, a coordinated strategy is more efficient than a random strategy. Meanwhile, the simulated attack performances of bias and manipulation attacks align with \gls{crlb}-2. The simulated attack performance of the jamming mode does not align with \gls{crlb}-2, due to the weighted approach we used, which significantly mitigated jamming attacks and stabilized localization.
   
   \begin{figure}[!t]
  \centering
  \subfigure[\label{fig:crlb_eva}]{\centering
  \reviseboxfinal{\includegraphics[width=0.44\columnwidth]{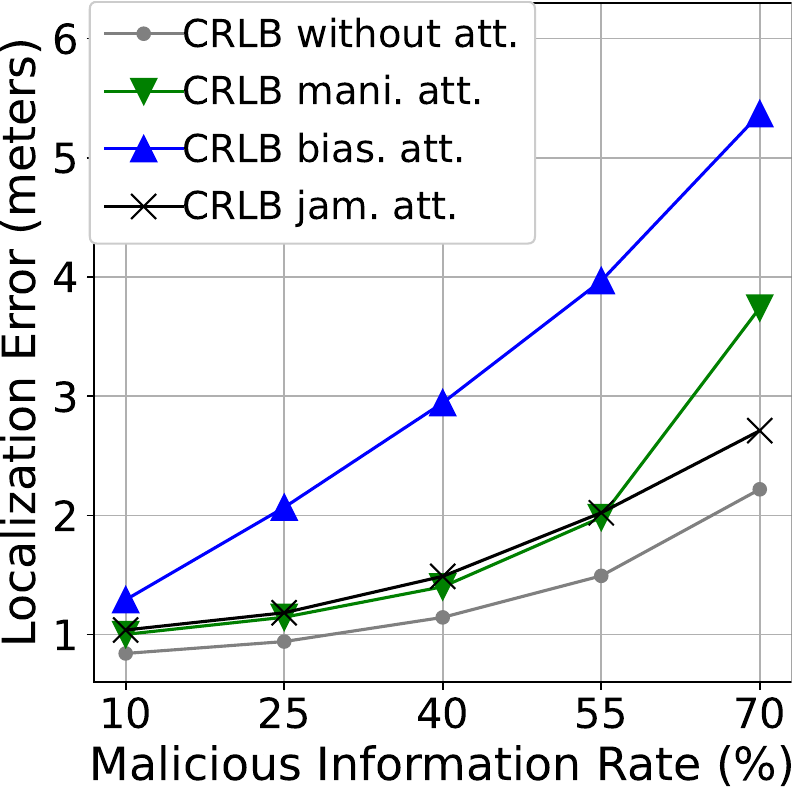}}}
  \hfill
  \subfigure[\label{fig:bias_eva}]{\centering
  \reviseboxfinal{\includegraphics[width=0.46\columnwidth]{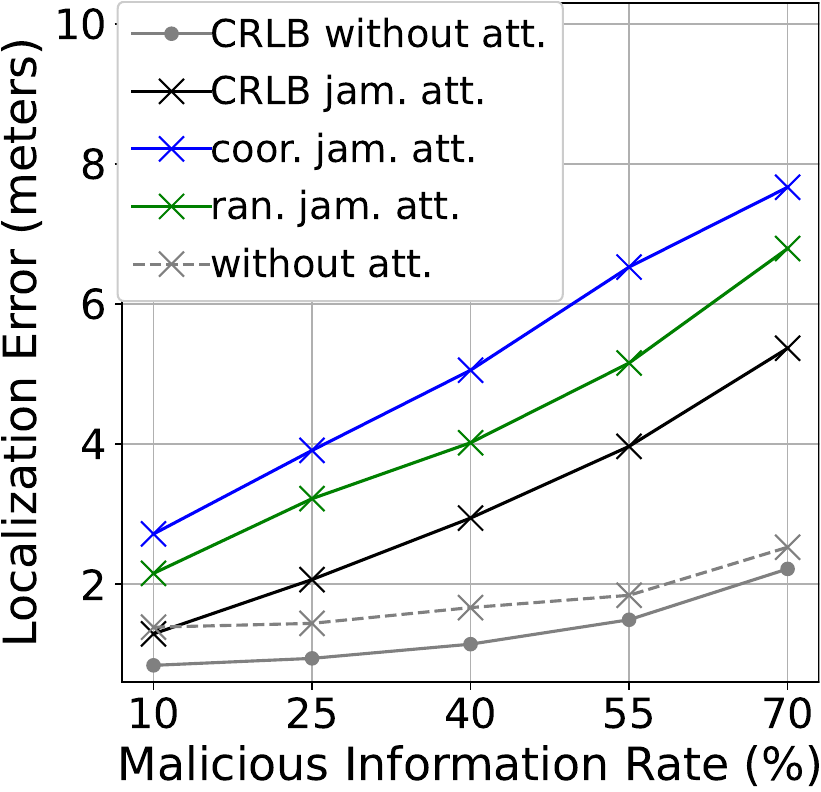}}}
  \\ 
  \subfigure[\label{fig:mani_eva}]{\centering
  \reviseboxfinal{\includegraphics[width=0.442\columnwidth]{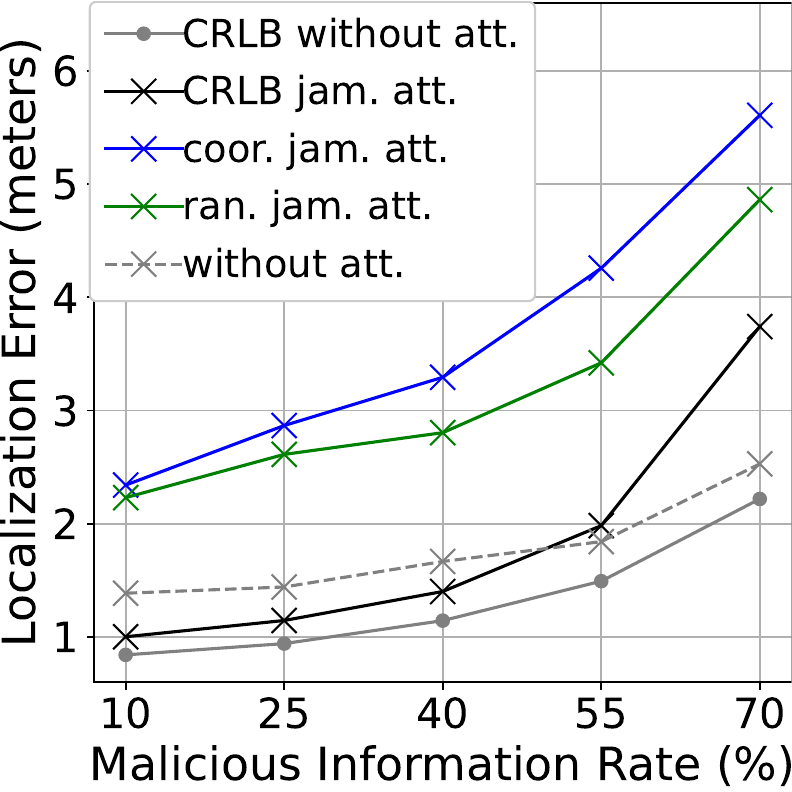}}}
  \hfill
  \subfigure[\label{fig:jam_eva}]{\centering
  \reviseboxfinal{\includegraphics[width=0.464\columnwidth]{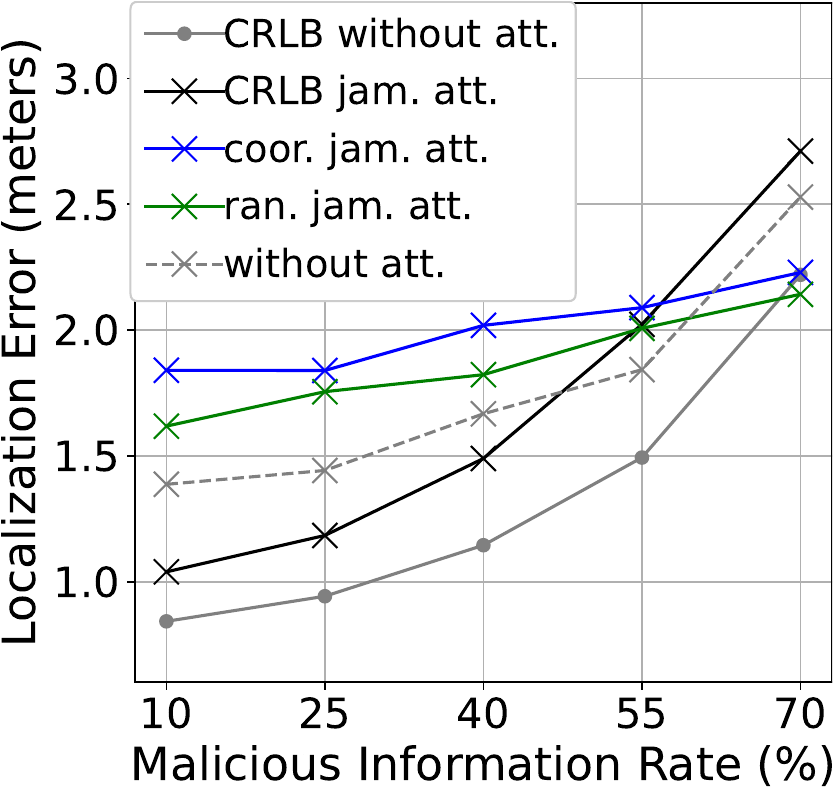}}}\vspace{-1mm}
  \caption{Comparison of \gls{crlb} and attack effectiveness: \subref{fig:crlb_eva} \gls{crlb}-2 of different attack modes, \subref{fig:bias_eva} effectiveness of bias attack, \subref{fig:mani_eva} effectiveness of manipulation attack, and \subref{fig:jam_eva} effectiveness of jamming attack, respectively}
  \label{fig:CRLB4mode}
  \end{figure}
  \subsection{Resource-constrained attacks under anomaly detection}\label{subsec:reconatk}
   Based on our analytical and simulation results, we can conclude that the effectiveness of attacks can be characterized as a function $\mathbf{E}(p_a,A_{t})$, where $p_a$ is the malicious information rate and $A_{t}$ denotes the attack parameter (including $\tilde{\sigma}^2_\mathrm{J}$, $\tilde{B}$ or $\tilde{\sigma}^2_\mathrm{M}$ when applicable). Attack parameters quantify the magnitude of deviation between falsified and true information. $\mathbf{E}$ has several characteristics, such as: 
   \begin{equation}\label{eq:atkcharacter}
   \begin{cases}
    \forall A_{t} \geqslant 0, \mathbf{E}(0, A_{t}) = 0,\\[1ex]
    \forall p_a \in [0,1], \mathbf{E}(p_a, 0) = 0,\\[1ex]
    \forall p_a \in [0,1] \cap A_{t} \in [0,+\infty), \frac{\partial \mathbf{E}}{\partial p_a} > 0, \\[1ex]
    \forall p_a \in [0,1] \cap A_{t} \in [0,+\infty), \frac{\partial \mathbf{E}}{\partial A_{t}} > 0.
    \end{cases}
    \end{equation}
    
     Similarly, the detection rate of an ideal anomaly detector can be described as a function $\mathbf{P}_d(p_a,A_{t})$. $\mathbf{P}_d$ has several characteristics: 
     \begin{equation}\label{eq:pdcharacter}
     \begin{cases}
     \forall p_a \in [0,1] \cap A_{t} \in [0,+\infty), \mathbf{P}_d \in (0,1),\\[1ex]
     \forall p_a \in [0,1] \cap A_{t} \in [0,+\infty), \frac{\partial \mathbf{P}_d}{\partial p_a} < 0, \\[1ex]
     \forall p_a \in [0,1] \cap A_{t} \in [0,+\infty), \frac{\partial \mathbf{P}_d}{\partial A_{t}} > 0.
    \end{cases}
    \end{equation}
   Eq.~\eqref{eq:pdcharacter} shows that as the attack parameter $A_{t}$ increases, the likelihood of successful detection improves. Conversely, as the proportion of false information injected into the system rises, detection accuracy tends to decline. Therefore, the effectiveness of attacks under anomaly detection can be written as:
    \begin{equation}\label{eq:atkeffectivesness1}
    \mathbf{E}_d(p_a,A_{t}) = \mathbf{E}(p_a,A_{t})\big(1-\mathbf{P}_d(p_a,A_{t})\big)
   \end{equation}
   $\mathbf{E}_d(p_a,A_{t})$ will then be \reviseprev{low-bounded}{lower bounded} with \gls{crlb} excluding all the falsified information, as well as the information from \reviseprev{miss-detected}{misdetected} \glspl{uav}, the \reviseprev{miss-detection}{false alarm} rate noted as $\mathbf{P}_f$, 
    \begin{equation}\label{eq:atkeffectivesness2}
    \mathbf{E}_d(p_a,A_{t}) \geqslant \max\Big(\mathbf{E}_d(p_a,A_{t}), \mathrm{CRLB}\big(N(1-p_a\mathbf{P}_d-\mathbf{P}_f)\big)\Big)
   \end{equation}
   In most scenarios, adversaries operate with limited resources, allowing them to manipulate only a small subset of \glspl{uav} within a multi-\gls{uav} system. Assuming $p_a$ is a constant \reviseprev{and $p_a < 0.5$}{below $0.5$} and anomaly detection is in place, one potential strategy for the adversary is to optimize attack parameters to maximize attack effectiveness.
   To find the optimal attack parameter, we can examine the derivative of $\mathbf{E}_d$ with respect to $A_{t}$.
   For simplicity, we write $\frac{\partial\mathbf{E}_d}{\partial A_t}$ as $\mathbf{E}'_d$, $\frac{\partial\mathbf{E}}{\partial A_t}$ as $\mathbf{E}'$, and $\frac{\partial\mathbf{P}_d}{\partial A_t}$ as $\mathbf{P}'_d$. The derivative of $\mathbf{E}_d$ with respect to $A_{t}$ is given by:

   \begin{equation*}\label{eq:firstderiavtive}
    \mathbf{E}'_d = \mathbf{E}'(1-\mathbf{P}_d) - \mathbf{P}'_d \mathbf{E}.
   \end{equation*}
   An optimal attack parameter exist when $\mathbf{E}'_d = 0$, we now have\reviseprev{,}{} 
   \begin{equation*}\label{eq:firstderiavtive1}
   \mathbf{E}'  = \mathbf{P}'_d \mathbf{E}+ \mathbf{P}_d \mathbf{E}',
   \end{equation*}
   \begin{equation*}\label{eq:firstderiavtive2}
   \mathbf{E}'  = \mathbf{P}'_d \int\mathbf{E}'+ \mathbf{P}_d \mathbf{E}'   .
   \end{equation*}
   Considering $\mathbf{E}'>0$,  we can divide both sides by $\mathbf{E}'$\reviseprev{}{:}
   \begin{equation*}\label{eq:firstderiavtive3}
   1  = \int \mathbf{P}'_d + \mathbf{P}_d.
   \end{equation*}
   Therefore, an optimal attack parameter exists when $\mathbf{P}_d = \frac{1}{2}$\reviseprev{}{.}
  \begin{figure}[!htbp]
  \centering
  \subfigure[\label{fig:concave}]{\centering
  \includegraphics[width=0.46\columnwidth]{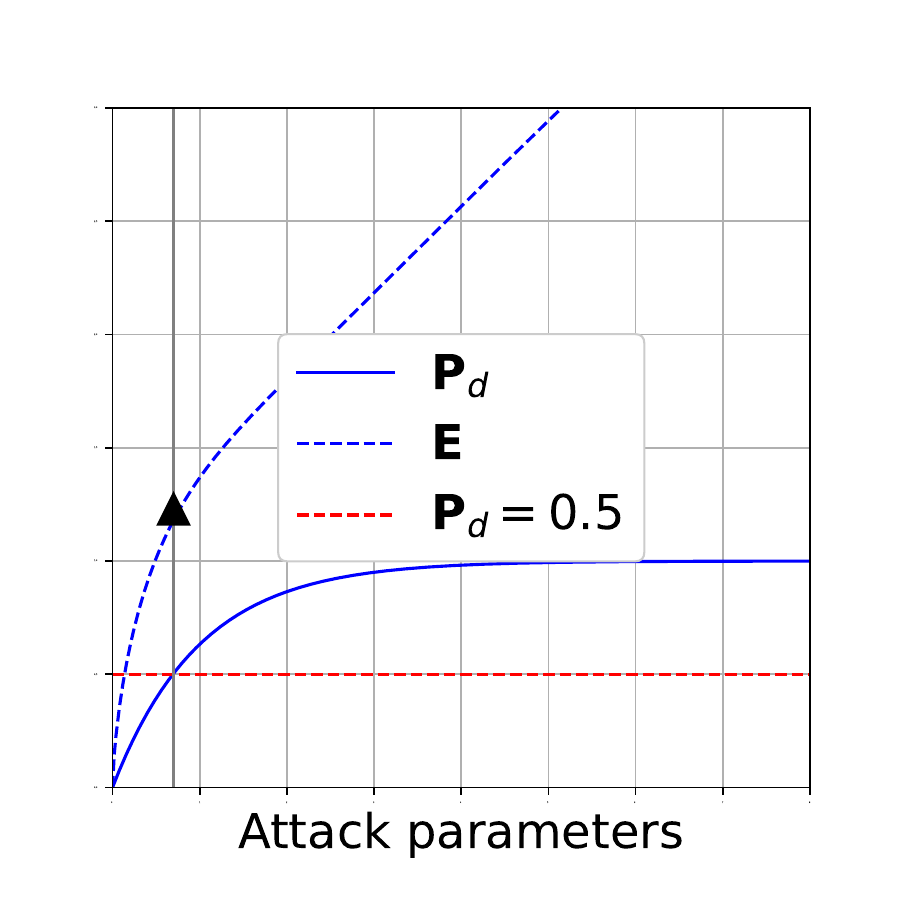}}
  \hfill
  \subfigure[\label{fig:non-convex}]{\centering
  \includegraphics[width=0.46\columnwidth]{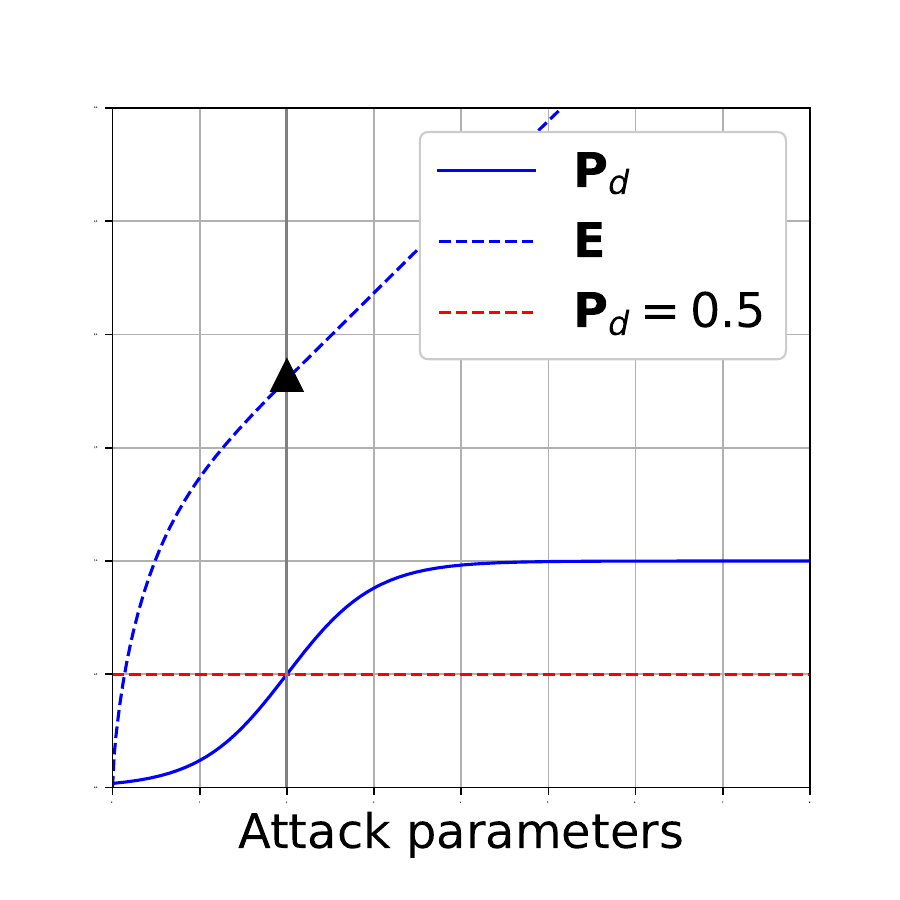}}
  \caption{\subref{fig:concave} $\mathbf{P}_d$ is concave and converging to $1$, \subref{fig:non-convex} $\mathbf{P}_d$ is convex-concave and converging to $1$.}
  \label{fig:PdEdiagram}
  \end{figure}
   \begin{hypothesis}\label{hyp:convexpd}
   $\mathbf{P}_d$ is convex.
   \end{hypothesis}
  \begin{proof}[Under Hypothesis \ref{hyp:convexpd}]
   Given the characteristics listed in Eq.~\eqref{eq:pdcharacter}, $\mathbf{P}_d$ will always be upper-bounded by $1$, therefore $\mathbf{P}_d$ can't be a convex function. \qedhere
   \end{proof}
    \begin{hypothesis}\label{hyp:concavepd}
   $\mathbf{P}_d$ is concave and converging to $1$.
   \end{hypothesis}
  \begin{proof}[Under Hypothesis \ref{hyp:concavepd}]
  A simple demonstration of $\mathbf{P}_d$ and $\mathbf{E}$ is shown in Fig.~\ref{fig:concave}. $\mathbf{P}_d$'s concavity causes it to quickly reach 0.5, limiting optimal attacks to low parameter regimes. While this suggests anomaly detection should be sensitive to minor changes when attack parameters seem insignificant, such high sensitivity is impractical in real radio environments. Measurement errors from legitimate data will also inevitably lead to possible high $p_f$, ultimately bounding performance by the \gls{crlb}, as per Eq.~\eqref{eq:atkeffectivesness2}.
  \end{proof}
  \begin{hypothesis}\label{hyp:vexconcavepd}
   $\mathbf{P}_d$ is convex-concave and converging to $1$.
   \end{hypothesis}
  \begin{proof}[Under Hypothesis \ref{hyp:vexconcavepd}]
  A convex-concave $\mathbf{P}_d$ is depicted in Fig.~\ref{fig:concave}. This design doesn't ensure that optimal attack performance is limited to an insignificant attack parameters regime. In general, such an anomaly detector is more achievable in practice. Given that $p_f$ will be considerably lower, this might result in actual performance superior to when $\mathbf{P}_d$ is concave. Further more, we can determinate the convergence of $\mathbf{E}_d$, according to L'H\^opital's rule, we have,
   \begin{equation*}\label{eq:limits1}
   \begin{split}
   \lim_{A_{t} \to\infty}\frac{\partial\mathbf{E}}{\partial A_{t}} &=  \lim_{A_{t} \to\infty}\frac{\mathbf{E}}{A_{t}} =  1,\\
   \lim_{A_{t} \to\infty}\frac{\partial\mathbf{P}_d}{\partial A_{t}} &= \lim_{A_{t} \to\infty}\frac{\mathbf{P}_d}{A_{t}} =  0,
   \end{split}
   \end{equation*} 
   \begin{equation*}\label{eq:limits2}
   \lim_{A_{t}\to\infty}\big(1-\mathbf{P}_d(A_{t})\big) = 0. 
   \end{equation*}
   Therefore, we now have $\lim_{A_{t} \to\infty}\frac{\partial\mathbf{E}d}{\partial A_{t}} = 0$. This result indicates that under resource constraints and in the presence of anomaly detection, the effectiveness of an attack will always be always converging. \end{proof}
  
    \section{Performance evaluation}\label{sec:SIMUTDAD}
    \begin{table}[!htbp]
		\centering
        \scriptsize
        \vspace{-2mm}
		\caption{Simulation setup 2}
		\label{tab:setup2}
		\begin{tabular}{>{}m{0.2cm} | m{1.6cm} l m{3.7cm}}
			\toprule[2px]
			&\textbf{Parameter}&\textbf{Value}&\textbf{Remark}\\
			\midrule[1px]
			
                &Map Size&$[300,300,10]$& Define map size\\
                
			  &$n$&$100$& number of anchor \gls{uav}s\\
			
			&$\sigma^2_{\mathbf{p},n}$&$\sim\mathcal{U}(0.1,3.0)$& Position error power / \si{\meter^2}\\
			&$V$&$\sim\mathcal{U}(0.3,1.7)$& Travel speed \si{\meter/s}\\ \multirow{-5}{*}{\rotatebox{90}{\textbf{System}}} 
            & $T$ & 15 s &  Simulation time\\

			\midrule[1px]

			&$[\lambda_r,\lambda_p]$ & $[0.2,-0.8]$& Reward and penalty\\
			&$\gamma$ & $0.5$& Forget factor\\

			&$\epsilon^t$ & $0.95$ & Confidence threshold\\
                \midrule[1px]
                \multirow{-6}{*}{\rotatebox{90}{\textbf{TAD}}}&$n_m$ & $33$& Malicious \gls{uav}s\\
                &$r_a$ & $0.7$&Attack rate\\ 
                \multirow{-3}{*}{\rotatebox{90}{\textbf{Attacker}}}
                &$\tilde{\mathrm{B}}_n$& $6$&Position bias \\
			&$\tilde{\sigma}^2_{\mathrm{M},n}$ & $600$& Manipulation index \\
                
            \bottomrule[2px]
		\end{tabular}
	\end{table}
    To evaluate the resilience and effectiveness of \gls{tad}, we conducted simulations within our cooperative localization system, subjecting it to varying attack modes and strategies. \revise{}{In this scenario, a target \gls{uav} navigated through a dynamic environment containing potentially malicious \glspl{uav}. We utilized downsampled \gls{3d} \gls{uav} trajectories based on \cite{traj2025nacar, delmerico2019drone, fonder2019midair}, which provided simulation results that more accurately reflect the complexity of real-world \gls{uav} flight patterns.} We made the assumption that the attacker possesses limited resources, allowing it to selectively hijack specific \glspl{uav}. The orchestrated attacks were designed to be either random or coordinated in a same time frame. The configuration for \gls{magd} adheres to the specifications outlined in Tab.~\ref{tab:setup1}. Parameters for the overall system, TAD, and the attackers are detailed in Tab.~\ref{tab:setup2}. The attack rate was set at 0.7. All anchor UAVs were strategically distributed across the map, each navigating towards random destinations. This resulted in an average of 15 anchor UAVs within a cooperative localization range of $50 $\text{m}. The simulation results are depicted in Fig.~\ref{fig:taddemo}. \revise{}{Our evaluation compares \gls{tad} against a recursive approach \cite{han2024secure}, and two \gls{doc}-based approaches adapted from \cite{DOCwon2019}: \gls{wdoc} and \gls{twdoc}. \gls{wdoc} and \gls{twdoc} differ in that TWDOC filters detected anomalies using a predefined threshold before applying them in localization.}
   \begin{figure}[!t]
   \centering
   \subfigure[\label{fig:trace_mani_bias}]{\centering
   \reviseboxfinal{\includegraphics[width=0.46\columnwidth]{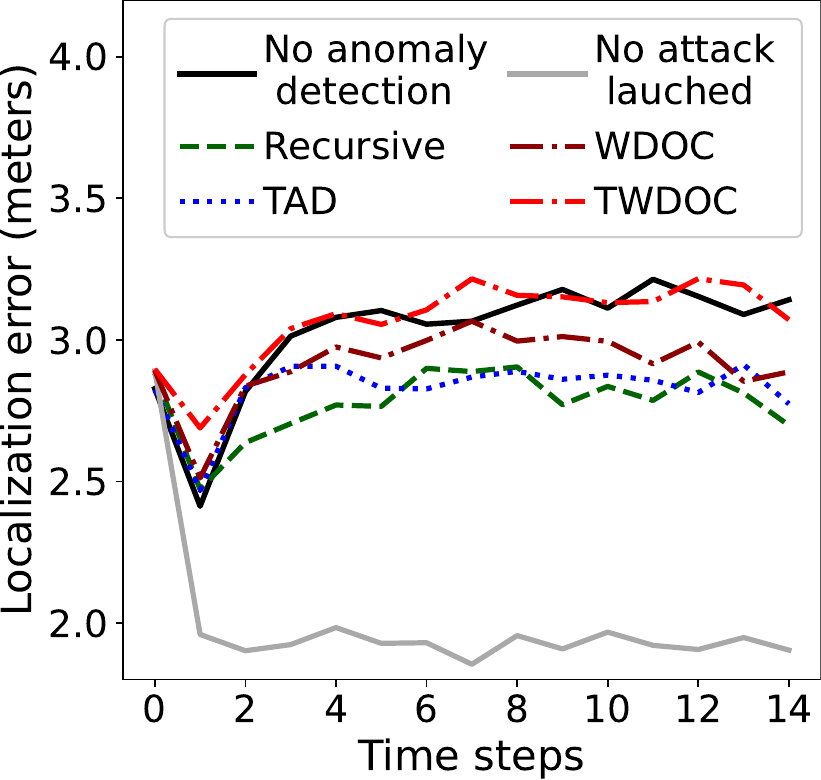}}}
   \hfill
   \subfigure[\label{fig:trace_mani_ran}]{\centering
   \reviseboxfinal{\includegraphics[width=0.456\columnwidth]{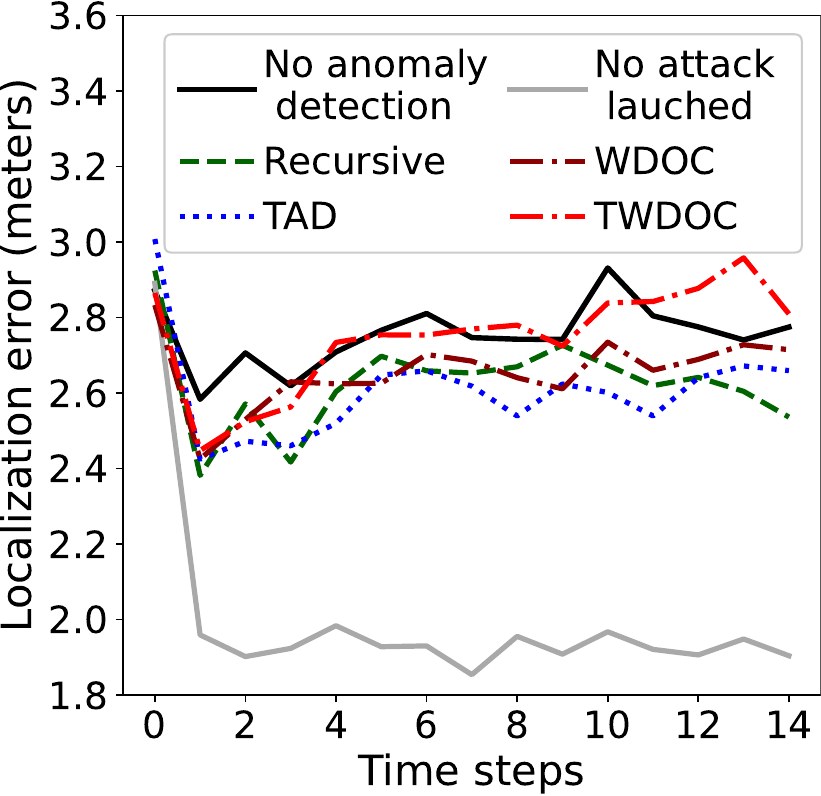}}}
   \\
   \subfigure[\label{fig:trace_bia_coor}]{\centering
   \reviseboxfinal{\includegraphics[width=0.46\columnwidth]{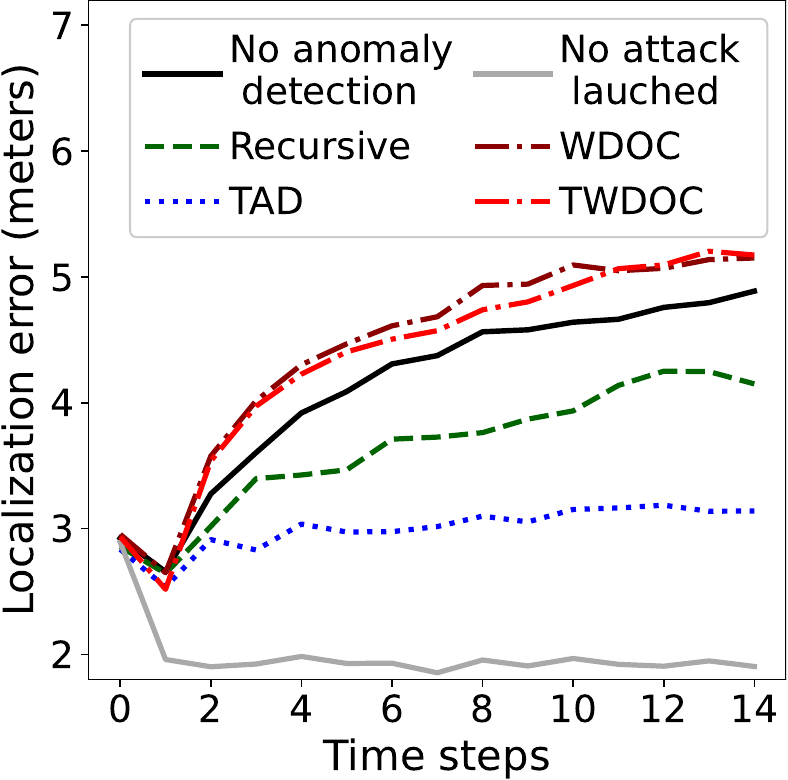}}}
   \hfill
   \subfigure[\label{fig:trace_bia_ran}]{\centering
   \reviseboxfinal{\includegraphics[width=0.472\columnwidth]{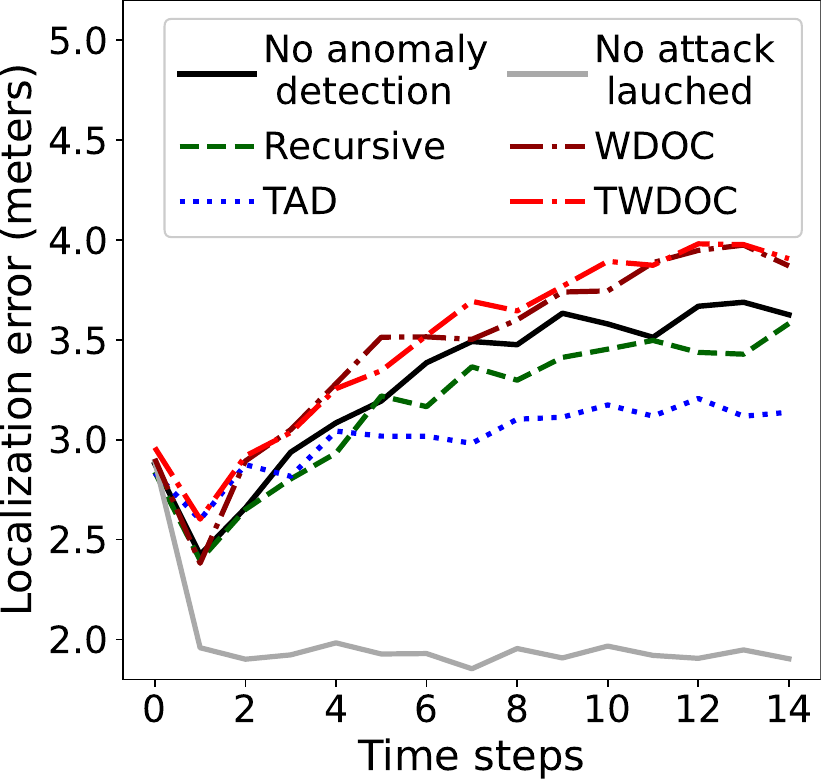}}}
   \revise{}{\vspace{-1mm}\caption{Error development of \subref{fig:trace_mani_bias} coordinated manipulation attack, \subref{fig:trace_mani_ran} random manipulation attack, \subref{fig:trace_bia_coor} coordinated bias attack, and \subref{fig:trace_bia_ran} random bias attack, respectively.}}
   \label{fig:taddemo}
   \end{figure}
   \revise{}{\begin{table}[!htbp]
   \centering
   \revise{}{\vspace{-2mm}\caption{Average localization error (meters) of different anomaly detection approaches}}
   \revisebox{
   \begin{tabular}{|c|c|c|c|c|c|}
   \hline
   &No AD & Recur. & TAD & WDOC & TWDOC \\ \hline
   No attack & 1.99&   -    &    -   &   -    &  -    \\ \hline
   Coor. mani &2.76 & 2.62 & 2.61 & 2.66 & 2.76 \\ \hline
   Ran. mani &3.06 & 2.79 & 2.83 & 2.93 & 3.07  \\ \hline
   Coor. bias &4.12 & 3.63 & 3.00 & 4.44 & 4.38    \\ \hline
   Ran. bias &3.28 & 3.17 & 2.96 & 3.46 & 3.49  \\ \hline
   \end{tabular}\label{tab:performAD}}
   \end{table}}
   
    \revise{}{As shown in Fig.~\ref{fig:taddemo} and Tab.~\ref{tab:performAD}, the performance ranking is $\emph{TAD}>\emph{Recursive}>\emph{WDOC}>\emph{TWDOC}$. DOC-based detection methods have limited effectiveness. While they can partially detect manipulation attacks with large unbiased errors, they perform poorly against subtle bias attacks, actually increasing errors compared to no detection. Their performance further degrades in UAV scenarios due to meter-level anchor position errors, a far cry from the negligible errors in their original wireless sensor network application. Additionally, \gls{twdoc}'s practice of excluding anomalies proves counterproductive against manipulation attacks where unbiased errors naturally cancel out, though this effect is less significant for bias attacks. The recursive approach relies on the first localization for anomaly detection. However, when the attack effect is significant at each time step, the recursive approach can be penetrated step by step. While its performance may match \gls{tad} in some cases, it requires running the localization algorithm twice, making it far more computationally expensive.}
    \begin{figure}[!t]
    \centering
   \subfigure[\label{fig:detec_bias}]{\centering
   \reviseboxfinal{\includegraphics[width=0.466\columnwidth]{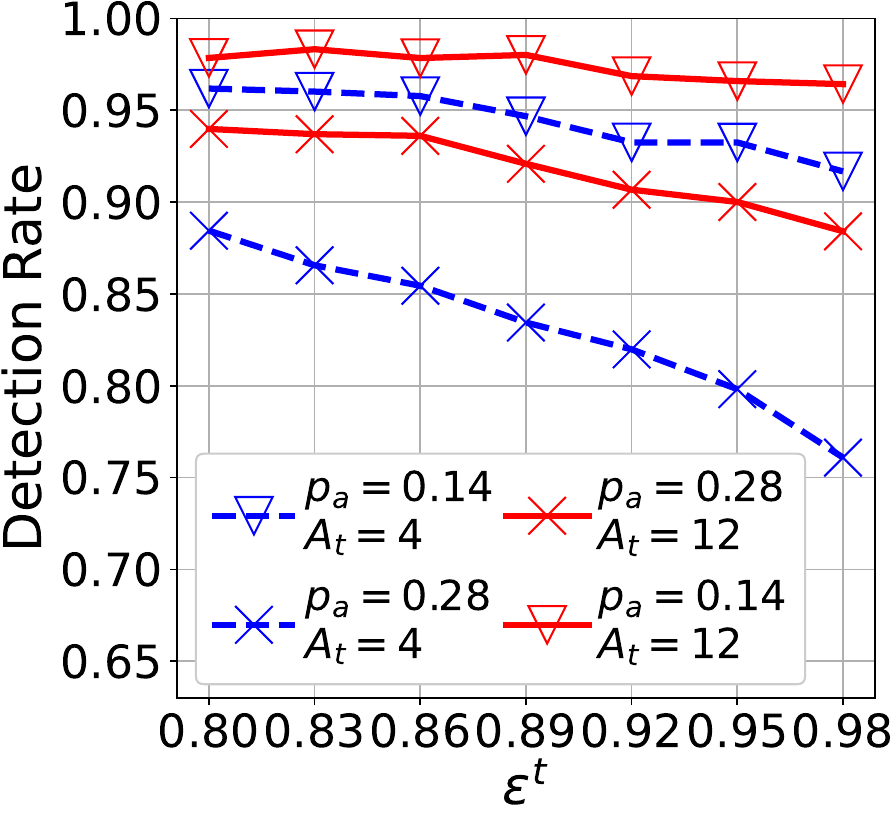}}}
   \hfill
   \subfigure[\label{fig:detec_mani}]{\centering
   \reviseboxfinal{\includegraphics[width=0.466\columnwidth]{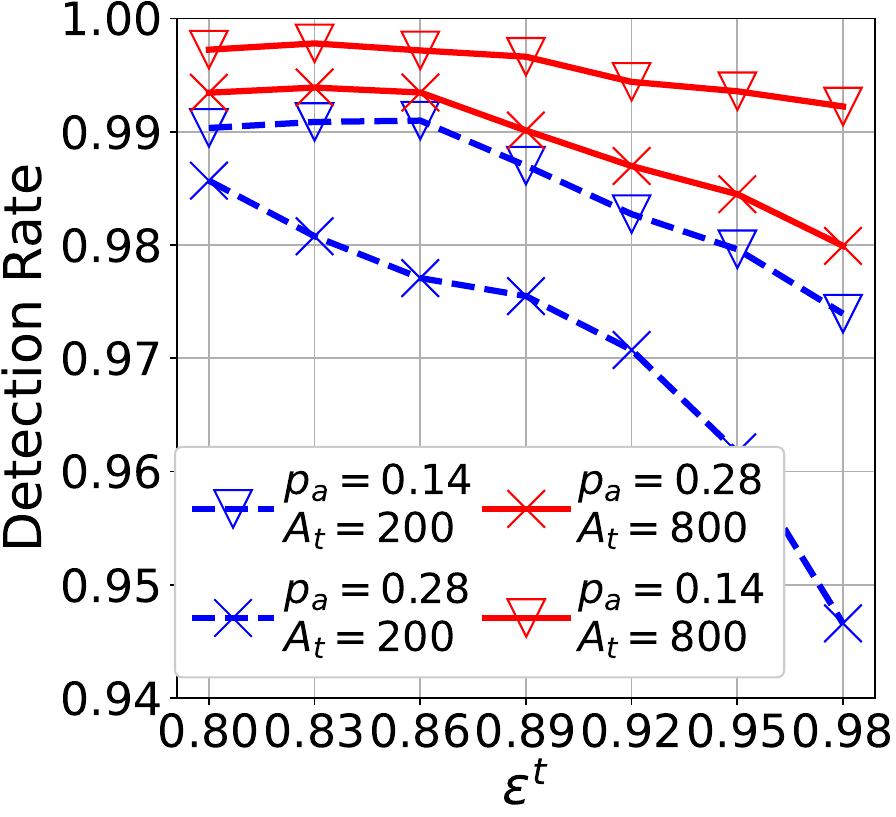}}}
   \\
   \subfigure[\label{fig:false_bias}]{\centering
   \reviseboxfinal{\includegraphics[width=0.466\columnwidth]{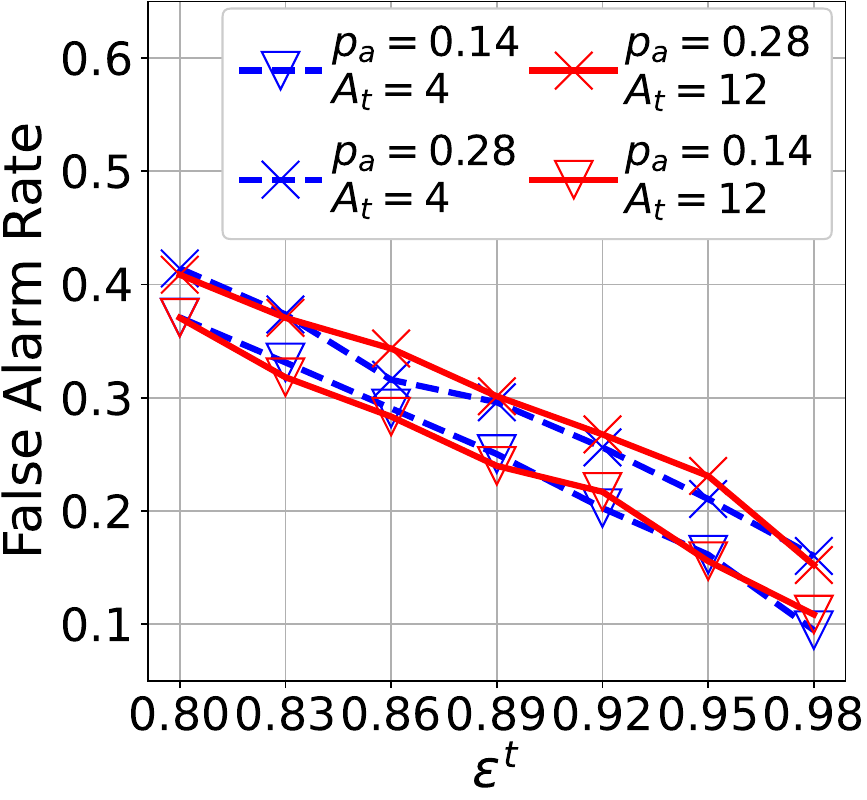}}}
   \hfill
   \subfigure[\label{fig:false_mani}]{\centering
   \reviseboxfinal{\includegraphics[width=0.466\columnwidth]{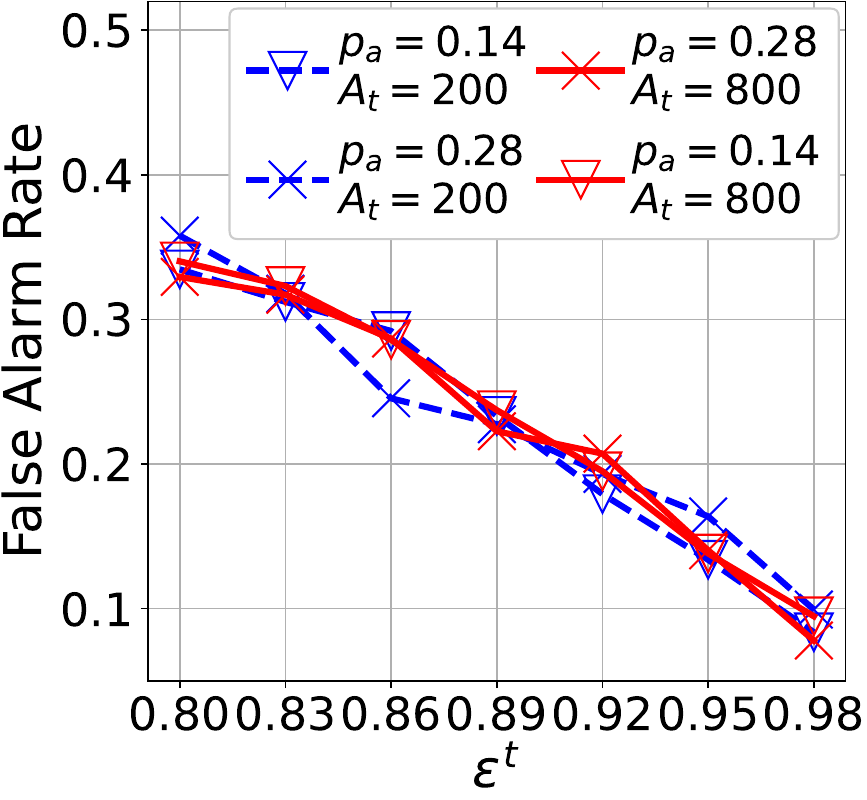}}}
   \revise{}{\vspace{-1mm}\caption{Sensitivity test of \gls{tad}: \subref{fig:detec_bias} detection rate under bias attack, \subref{fig:detec_mani} detection rate under manipulation attack, \subref{fig:false_bias} false alarm rate under bias attack, \subref{fig:false_mani} false alarm rate under manipulation attack.}}
   \label{fig:sensi_test}
   \end{figure}

   \revise{}{As described in Sub.~\ref{subsec:reconatk}, the anomaly detector's detection rate can be characterized by $\mathbf{P}_d(p_a,A_t)$. While the $\mathbf{P}_d$ of \gls{tad} lacks mathematical tractability, we validate its sensitivity through numerical simulations. Simulations explore parameter combinations with $\epsilon^e$ ranging from $0.80$ to $0.98$, while testing different attack scenarios with $p_a = {0.14, 0.28}$ and $A_t = {200, 800}$. It is generally observed that as $\epsilon^e$ decreases, \gls{tad} becomes more sensitive to attacks, though this comes at the cost of a drastically increased false alarm rate. The detection rate patterns across different attack scenarios support our previously described characteristics of $\mathbf{P}_d$, specifically that $\frac{\partial \mathbf{P}_d}{\partial p_a} < 0$ and $\frac{\partial \mathbf{P}_d}{\partial A{t}} > 0$. Moreover, the false alarm rate $\mathbf{P}_f$ exhibits strong dependency on $p_a$ while under bias attacks but remaining independent of $A_{t}$. }
    
    \revise{}{With increasing $\epsilon^t$ and all other parameters held constant, it is apparent that $\mathbf{E}_d$ is monotonically increasing, as per Eq.~(\ref{eq:atkeffectivesness1}), while $\mathrm{CRLB}\big(N(1-p_a\mathbf{P}_d-\mathbf{P}_f)\big)$ is monotonically decreasing. Therefore, an optimal $\epsilon^t$ exists as per Eq.~(\ref{eq:atkeffectivesness2}). Given that the error power under attacks, $\sigma'_{M}$, is not deducible, $\mathbf{E}_d$ of different attacks is consequently not mathematically traceable.} \revise{}{Firstly, we evaluate $\mathbf{E}_d$ with varying attack parameters, results are obtained from $2000$ Monte Carlo iterations. The simulation results presented in Figs.~\ref{fig:biaspara}--\ref{fig:mani_para} corroborate our analytical previous findings, demonstrating that attack performance consistently converges to a threshold under anomaly detection as the attack parameters increase.} However, demonstrating optimal attack parameters in the simulation proves challenging, as it's not possible to iterate through all possible attack parameters due to computational constraints. \revise{}{} A comparison \reviseprev{of}{between} Figs.~\ref{fig:biaspara} and \ref{fig:mani_para} reveals distinct patterns in the effectiveness of attacks across different modes. The bias mode exhibits a more linear-like (compared to mani- mode) response to changes in attack parameters. In contrast, the manipulation mode demonstrates a gradual convergence trend as attack parameters increase. It's important to note that no ideal anomaly detector achieves $\lim_{A_{t}\to\infty}\mathbf{P}_d(A_{t}) = 1$. This implies a breakthrough point exists where attack effectiveness surpasses the threshold. However, such extreme attacks are easily identifiable and can be countered by simply discarding the aberrant data. \revise{}{Secondly, we investigate whether there exists an optimal $\epsilon^t$ that minimizes $\mathbf{E}_d$ under the given attack parameters. The simulation results, presented in Fig.~\ref{fig:ERRinopt_ep}, compare $\mathbf{E}_d$ against two baselines: one without any attacks and another without \gls{tad}. The results show that $\mathbf{E}_d$ initially decreases due to the reduction in false positive rate $\mathbf{P}_f$, then subsequently increases due to the decreased detection rate $\mathbf{P}_d$. An optimal perturbation threshold $\epsilon^t$ is found to exist within the range $[0.89, 0.95]$. }
    \vspace{-3mm}
    \begin{figure}[!htbp]
		\centering
        \subfigure[Attack effectiveness under bias attack without anomaly detection or with \gls{tad}\label{fig:biaspara}]{
            \includegraphics[width=0.75\linewidth]{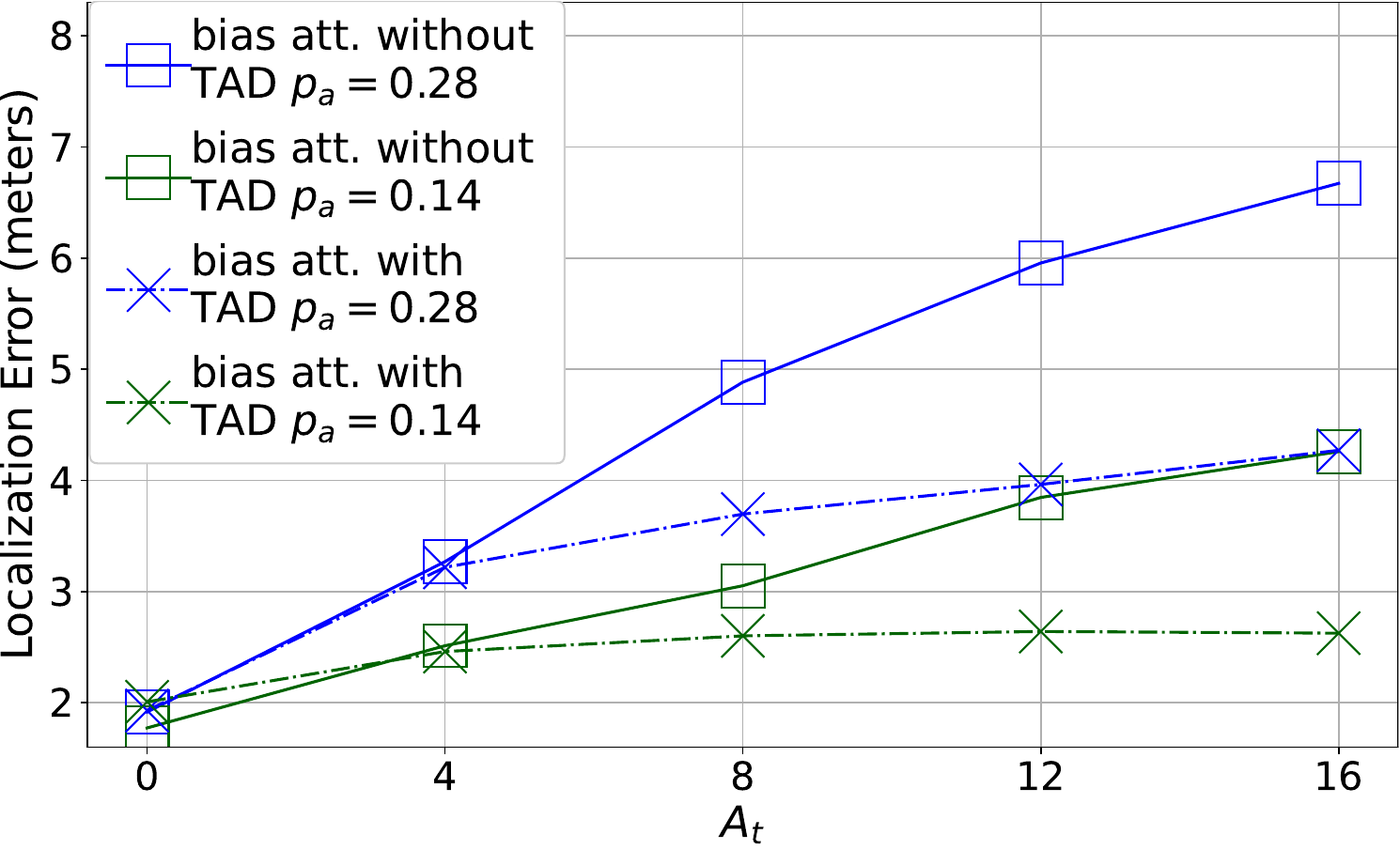}}
        \subfigure[Attack effectiveness under manipulation attack without anomaly detection or with \gls{tad}\label{fig:mani_para}]{
		      \includegraphics[width=0.75\linewidth]{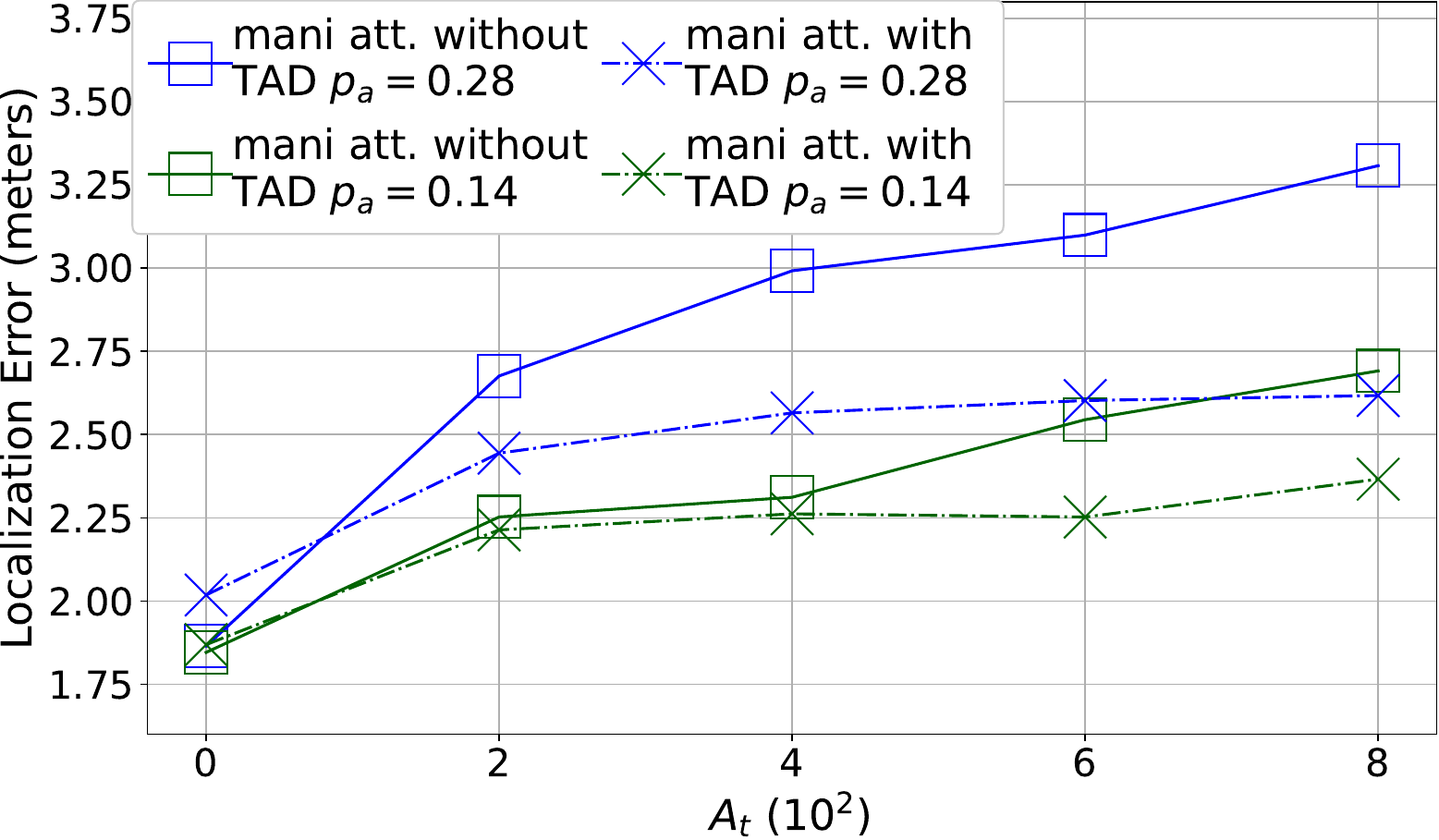}}
        \vspace{-1mm}
		\caption{Attack effectiveness evaluation}
		\label{fig:attkeffi}
	\end{figure}  
    \begin{figure}[!htbp]
		\centering
		\revisebox{\includegraphics[width=0.78\linewidth]{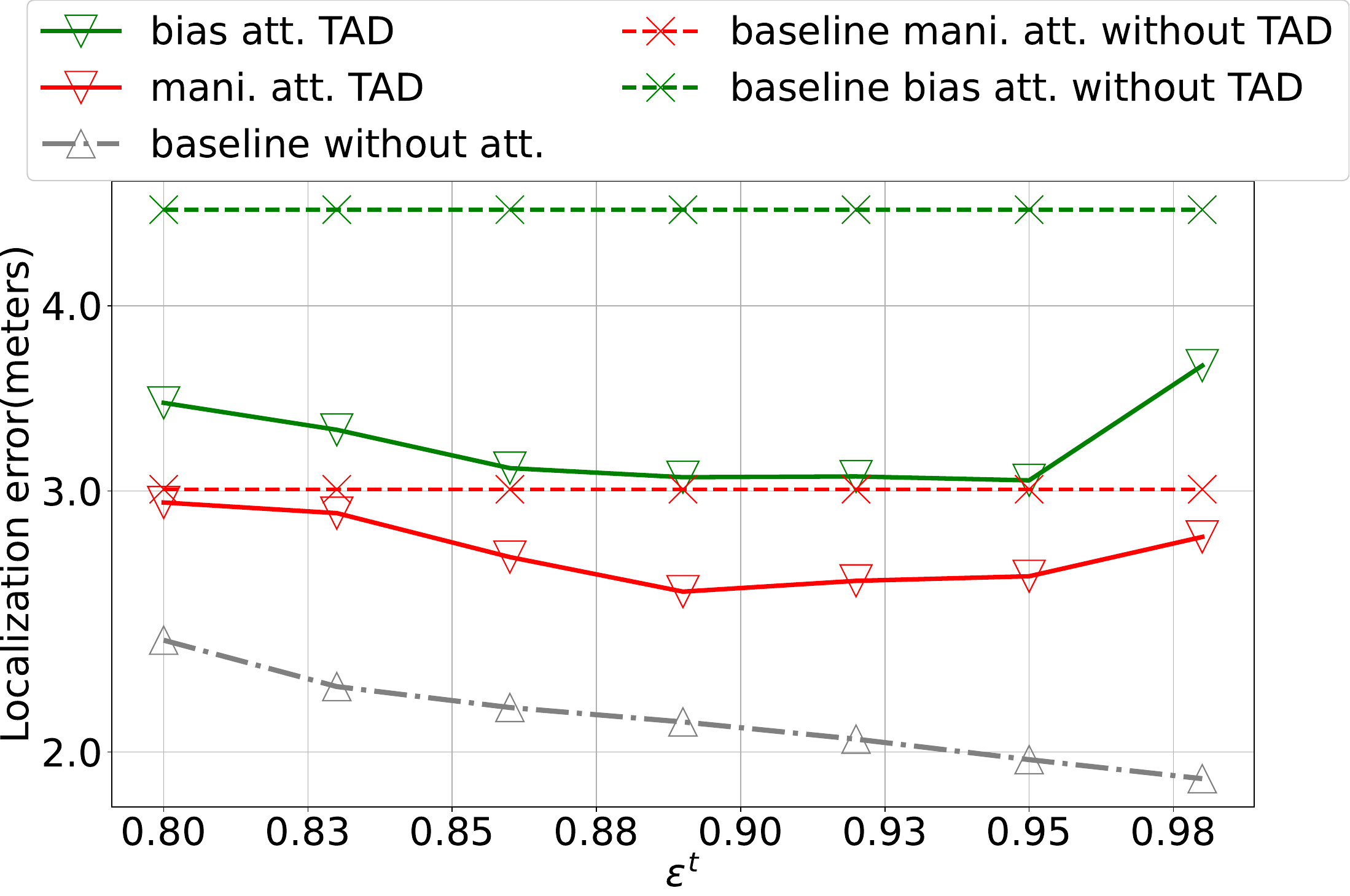}}
        \vspace{-3mm}
		\revise{}{\caption{Localization error with varying $\epsilon^t$. For manipulation attack $A_t = 800$, $p_a = 0.14$; for bias attack $A_t = 8$, $p_a = 0.14$.}}
		\label{fig:ERRinopt_ep}
	\end{figure}
    \begin{table}[!htbp]
   \centering
    \vspace{-1mm}
   \revise{}{\caption{Average localization error (meters)}}
   \revisebox{\begin{tabular}{|c|c|}
    \hline
    No attacks&1.60\\ \hline
    Without TAD under attacks & 9.42 \\ \hline
    TAD under attacks & 2.12 \\ \hline
    TAD and RP & 1.82 \\ \hline
    TAD and RP while under fals. repu. & 1.78 \\ \hline
    \end{tabular}}
   \label{tab:performTP}
   \end{table}
   
   We then conducted simulations to evaluate the performance of both \gls{tad} and \gls{rp}, involving malicious \glspl{uav} employing a coordinated stalking strategy to target the victim. The parameters for \gls{magd} and \gls{tad} adhere to those outlined in Tab.~\ref{tab:setup1} and Tab.~\ref{tab:setup2}, with the attack mode specifically set to bias mode. Ten \glspl{uav} uploaded their local reputations to the cloud. The percentage of malicious \glspl{uav} was set at $30\%$, comprising three \glspl{uav} capable of attacking others while sharing falsified reputations. We assumed the attacker possesses knowledge of the victim target's actual position, albeit with some ambiguity. (The estimated position $\hat{\mathbf{p}}_k$, shared by $u_k$, was avoided as it can be misleading once attacks become effective.) Fig.~\ref{fig:ERRinTP} illustrates the average errors obtained from $100$ simulations at different time steps, \revise{}{while Tab.\ref{tab:performTP} presents the consolidated average localization error across all $100$ time steps.} Our simulation results demonstrated the effectiveness of \gls{tad} in mitigating stalking attack strategies at the specified attack density, resulting in a significant reduction in localization errors compared to scenarios without \gls{tad}. Moreover, the incorporation of RP improved localization performance, accelerating error convergence more than \gls{tad} alone, especially after 40 time steps. Notably, RP also showed resilience against falsified reputation information.
    \begin{figure}[!htbp]
		\centering
		\revisebox{\includegraphics[width=0.75\linewidth]{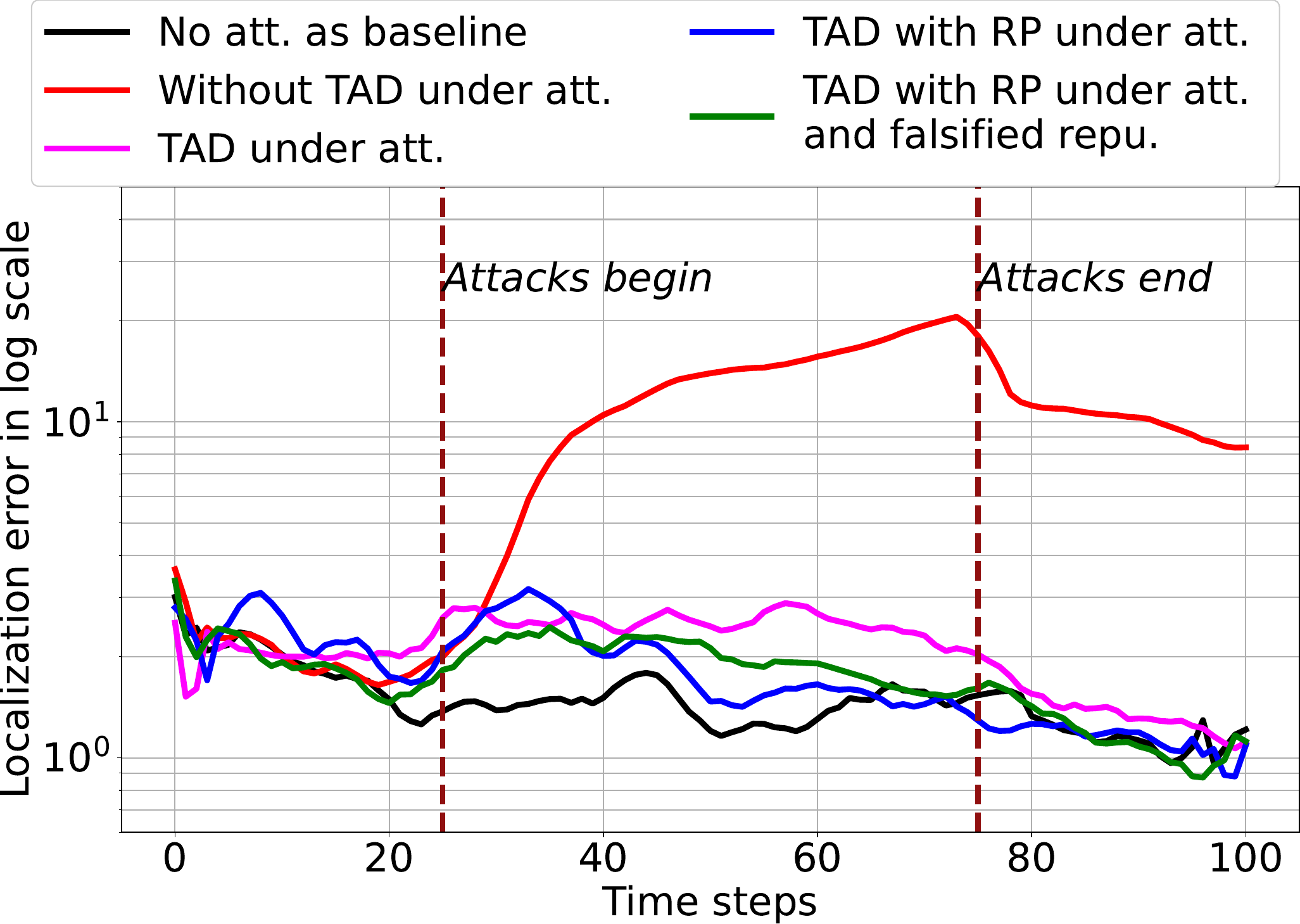}}\vspace{-3mm}
		\revise{}{\caption{Localization error at different time step (for the baseline, UAVs are set to follow the target as well)}}
		\label{fig:ERRinTP}
	\end{figure} 
     
  \section{Conclusion}\label{conclu}
  In this paper, we \reviseprev{}{have} derived the \gls{3d} \gls{crlb} for cooperative localization through geometric interpretations and conducted an in-depth exploration of localization in non-uniform spatial distributions. We aspire that these theoretical findings offer intuitive insights into \gls{3d} localization. Moreover, we \reviseprev{unveiled a novel approach (\gls{magd}) designed to seamlessly adapt to}{have adopted the \gls{magd} approach to address} the dynamic mobility and \reviseprev{fluctuating}{} availability of anchor \gls{uav}s. In addition, we analyzed the impact of falsified information on cooperative localization and proposed adaptive anomaly detection strategies (\gls{tad} and \gls{rp}) to enhance resilience.

    \reviseprev{However, it's important to note}{Yet it shall be noted} that our study \reviseprev{didn't}{does not} extensively delve into potential attacks against \reviseprev{\gls{rp}}{the \gls{rp} framework itself}. An adept adversary \reviseprev{might}{may} manipulate a subset of compromised \reviseprev{\acp{uav}}{nodes} to launch targeted attacks on specific \acp{uav}, while another subset uploads falsified reputation. This presents a challenge for both our \gls{tad} and \gls{rp}, requiring effective detection and mitigation strategies. \reviseprev{To counteract this potential threat, there is a pressing need for a novel approach to identify emerging attack patterns.}{The pressing need for novel approaches to resolve this potential threat is leading us to future studies.}
 


  
  \appendix
  \section*{Proof of Theorem~\ref{theorem:crlb}}\label{appendix:Theorem_proof}
  \begin{proof}
  When represented as the Cayley–Menger determinant, the volume of a tetrahedron formed by the anchors $u_n$, $u_m$, $u_l$, \revise{}{and the target $u_k$} can be denoted \revise{}{as:} 
  \begin{equation}\label{eq:3Dxyz}
  V_{n,m,l}(x,y,z) = \frac{1}{6} \begin{vmatrix} x_k-x_n&y_k-y_n&z_k-z_n \\
  x_k-x_m&y_k-y_m&z_k-z_m \\x_k-x_l&y_k-y_l&z_k-z_l \\
  \end{vmatrix}.
  \end{equation}
  The determinant of \gls{fim}  being provided in Eq.~\eqref{eq:fim3D} can be therefore rewritten as
  \begin{align}\label{deterFIM}
  {|\mathbf{I}(h)|} = \frac{6}{(\sigma^2_M)^3}\sum_{n=1}^N\sum_{m=1}^N\sum_{l=1}^N\frac{V_{n,m,l}^2(x,y,z)}{(d_{k,n} d_{k,m}d_{k,l})^2}.
  \end{align}
  Meanwhile, $adj\{\mathbf{I}(h)\}$ takes the form in Eq.~\eqref{eq:adjfim3D}. 
  \begin{figure*}[t]
  \begin{align}\label{eq:adjfim3D}
  adj\{\mathbf{I}(h)\} = \frac{2}{(\sigma_M^2)^2} \begin{bmatrix}
  \sum\limits_{n=1}^{N}\sum\limits_{m=1}^{N} \frac{V_{n,m}^2(y,z)}{(d_{k,m}d_{k,n})^2} & \sum\limits_{n=1}^{N}\sum\limits_{m=1}^{N} \frac{V_{n,m}(x,z)V_{n,m}(y,z)}{(d_{k,m}d_{k,n})^2} & \sum\limits_{n=1}^{N}\sum\limits_{m=1}^{N} \frac{V_{n,m}(x,y)V_{n,m}(y,z)}{(d_{k,m}d_{k,n})^2} \\
  \sum\limits_{n=1}^{N}\sum\limits_{m=1}^{N} \frac{V_{n,m}(x,z)V_{n,m}(y,z)}{(d_{k,m}d_{k,n})^2} & \sum\limits_{n=1}^{N}\sum_{m=1}^{N} \frac{V_{n,m}^2(x,z)}{(d_{k,m}d_{k,n})^2} & \sum_{n=1}^{N}\sum\limits_{m=1}^{N} \frac{V_{n,m}(x,y)V_{n,m}(x,z)}{(d_{k,m}d_{k,n})^2} \\
  \sum\limits_{n=1}^{N}\sum\limits_{m=1}^{N} \frac{V_{n,m}(x,y)V_{n,m}(y,z)}{(d_{k,m}d_{k,n})^2} & \sum\limits_{n=1}^{N}\sum\limits_{m=1}^{N} \frac{(V_{n,m}(x,y)V_{n,m}(x,z)}{(d_{k,m}d_{k,n})^2} & \sum\limits_{n=1}^{N}\sum\limits_{m=1}^{N} \frac{V_{n,m}^2(x,y)}{(d_{k,m}d_{k,n})^2}
  \end{bmatrix}
  \end{align}
  \rule{\textwidth}{0.4pt}
\end{figure*}
  Recalling Eq\reviseprev{}{s}.~\eqref{eq:CRLB}--\eqref{eq:ivFIM}, we now have the \gls{crlb} \reviseprev{in the Eq.~\eqref{eq:CRLBclosed}}{}  
  \begin{align}\label{eq:CRLBclosed}
    \sigma_{p}^2 &\geqslant \frac{\sigma_M^2}{3}\frac{\sum_{n=1}^{N}\sum_{m=1}^{N-1}\frac{V_{n,m}^2(x,y)+V_{n,m}^2(y,z)+V_{n,m}^2(x,z)}{(d_{k,n}d_{k,m})^2}}{\sum_{n=1}^N\sum_{m=1}^N\sum_{l=1}^N\frac{V_{n,m,l}^2(x,y,z)}{(d_{k,n}d_{k,m}d_{k,l})^2}},
  \end{align}
  where $V_{n,m}(x,y)$, $V_{n,m}(x,z)$ and $V_{n,m}(y,z)$ are the area of the triangle formed by $u_k$, $u_n$ and $u_m$ being projected onto a Cartesian plane, as described in Eqs.~(\ref{eq:2Dxy})-(\ref{eq:2Dyz}) and depicted in Fig.~\ref{fig:V3D}.
  \begin{align}\label{eq:2Dxy}
  V_{n,m}(x,y) = \frac{1}{2} \begin{vmatrix} x_k-x_n&y_k-y_n \\
  x_k-x_m&y_k-y_m
  \end{vmatrix}
  \end{align}    
  \begin{align}\label{eq:2Dxz}
  V_{n,m}(x,z) = \frac{1}{2} \begin{vmatrix} x_k-x_n&z_k-z_n \\
  x_k-x_m&z_k-z_m
  \end{vmatrix}
  \end{align}  
  \begin{align}\label{eq:2Dyz}
  V_{n,m}(y,z) = \frac{1}{2} \begin{vmatrix} y_k-y_n&z_k-z_n \\
  y_k-y_m&z_k-z_m
  \end{vmatrix}
  \end{align} 

  Eq.~\eqref{eq:CRLBclosed} can be further decomposed:
  \begin{align}\label{eq:CRLBclosed2}
    \sigma_{p}^2 &\geqslant \frac{\sigma_M^2}{3}\cdot\frac{f_1}{f_2}.
  \end{align}
  \begin{align}
    f_1 &= \sum_{n=1}^{N}\sum_{m=1}^{N}\frac{V_{n,m}^2(x,y)+V_{n,m}^2(y,z)+V_{n,m}^2(x,z)}{(d_{k,n}d_{k,m})^2},\label{eq:fdn}\\
    f_2 &= \sum_{n=1}^N\sum_{m=1}^N\sum_{l=1}^{N}\frac{V_{n,m,l}^2(x,y,z)}{(d_{k,n}d_{k,m}d_{k,l})^2}.\label{eq:fdn2}
  \end{align}
  This derivation follows the work introduced in \cite{app13032008}, we extend it by considering $N$ is large enough and $u_n$, $u_m$, and $u_l$ are uniformly distributed around $u_k$, thus $\sum_{n=1}^Nd_{k,n} = \sum_{n=1}^Nd_{k,m} = \sum_{n=1}^Nd_{k,l} = N\Bar{d}$, while $\Bar{d}$ is the average distance. Now, for a random triangle formed by $u_k$, $u_n$, and $u_l$, denoted as $\triangle U_{k,n,m}$, with an area $\mathcal{A}$. The areas of its projections onto the three Cartesian planes are denoted as $\alpha\mathcal{A}$, $\beta\mathcal{A}$, and $\gamma\mathcal{A}$. There is
  \begin{align}\label{eq:constraint}
  \begin{split}
    \alpha\mathcal{A} + \beta\mathcal{A} + \gamma\mathcal{A} = \mathcal{A}.
  \end{split}
  \end{align} 
  We can get $\alpha + \beta + \gamma = 1$, and $\alpha^2 + \beta^2 +\gamma^2 \leqslant 1$ due to the Cauchy-Schwarz inequality. Now $f_1, f_2$ can be rewritten as
  \begin{equation}\label{eq:CRLBclosedf1}
    f_1  = N^2\frac{(\alpha^2+\beta^2+\gamma^2)\mathcal{A}^2}{(\Bar{d})^4},\quad
    f_2 = N^3\frac{(\bar{d_H})^2\mathcal{A}^2}{9(\Bar{d})^6},
  \end{equation}
  where $\bar{d_H}$ is the average height of the tetrahedron (depicted in Fig.~\ref{fig:V3D}). Given that the average distance $d_{k,l}=\Bar{d}$, thus $\bar{d_H} = \mathbb{E}\{|\Bar{d} \cos \theta|\}$, where $\theta$ is the angle between the height and $d_{k,l}$. With $u_l$ uniformly distributed around $u_k$, we have $(\bar{d_H})^2 = {\Bar{d}}^2 \mathbb{E}\{\cos^2 \theta\}$ and $\mathbb{E}\{\cos^2 \theta\} = \frac{1}{2}$, thus $\bar{d_H} = \frac{\sqrt{2}}{2}\Bar{d}$. Summarizing Eq.~(\ref{eq:CRLBclosedf1}), the \gls{crlb} can be reformulated:
  \begin{equation}\label{eq:CRLBclosed4}
    \sigma_{p}^2 \geqslant \frac{\sigma_M^2}{3}\cdot\frac{f_1}{f_2}
    = 6\sigma_M^2\frac{(\alpha^2+\beta^2+\gamma^2)}{N}.\\
  \end{equation}
\begin{figure}[!h]
  \centering
  \includegraphics[width=0.6\linewidth]{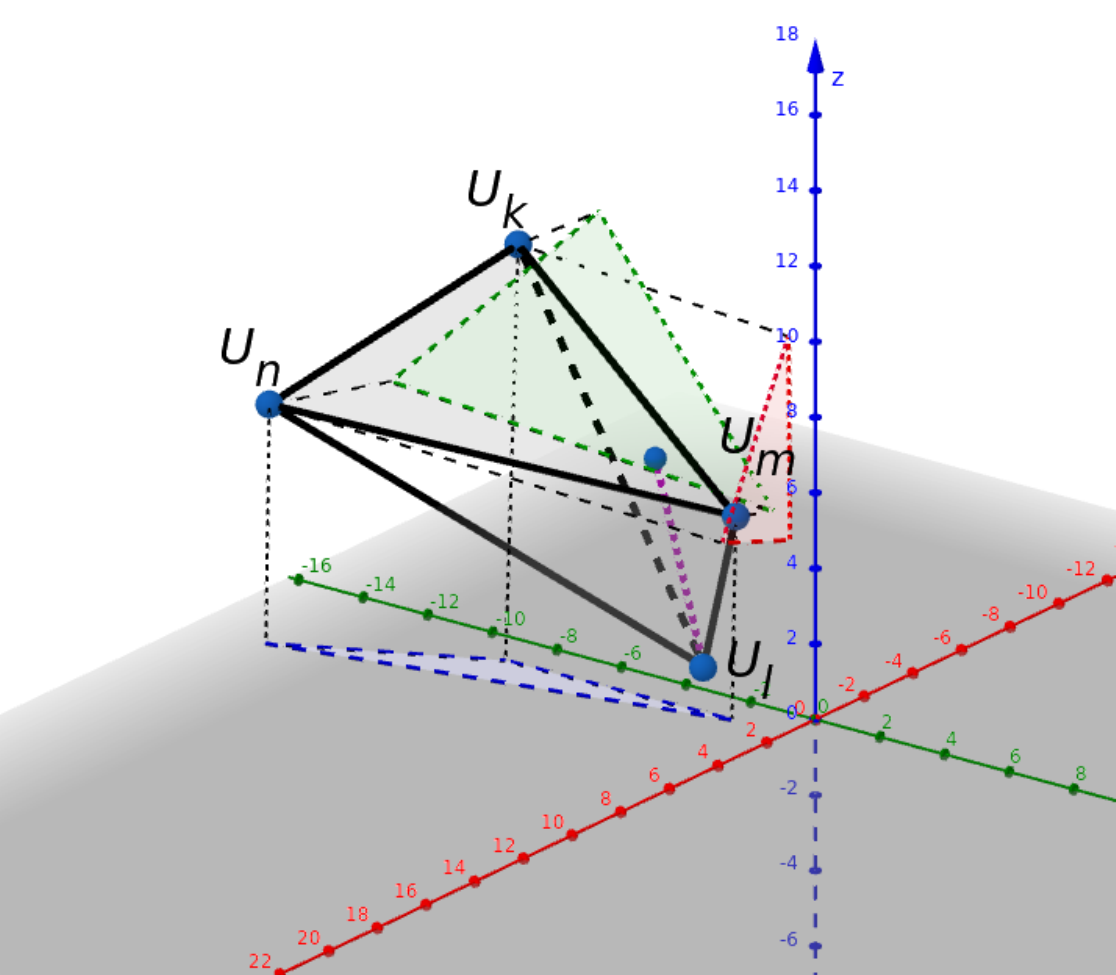}
  \caption{Geometric interpretation of 3D CRLB. Green, blue and red triangles are the projections on Cartesian planes; The purple dashed line is the height of tetrahedron with respect to $\triangle U_{k,n,m}$.}
  \label{fig:V3D}
  \end{figure}
  Leveraging $\alpha^2 + \beta^2 +\gamma^2 \leqslant 1$, we have
  \begin{equation}\label{eq:CRLBclosed41}
  \sigma_{p}^2 \geqslant \frac{6\sigma_M^2}{N}.
  \end{equation}
\end{proof}

\begin{IEEEbiography}[{\includegraphics[width=1in,height=1.25in,clip,keepaspectratio]{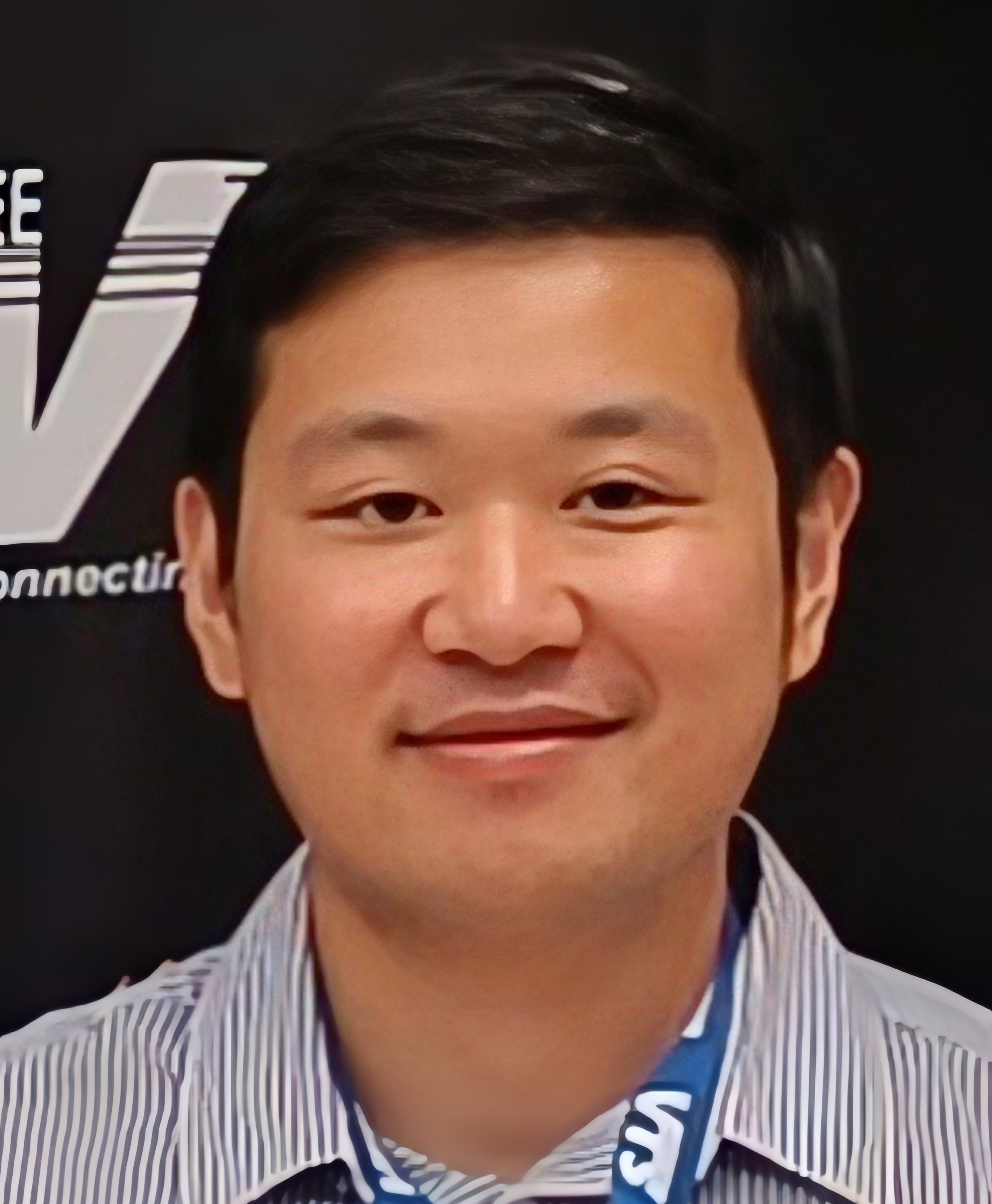}}]{Zexin Fang} received the B.E. degree from Hu Bei University, China, in 2018, the M.Sc. degree from University of Kaiserslautern-Landau (RPTU), Germany, in 2022. He is currently working toward the Ph.D. degree with RPTU. His research interests include localization, network security, and trustworthy AI.
\end{IEEEbiography}

\begin{IEEEbiography}[{\includegraphics[width=1in,height=1.25in,clip,keepaspectratio]{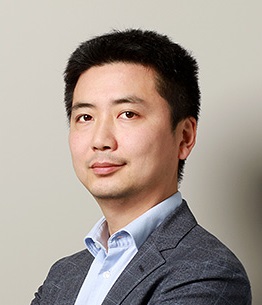}}]{Bin Han} received the B.E. degree in 2009 from Shanghai Jiao Tong University, China, the M.Sc. degree in 2012 from Darmstadt University of Technology, Germany, and the Ph.D. (Dr.-Ing.) degree in 2016 from Karlsruhe Institute of Technology, Germany. He joined University of Kaiserslautern-Landau (RPTU) in July 2016, working as Postdoctoral Researcher and Senior Lecturer. In November 2023, he was granted the teaching license (Venia Legendi) by RPTU, and there with became a Privatdozent. He is the author of two books, six book chapters, and over 80 research papers. He has participated in multiple EU FP7, Horizon 2020, and Horizon Europe research projects. He is an Editorial Board Member for Network, and has served in the organizing committees and/or TPCs of IEEE GLOBECOM, IEEE ICC, EuCNC, European Wireless, and ITC. He is actively involved in the IEEE Standards Association Working Groups P1955, P2303, P3106, and P3454.
\end{IEEEbiography}

\begin{IEEEbiography}[{\includegraphics[width=1in,height=1.25in,clip,keepaspectratio]{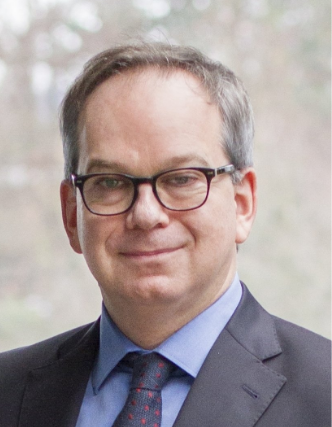}}]{Hans D. Schotten} received his Ph.D. from RWTH Aachen University in 1997. He worked for Ericsson (1999-2003) and Qualcomm (2003-2007). In 2007, he became a full professor at the University of Kaiserslautern. Since 2012, he has also served as scientific director at the German Research Center for Artificial Intelligence (DFKI). He was dean of Electrical Engineering (2013-2017) and became chairman of the German Society for Information Technology in 2018. Professor Schotten has authored over 300 papers and contributed to more than 50 research projects.
\end{IEEEbiography}

\end{document}

%% file: glossary.tex
\makeglossaries
\newacronym{uav}{UAV}{Unmanned Aerial Vehicle}
\newacronym{magd}{MAGD}{Mobility Adaptive Gradient Descent}
\newacronym{tad}{TAD}{Time-evolving Anomaly Detection}
\newacronym{gps}{GPS}{Global Positioning System}
\newacronym{tof}{ToF}{Time of Flight}
\newacronym{rssi}{RSSI}{Received Signal Strength Indicator}
\newacronym{mle}{MLE}{Maximum Likelihood Estimator}
\newacronym{rp}{RP}{Reputation Propagation}
\newacronym{ls}{LS}{Least Square}
\newacronym{wls}{WLS}{Weighted Least Square}
\newacronym{swls}{SWLS}{Secure Weighted Least Square}
\newacronym{ln-1}{LN-1}{$l$1 Norm}
\newacronym{gd}{GD}{Gradient Descend}
\newacronym{admm}{ADMM}{Alternating
Direction Method of Multiplier}
\newacronym{pdf}{PDF}{Probability Density Function}
\newacronym{mse}{MSE}{Mean Squared Error}
\newacronym{fim}{FIM}{Fisher Information Matrix}
\newacronym{crlb}{CRLB}{Cramer-Rao Lower Bound}
\newacronym{3gpp}{3GPP}{Third Generation Partnership Project}
\newacronym{sgd}{SGD}{Stochastic Gradient Descent}
\newacronym{cdf}{CDF}{Cumulative Distribution Function}
\newacronym{c2}{C2}{Command and Control}
\newacronym{ue}{UE}{User Equipment}
\newacronym{2d}{2D}{Two-dimensional}
\newacronym{3d}{3D}{Three-dimensional}
\newacronym{ris}{RIS}{Reconfigurable Intelligent Surface}
\newacronym{aoa}{AoA}{Angle of Arrival}
\newacronym{nlos}{NLoS}{None-Line-of-Sight}
\newacronym{los}{LoS}{Line-of-Sight}
\newacronym{tdoa}{TDOA}{Time Difference of Arrival}
\newacronym{doa}{DOA}{Direction of Arrival}
\newacronym{bs}{BS}{Base Station}
\newacronym{rf}{RF}{Radio Frequency}
\newacronym{arhs}{AHRS}{Attitude and Heading Reference Systems}
\newacronym{b5g}{B5G}{Beyond-Fifth-Generation}
\newacronym{6g}{6G}{Sixth Generation}
\newacronym{doc}{DOC}{Degree of Consistency}
\newacronym{wdoc}{WDOC}{Weighted Degree of Consistency Approach}
\newacronym{twdoc}{TWDOC}{Thresholded Weighted Degree of Consistency Approach}
\newacronym{fista}{FISTA}{Fast Iterative Shrinkage-Thresholding
Algorithm}
\newacronym{lp}{LP}{Linear Programming}
\newacronym{lts}{LTS}{Least Trimmed Square algorithm}
\newacronym{gta}{GTA}{Game Theoretic Aggregation algorithm}
\newacronym{glrt}{GLRT}{Generalized Likelihood Ratio Test}

%% file: main.bbl
\begin{thebibliography}{10}
  \providecommand{\url}[1]{#1}
  \csname url@samestyle\endcsname
  \providecommand{\newblock}{\relax}
  \providecommand{\bibinfo}[2]{#2}
  \providecommand{\BIBentrySTDinterwordspacing}{\spaceskip=0pt\relax}
  \providecommand{\BIBentryALTinterwordstretchfactor}{4}
  \providecommand{\BIBentryALTinterwordspacing}{\spaceskip=\fontdimen2\font plus
  \BIBentryALTinterwordstretchfactor\fontdimen3\font minus
    \fontdimen4\font\relax}
  \providecommand{\BIBforeignlanguage}[2]{{%
  \expandafter\ifx\csname l@#1\endcsname\relax
  \typeout{** WARNING: IEEEtran.bst: No hyphenation pattern has been}%
  \typeout{** loaded for the language `#1'. Using the pattern for}%
  \typeout{** the default language instead.}%
  \else
  \language=\csname l@#1\endcsname
  \fi
  #2}}
  \providecommand{\BIBdecl}{\relax}
  \BIBdecl
  
  \bibitem{FengSCE}
  Z.~Feng, Z.~Na, M.~Xiong \emph{et~al.}, ``Multi-collaborative wireless
    communication networks for single cell edge users,'' \emph{Mob. Netw. Appl.},
    vol.~27, pp. 1--15, 01 2022.
  
  \bibitem{5Gposition}
  S.~Dwivedi, R.~Shreevastav, F.~Munier \emph{et~al.}, ``Positioning in {5G}
    networks,'' \emph{{IEEE} Commun. Mag.}, vol.~59, no.~11, pp. 38--44, 2021.
  
  \bibitem{Radiotri}
  P.~Krapež and M.~Munih, ``Anchor calibration for real-time-measurement
    localization systems,'' \emph{{IEEE} Trans. Instrum. Meas.}, vol.~69, no.~12,
    pp. 9907--9917, 2020.
  
  \bibitem{RIS3d}
  J.~He, A.~Fakhreddine, C.~Vanwynsberghe \emph{et~al.}, ``{3D} localization with
    a single partially-connected receiving {RIS}: {Positioning} error analysis
    and algorithmic design,'' \emph{{IEEE} Trans. Veh. Technol.}, vol.~72,
    no.~10, pp. 13\,190--13\,202, 2023.
  
  \bibitem{siglenode3d}
  J.~S. Russell, M.~Ye, B.~D.~O. Anderson \emph{et~al.}, ``Cooperative
    localization of a {GPS}-denied {UAV} using direction-of-arrival
    measurements,'' \emph{{IEEE} Trans. Aerosp. Electron. Syst.}, vol.~56, no.~3,
    pp. 1966--1978, 2020.
  
  \bibitem{5sensors}
  C.~Xu, Z.~Wang, Y.~Wang \emph{et~al.}, ``Three passive {TDOA-AOA}
    receivers-based flying-{UAV} positioning in extreme environments,''
    \emph{{IEEE} Sensors J.}, vol.~20, no.~16, pp. 9589--9595, 2020.
  
  \bibitem{3gpp_uas_support}
  \BIBentryALTinterwordspacing
  {3rd Generation Partnership Project (3GPP)}, ``Unmanned aerial system (UAS)
    support in {3GPP},'' 3GPP, {3GPP TS} 22.125, 2022. Available:
    \url{https://www.etsi.org/deliver/etsi_ts/122100_122199/122125/17.06.00_60/ts_122125v170600p.pdf}
  \BIBentrySTDinterwordspacing
  \revise{}{\bibitem{cooploc2024Hu}
  H.~Hu, Y.~Chen, B.~Peng \emph{et~al.}, ``Cooperative positioning of UAVs
    internet of things based on optimization algorithm,'' \emph{Wireless
    Networks}, vol.~30, no.~5, pp. 4495--4505, July 2024.}
  \revise{}{\bibitem{Cooploc2024jin}
  R.~Jin, G.~Zhang, L.-T. Hsu \emph{et~al.}, ``A survey on cooperative
    positioning using GNSS measurements,'' \emph{IEEE Trans. Intel.
    Vehi.}, pp. 1--20, 2024.}
  \revise{}{\bibitem{CP2022minetto}
  A.~Minetto, M.~C. Bello, and F.~Dovis, ``DGNSS cooperative positioning in
    mobile smart devices: A proof of concept,'' \emph{IEEE Trans. 
    Veh. Technol.}, vol.~71, no.~4, pp. 3480--3494, 2022.}
  \revise{}{\bibitem{huang2015distributed}
   B. Huang, Z. Yao, X. Cui \emph{et~al.}, “Distributed GNSS Collaborative Localization: Performance Analysis and Simulation,” in \emph{Proc. ION GNSS+}, 2015, pp. 2444–2454.}
  
  \bibitem{Ragmalicous}
  B.~Mukhopadhyay, S.~Srirangarajan, and S.~Kar, ``{RSS}-based localization in
    the presence of malicious nodes in sensor networks,'' \emph{{IEEE} Trans.
    Instrum. Meas.}, vol.~70, pp. 1--16, 2021.
  
  \bibitem{Jhagame}
  S.~Jha, S.~Tripakis, S.~Seshia \emph{et~al.}, ``Game theoretic secure
    localization in wireless sensor networks,'' \emph{2014 Interna.
    Conf. Inter. of Things}, pp. 85--90, 10 2014.
  \revise{}{\bibitem{DOCwon2019}
  J.~Won and E.~Bertino, ``Robust sensor localization against known sensor
    position attacks,'' \emph{IEEE Transa. on Mob. Compu.}, vol.~18,
    no.~12, pp. 2954--2967, 2019.}
  \revise{}{\bibitem{Tomic2022DetectingDA}
  S.~Tomic and M.~Beko, ``Detecting distance-spoofing attacks in
    arbitrarily-deployed wireless networks,'' \emph{{IEEE} Trans. Veh. Technol.},
    vol.~71, pp. 4383--4395, 2022.}
  \revise{}{\bibitem{scp2021beko}
  M.~Beko and S.~Tomic, ``Toward secure localization in randomly deployed
    wireless networks,'' \emph{IEEE Internet of Things J.}, vol.~8, no.~24,
    pp. 17\,436--17\,448, 2021.}
  \revise{}{\bibitem{LighadXie2021}
  N.~Xie, Y.~Chen, Z.~Li \emph{et~al.}, ``Lightweight secure localization
    approach in wireless sensor networks,'' \emph{IEEE Trans. on
    Commu.}, vol.~69, no.~10, pp. 6879--6893, 2021.}
  
  
  \bibitem{FHS2023reliable}
  Z.~Fang, B.~Han, and H.~D. Schotten, ``A reliable and resilient framework for
    multi-{UAV} mutual localization,'' in \emph{IEEE Conf. Veh. Techno. Fall}, 2023, pp.
    1--7.
  
  \bibitem{sensor2d}
  \BIBentryALTinterwordspacing
  W.~Li, B.~Jelfs, A.~Kealy \emph{et~al.}, ``Cooperative localization using
    distance measurements for mobile nodes,'' \emph{Sensors}, vol.~21, no.~4,
    2021.
  \BIBentrySTDinterwordspacing
  
  \bibitem{staticsensor1}
  V.~Bianchi, P.~Ciampolini, and I.~De~Munari, ``{RSSI}-based indoor localization
    and identification for {ZigBee} wireless sensor networks in smart homes,''
    \emph{{IEEE} Trans. Instrum. Meas.}, vol.~68, no.~2, pp. 566--575, 2019.
  
  \bibitem{staticsensor2}
  X.~Mei, H.~Wu, J.~Xian \emph{et~al.}, ``{RSS}-based byzantine fault-tolerant
    localization algorithm under {NLOS} environment,'' \emph{{IEEE} Commun.
    Lett.}, vol.~25, no.~2, pp. 474--478, 2021.
  
  \bibitem{drones6020028}
  \BIBentryALTinterwordspacing
  E.~Çetin, A.~Cano, R.~Deransy \emph{et~al.}, ``Implementing mitigations for
    improving societal acceptance of urban air mobility,'' \emph{Drones}, vol.~6,
    no.~2, 2022. 
  \BIBentrySTDinterwordspacing
  
  \bibitem{3dspacelocal}
  V.~Sneha and N.~Munusamy, ``Localization in wireless sensor networks: {A}
    review,'' \emph{Cybern. Inf. Technol.}, vol.~20, pp. 3--26, 11 2020.
  \revise{}{\bibitem{Kumari2019}
  J.~Kumari, P.~Kumar, and S.~K. Singh, ``Localization in three-dimensional
    wireless sensor networks: A survey,'' \emph{The J. of Supercomputing},
    vol.~75, no.~8, pp. 5040--5083, Aug. 2019.}
  \revise{}{\bibitem{3duavloc2013vill}
  L.~A. Villas, D.~L. Guidoni, and J.~Ueyama, ``3D localization in wireless
    sensor networks using unmanned aerial vehicle,'' in \emph{2013 IEEE 12th
  Sympo. Network Compu. and Appli.}, 2013, pp.
    135--142.}
  
  \bibitem{MSK2021rss}
  B.~Mukhopadhyay, S.~Srirangarajan, and S.~Kar, ``{RSS}-based localization in
    the presence of malicious nodes in sensor networks,'' \emph{{IEEE} Trans.
    Instrum. Meas.}, vol.~70, pp. 1--16, 2021.
 \revise{}{\bibitem{MLXu2021}
  B.~Xu, S.~Li, A.~A.~Razzaqi, ``A novel measurement information anomaly
    detection method for cooperative localization,'' \emph{IEEE Transactions on
    Instrumentation and Measurement}, vol.~70, pp.~1--18, 2021.}
  \bibitem{rssidisad}
  D.~Giovanelli and E.~Farella, ``{RSSI} or time-of-flight for {Bluetooth Low
    Energy} based localization? {A}n experimental evaluation,'' in \emph{2018
    11th IFIP WMNC}, 2018, pp. 1--8.
  
  \bibitem{shortrangenovel}
  H.~Wang, J.~Wan, and R.~Liu, ``A novel ranging method based on {RSSI},''
    \emph{Energy Procedia}, vol.~12, pp. 230--235, 12 2011.
  
  \bibitem{shortrangeoutdoor}
  A.~Fakhri, S.~Gharghan, and S.~Mohammed, ``Path-loss modelling for {WSN}
    deployment in indoor and outdoor environments for medical applications,''
    \emph{Int. J. Eng.}, vol.~7, pp. 1666--1671, 08 2018.
  
  \bibitem{longseparation}
  J.~Allred, A.~Hasan, S.~Panichsakul \emph{et~al.}, ``{SensorFlock}: {An}
    airborne wireless sensor network of micro-air vehicles,'' in \emph{Proc. 5th
    Int. Conf. Embed. Netw. Sens. Syst}, 11 2007, pp. 117--129.
  
  \bibitem{app13032008}
  S.~Srinivas, S.~Welker, A.~Herschfelt \emph{et~al.}, ``Cramer Rao lower bounds
    on {3D} position and orientation estimation in distributed ranging systems,''
    \emph{Appl. Sci.}, vol.~13, no.~3, 2023.
  
  \bibitem{GVW2012efficient}
  R.~Garg, A.~L. Varna, and M.~Wu, ``An efficient gradient descent approach to
    secure localization in resource constrained wireless sensor networks,''
    \emph{{IEEE} Trans. Inf. Forensics Security}, vol.~7, no.~2, pp. 717--730,
    2012.
  \revise{}{\bibitem{FISTABeck2009}
  A.~Beck and M.~Teboulle, ``A fast iterative shrinkage-thresholding algorithm
    with application to wavelet-based image deblurring,'' in \emph{2009 IEEE
    Int. Conf. Acou., Speech Signal Proce.}, 2009,
    pp. 693--696.}
  \revise{}{\bibitem{LPlocXukun2016}
  K.~Xu, H.-l. Liu, D.~Liu \emph{et~al.}, ``Linear programming algorithms for
    sensor networks node localization,'' in \emph{2016 IEEE Int.
    Conf. on Consumer Electronics }, 2016, pp. 536--537.}
  \revise{}{\bibitem{admmHe2021}
  C.~He, Y.~Yuan, and B.~Tan, ``Alternating direction method of multipliers for
    {TOA-based} positioning under mixed sparse {LOS/NLOS} environments,'' \emph{IEEE
    Access}, vol.~9, pp. 28\,407--28\,412, 2021.}
  
  \bibitem{wu2018wngrad}
  X.~Wu, R.~Ward, and L.~Bottou, ``WNGrad: Learn the learning rate in gradient
    descent,'' \emph{arXiv preprint arXiv:1803.02865}, 2018.
  
  \bibitem{HKZ+2023trustawareness}
  B.~Han, D.~Krummacker, Q.~Zhou \emph{et~al.}, ``Trust-awareness to secure swarm
    intelligence from data injection attack,'' in \emph{IEEE Int. Conf. Communication}, 2023, pp.
    1406--1412.
  \revise{}{\bibitem{estimation1975Brad}
  B.~Efron, ``Biased versus unbiased estimation,'' \emph{Advances in
    Mathematics}, vol.~16, no.~3, pp. 259--277, 1975. }%
  \revise{}{\bibitem{traj2025nacar}
  O.~Nacar, M.~Abdelkader, L.~Ghouti \emph{et~al.}, ``Vector: Velocity-enhanced
    GRU neural network for real-time 3D UAV trajectory prediction,''
    \emph{Drones}, vol.~9, no.~1, 2025.}
  \revise{}{\bibitem{delmerico2019drone}
  J.~Delmerico, T.~Cieslewski, H.~Rebecq \emph{et~al.}, ``{Are We Ready for
    Autonomous Drone Racing? The UZH-FPV Drone Racing Dataset},'' in
    \emph{Int. Con. on Robotics and Automation
    }, 2019, pp. 6713--6719.}%
  \revise{}{\bibitem{fonder2019midair}
  M.~Fonder and M.~V. Droogenbroeck, ``{Mid-Air: A Multi-Modal Dataset for
    Extremely Low Altitude Drone Flights},'' in \emph{IEEE/CVF
    Conf. on Computer Vision and Pattern Recognition Workshops},
    2019, pp. 553--562.}

  \bibitem{han2024secure}
  B.~Han and H.~D. Schotten, ``A secure and robust approach for distance-based
    mutual positioning of unmanned aerial vehicles,'' in \emph{2024 IEEE WCNC},
    2024, pp. 1--6.
  
  \end{thebibliography}
